\documentclass[USenglish,onecolumn]{article}

\usepackage[utf8]{inputenc}
\usepackage[big]{dgruyter}
\usepackage{natbib}

\usepackage{graphicx}
\usepackage{float}
\usepackage{mwe}
\usepackage[export]{adjustbox}
\usepackage{amssymb}
\usepackage{bbm}
\usepackage{bm}
\usepackage{amssymb}
\usepackage{amsmath}
\usepackage{amsfonts}
\usepackage[shortlabels]{enumitem}
\usepackage{setspace}
\usepackage{booktabs}

\usepackage{accents}

\usepackage{multirow}
\usepackage{multicol}
\usepackage{mathrsfs}
 \usepackage{tabularray}
 \usepackage{mathtools}
 \usepackage{array}
\usepackage{subcaption}

\newcommand{\Var}{{\text{Var}}}
\newcommand{\Cov}{{\text{Cov}}}

\newcommand{\E}{{\mathbb{E}}}
\renewcommand{\Pr}{{\text{P}}}

\newtheorem{theorem}{Theorem}[section]
\newtheorem{corollary}{Corollary}[section]

\newtheorem{lemma}{Lemma}
\newtheorem{proposition}[theorem]{Proposition}

\newtheorem{assumption}{Assumption}

\usepackage{color}

\begin{document}

  \articletype{Research Article{\hfill}Open Access}

  \author*[1]{Ambarish Chattopadhyay}

\author[2]{Guido W. Imbens}

  \affil[1]{Applied Statistics Unit, Statistical Sciences Division, Indian Statistical Institute, Kolkata, West Bengal, 700108, India; E-mail: ambarish@isical.ac.in}

  \affil[2]{Graduate School of Business and Department of Economics, Stanford University, SIEPR, and
NBER, Stanford, California 94305, U.S.A; E-mail: imbens@stanford.edu}

  \title{\huge Neymanian inference in randomized experiments}

  \runningtitle{Neymanian inference}


\begin{abstract}
{In his seminal 1923 work, Neyman studied the variance estimation problem for the difference-in-means estimator of the average treatment effect in completely randomized experiments. He proposed a variance estimator that is conservative in general and unbiased under homogeneous treatment effects. While widely used under complete randomization, there is no unique or natural way to extend this estimator to more complex designs. To this end, we show that Neyman's estimator can be alternatively derived in two ways, leading to two novel variance estimation approaches: the imputation approach and the contrast approach. While both approaches recover Neyman's estimator under complete randomization, they yield fundamentally different variance estimators for more general designs. In the imputation approach, the variance is expressed in terms of observed and missing potential outcomes and then estimated by imputing the missing potential outcomes, akin to Fisherian inference. In the contrast approach, the variance is expressed in terms of unobservable contrasts of potential outcomes and then estimated by exchanging each unobservable contrast with an observable contrast. We examine the properties of both approaches, showing that for a large class of designs, each produces non-negative, conservative variance estimators that are unbiased in finite samples or asymptotically under homogeneous treatment effects.}
\end{abstract}


  \journalname{Journal of Causal Inference}
\DOI{DOI}
  \startpage{1}
  \received{..}
  \revised{..}
  \accepted{..}

  \journalyear{2026}
  \journalvolume{14(1)}

\maketitle

\section{Introduction}
\label{sec_intro}

\subsection{Design-based inference in randomized experiments}
\label{sec_designbased}

In randomized experiments, the act of randomization is controlled by the investigator and provides a basis to quantify uncertainty in a transparent manner. Thus, quite evidently, there has been an ever-growing interest in developing design-based (or randomization-based) approaches to causal inference in randomized experiments. 
There are two main modes of design-based causal inference: Fisherian and Neymanian (see, e.g., \cite{dasgupta2015causal,ding2017paradox}). Fisherian inference focuses on the causal effect of treatment on individual units within a population \cite{fisher1935design}, whereas Neymanian inference focuses on the average causal effect of treatment across a population \cite{neyman1923application}. In this paper, we explore the Neymanian mode of inference, focusing on obtaining unbiased estimators of the average treatment effect and developing estimators of their variances.

In his seminal work in 1923, Neyman
used the difference-in-means statistic to unbiasedly estimate the average treatment effect in completely randomized experiments \cite{neyman1923application}.
However, estimating its variance unbiasedly posed a challenge, since the variance depends on joint distributions of the potential outcomes under treatment and control for each unit, of which only one is observable.
To this end, Neyman proposed an unbiased estimator for an upper bound to the true variance such that, the estimator is unbiased for the true variance under treatment effect homogeneity, i.e., when treatment effects are the same for every unit. 
While Neyman's estimator is widely used under complete randomization, there is no unique or natural way to extend this estimator to more general experimental designs.
Specifically, given a general design and an unbiased estimator $\hat{\tau}$ of the average treatment effect (in particular, the Horvitz-Thompson estimator), it is not always straightforward to devise a principled approach that yields variance estimators which are generally conservative, unbiased (or nearly unbiased) under treatment effect homogeneity, and reduce to Neyman’s estimator in the case of complete randomization.

\subsection{Contribution and outline}
\label{sec_contribution}

To address the above problem, in this paper, we show that Neyman's variance estimator for completely randomized experiments can be derived using two alternative approaches, each of which is generalizable to more complex designs, where the average treatment effect is estimated unbiasedly using the Horvitz-Thompson estimator. We term them the \textit{imputation} approach and the \textit{contrast} approach. 
While both approaches recover Neyman's estimator under complete randomization, they yield fundamentally different variance estimators for more general designs.  

In the imputation approach, the variance of the estimated treatment effect is first expressed as a function of observed and missing potential outcomes across all the units. The variance is then estimated simply by estimating or imputing the missing potential outcome \textcolor{black}{(or a function of the observed and missing potential outcome)} of each unit, similar to Fisherian inference. This approach is easy to implement and does not require the knowledge of the assignment mechanism in closed-form.

\textcolor{black}{In the canonical version of the imputation approach, where} the \textcolor{black}{missing} potential outcome \textcolor{black}{for unit $i$} is imputed by setting the unit-level effect for unit $i$ to a deterministic value $\beta_i$, the resulting variance estimator is shown to be conservative in finite samples for a large class of designs \textcolor{black}{(e.g., all symmetric designs)}, irrespective of the choice of $\beta_i$. 
In addition, for some designs, the variance estimator is asymptotically unbiased under effect homogeneity, even when the true effects differ from the imputed effects.
This approach is also shown to produce reasonable variance estimators for a large class of designs if we set the unit-level effects to be the estimated average treatment effect $\hat{\tau}$. In particular, for completely randomized designs with equal group sizes, this choice results in a variance estimator that is asymptotically equivalent to the standard Neymanian variance estimator, implying that the Neymanian variance estimator can be obtained through a Fisherian mode of inference. Beyond complete randomization, we show that the imputation approach based on $\hat{\tau}$ is asymptotically unbiased for the true variance under mild conditions on the design and the potential outcomes.

Finally, by introducing the notion of \textcolor{black}{``direct imputation'', we extend the imputation approach to arbitrary experimental designs}. Instead of estimating the missing potential outcome for each unit, the direct imputation approach aims to estimate a function of its observed and missing potential outcomes that directly relates to the variance. We show that this direct imputation approach yields a class of conservative variance estimators and, by leveraging a Jackknife-based method, provide a practical recommendation for a direct imputation-based variance estimator applicable to arbitrary designs. 

In the contrast approach, the variance is first expressed as a function of several unobservable contrasts of potential outcomes, where each contrast corresponds to an assignment vector in the support of the design. Here, a contrast is a linear combination of the treated and control potential outcomes across all units, where the coefficients sum to zero.
The variance is estimated by exchanging each unobservable contrast with an observable (and hence, estimable) contrast, or averages of observable contrasts. Unlike the imputation approach, the contrast approach does not focus on separately estimating the missing potential outcome for each unit, rather it focuses on directly estimating contrasts of potential outcomes across all units. We analyze the finite sample properties of this approach and show that, for a class of designs, the contrast approach yields a variance estimator that is conservative in general and unbiased under homogeneity. 


The approaches discussed in this paper contribute to the expanding literature on design-based inference in randomized experiments. A number of papers have focused on developing design-based variance estimators for completely randomized designs \cite{robins1988confidence, aronow2014sharp, nutz2022directional}, stratified and paired randomized designs \cite{kempthorne1955randomization, wilk1955randomization, gadbury2001randomization, abadie2008estimation, imai2008variance, higgins2015blocking, fogarty2018mitigating, pashley2021insights, luo2023interval}. Our work adds to this literature by focusing on principles that target a more general class of designs. 

In this regard, recent works develop methods to conduct design-based inference for arbitrary experimental designs \cite{aronow2013class, aronow2013conservative, harshaw2021optimized}. However, these important contributions typically focus on one facet of Neymanian inference, namely conservativeness, without directly addressing the other, i.e., unbiasedness under homogeneity.
Our emphasis on tailoring the variance estimator towards homogeneity stems from two reasons. First, a variance estimator that is unbiased under homogeneity seems more interpretable than one that is unbiased under an arbitrary restriction on the potential outcomes. Second, for a class of `measurable' designs (see Section \ref{sec_measurable}), the variance estimator that is unbiased under homogeneity is minimax (among a class of estimators) in that it minimizes the worst-case bias in estimating the true variance \cite{mukerjee2018using}. A closely related work is that of Mukerjee et al. \cite{mukerjee2018using}, which introduces a class of variance decompositions encompassing Neyman's variance decomposition under complete randomization as a special case. Despite its generality, identifying an appropriate decomposition for complex designs can often be challenging in practice. We show that, for some designs, the contrast approach recovers the decomposition-based variance estimators, thereby providing a concrete choice of the decomposition. Moreover, the aforementioned works typically employ Horvitz-Thompson-type estimators for variance estimation, which carries the risk of producing negative variance estimates \cite{lohr2021sampling}. In contrast, the imputation and contrast approaches guarantee non-negative variance estimates.

Finally, our work also contributes to the recent literature on synthesizing Neymanian and Fisherian modes of inference in randomized experiments \cite{ding2017paradox, ding2018randomization,  wu2021randomization}. These existing works pertain to hypothesis tests on average treatment effects and operate under completely randomized, stratified randomized, or factorial designs. In contrast, the imputation approach pertains to the estimation of average treatment effects and applies to a more general class of designs. Moreover, while existing results connecting the two modes of inference are valid in large samples, most of the results concerning the imputation approach hold in finite samples.

The paper is structured as follows. In Section \ref{sec_setup}, we present the experimental design setup and notations. In Section \ref{sec_measurable}, we formalize the notion of Neymanian inference and review the Neymanian decomposition approach to variance estimation. In Sections \ref{sec_substitution} and \ref{sec_imputation}, we formally propose and analyze the contrast approach and the imputation approach to Neymanian inference, respectively. In Section \ref{sec_final}, we conclude with a summary and remarks.

\section{Setup, notations, and the estimation problem}
\label{sec_setup}

Consider a randomized experiment conducted on a finite population of $N$ units indexed by $i = 1,2,...,N$. Under the potential outcomes framework \cite{neyman1923application, rubin1974estimating}, let $Y_i(0)$ and $Y_i(1)$ be the potential outcomes for unit $i$ under control and treatment, respectively. In these notations, we assume that the stable unit treatment value assumption (SUTVA; \cite{rubin1980randomization}) holds, i.e., there is no interference across units and no different versions of each treatment level that may lead to different potential outcomes. Throughout the paper, the causal estimand of interest is the average treatment effect, $\tau = (1/N)\sum_{i=1}^{N}\{Y_i(1) - Y_i(0)\}.$
For unit $i$, let $W_i \in \{0,1\}$ be the treatment indicator, i.e., $W_i = 1$ if unit $i$ receives treatment and $W_i = 0$ otherwise. The observed outcome for unit $i$ is thus given by $Y^{\text{obs}}_i  =W_iY_i(1) + (1-W_i)Y_i(0)$. 

In this paper, we adopt a design-based (or randomization-based) perspective, where the potential outcomes are considered fixed quantities and randomness stems solely from the assignment of treatments to units.
Denote a generic experimental design by $d$. The corresponding assignment mechanism is defined as the joint probability distribution of $\bm{W} = (W_1,...,W_N)^\top$ under $d$, i.e., for $\bm{w} \in \{0,1\}^N$, $p_{\bm{w}} = \Pr_d(\bm{W} = \bm{w})$.
Moreover, denote $\mathcal{W} = \{\bm{w}\in \{0,1\}^N: p_{\bm{w}}>0\}$ as the support of $d$, and $\pi_i = \Pr_d(W_i = 1)$ as the propensity score.
For instance, in a completely randomized design (CRD) with $N_t $ treated and $N_c$ control units, $\mathcal{W} = \{\bm{w}\in \{0,1\}: \sum_{i=1}^{N}w_i = N_t\}$, $p_{\bm{w}} = \mathbbm{1}(w\in \mathcal{W})/\binom{N}{N_t}$, and $\pi_i = N_t/N$ for all $i \in \{1,2,...,N\}$.   

For estimating $\tau$ unbiasedly under a general design $d$, we require the following positivity assumption, which states that each unit has a positive probability of receiving either treatment or control.
\begin{assumption}[Positivity] \normalfont
    For design $d$, $0<\pi_i<1$ for all $i \in \{1,2,...,N\}$.
    \label{assump_positivity}
\end{assumption}
Under Assumption \ref{assump_positivity}, we can unbiasedly estimate the average treatment effect $\tau$ using the Horvitz-Thompson (or the inverse probability weighting) estimator $\hat{\tau} = (1/N)\sum_{i=1}^{N}W_i Y^{\text{obs}}_i/\pi_i - (1/N)\sum_{i=1}^{N}(1-W_i) Y^{\text{obs}}_i/(1-\pi_i).$
For a CRD, the Horvitz-Thompson estimator $\hat{\tau}$ boils down to the standard difference-in-means estimator, given by $\hat{\tau} = (1/N_t)\sum_{i:W_i = 1}Y^{\text{obs}}_i -  (1/N_c)\sum_{i:W_i = 0}Y^{\text{obs}}_i.$
Unless otherwise specified, throughout the rest of the paper, we assume that Assumption \ref{assump_positivity} holds.  
\section{Formal problem and Neymanian decomposition}
\label{sec_measurable}
\label{sec_neyman_decomposition}
Our focus in this paper is on estimating the design-based variance of $\hat{\tau}$ in finite samples. 
For a CRD, Neyman \cite{neyman1923application} (see also \cite{imbens2015causal}, Chapter 6) proposed a conservative variance estimator that is unbiased under treatment effect homogeneity (or simply, homogeneity), i.e., when all unit-level treatment effects $Y_i(1)-Y_i(0)$ are constant. 
More formally, Neyman showed that, for a CRD,
\begin{align}
    \Var(\hat{\tau}) &= \frac{S^2_1}{N_t}+ \frac{S^2_0}{N_c}  - \frac{\textcolor{black}{S^2_{\tau}}}{N}, \label{eq_neyman}
\end{align}
where $\bar{Y}(1) = (1/N)\sum_{i=1}^{N}Y_i(1)$, $\bar{Y}(0) = (1/N)\sum_{i=1}^{N}Y_i(0)$, $S^2_1 = \frac{1}{N-1}\sum_{i=1}^{N}\{Y_i(1) -\bar{Y}(1)\}^2 $, $S^2_0 = \frac{1}{N-1}\sum_{i=1}^{N}\{Y_i(0) -\bar{Y}(0)\}^2 $, and $\textcolor{black}{S^2_{\tau}} = \frac{1}{N-1}\sum_{i=1}^{N}[\{Y_i(1)-Y_i(0)\} -\{\bar{Y}(1) - \bar{Y}(0)\}]^2$.  The variance decomposition in Equation \ref{eq_neyman} is called the Neymanian decomposition \cite{mukerjee2018using}.
While the first two terms in Equation \ref{eq_neyman} are unbiasedly estimable, the third term, $-\textcolor{black}{S^2_{\tau}}/N$, is not identifiable due to the fundamental problem of causal inference \cite{holland1986statistics}. However, since this term is always non-positive, estimating the first two terms unbiasedly would guarantee that the variance estimator is conservative. Hence, one can obtain the following conservative variance estimator (called the Neymanian estimator),
\begin{align}
    \hat{V}_{\text{Neyman}} = \frac{1}{N_t(N_t-1)}\sum_{i:W_i=1}(Y^{\text{obs}}_i -\bar{Y}_t)^2 + \frac{1}{N_c(N_c-1)}\sum_{i:W_i = 0}(Y^{\text{obs}}_i -\bar{Y}_c)^2,
\end{align}
where $\bar{Y}_t$ and $\bar{Y}_c$ are the means of the observed outcomes in the treatment and control groups, respectively. It follows that, $\E(\hat{V}_{\text{Neyman}}) \geq \Var(\hat{\tau})$. Moreover, when $Y_i(1) - Y_i(0) = \tau$ for all $i$, $\E(\hat{V}_{\text{Neyman}}) = \Var(\hat{\tau})$.

In this paper, we aim to conduct Neymanian inference for a more general class of experimental designs. More formally, for design $d$, we want to obtain an estimator $\hat{V}_d$ such that $\E(\hat{V}_d) \geq \Var_d(\hat{\tau})$ regardless of the potential outcomes. Moreover, when $Y_i(1) - Y_i(0) = \tau$, $\E(\hat{V}_d) = \Var_d(\hat{\tau})$ (or $\E(\hat{V}_d) \approx \Var_d(\hat{\tau})$).


An instinctive way to address the Neymanian inference problem for a general design $d$ is to decompose the variance of $\hat{\tau}$ under $d$ akin to that in Equation \ref{eq_neyman}, i.e., to write the variance as the sum of a component that is potentially estimable and another component that is not estimable in general, but is non-positive and vanishes under treatment effect homogeneity.
We call this approach the Neymanian decomposition approach. 
In this regard, Mukerjee et al. \cite{mukerjee2018using} provides a general class of Neymanian decompositions, which apply to linear unbiased estimators of finite population-level treatment contrasts in multivalued treatment settings. Proposition \ref{prop_decomp2} presents a special case of these decompositions for our current problem of average treatment effect estimation with binary treatment.
\begin{proposition} \normalfont
Let $\bm{Q}$ be an $N\times N$ matrix with $q_{ii'}$ as its $(i,i')$th element. Assume that (i) $\bm{Q}$ is non-negative definite, (ii) $q_{ii} = 1/N^2$ for all $i \in \{1,2,...,N\}$, and
(iii) $\sum_{j=1}^{N}q_{ij} = 0$ for all $i \in \{1,2,...,N\}$. \textcolor{black}{Moreover, for $i,i' \in \{1,2,...,N\}$ and $w,w' \in \{0,1\}$, let $p_{ii'}(w,w') :=\Pr_d(W_i = w, W_{i'} = w')$ be the probability that units $(i,i')$ receive treatment combination $(w,w')$.}
Then, for an arbitrary design $d$, 
\begin{equation}
\Var_d(\hat{\tau}) = \Tilde{V}_d(\bm{Q}) - \{\bm{Y}(1) - \bm{Y}(0)\}^\top \bm{Q} \{\bm{Y}(1) - \bm{Y}(0)\}, \label{eq_neyman_gen2}    
\end{equation}
where $\bm{Y}(1) = (Y_1(1),...,Y_N(1))^\top$, $\bm{Y}(0) = (Y_1(0),...,Y_N(0))^\top$, and
\begin{align}
&\Tilde{V}_d(\bm{Q}) \nonumber \\
& = \frac{1}{N^2}\left\{\sum_{i=1}^{N}\frac{Y^2_i(1)}{\pi_i} + \sum_{i=1}^{N} \frac{Y^2_i(0) }{(1-\pi_i)} \right\} + 2\mathop{\sum\sum}_{i<i'}\left[Y_i(1)Y_{i'}(1)\left\{\frac{p_{ii'}(1,1)}{N^2\pi_i \pi_{i'}} + q_{ii'} - \frac{1}{N^2}\right\} \right. \nonumber\\
& \hspace{2.2in} \left. + Y_i(0)Y_{i'}(0)\left\{\frac{p_{ii'}(0,0)}{N^2(1-\pi_i) (1-\pi_{i'})}  + q_{ii'}- \frac{1}{N^2} \right\} \right] \nonumber \\
& \quad - 2\mathop{\sum\sum}_{i<i'}\left[Y_i(1)Y_{i'}(0)\left\{\frac{p_{ii'}(1,0)}{N^2\pi_i (1- \pi_{i'})} + q_{ii'} - \frac{1}{N^2}\right\} + Y_i(0)Y_{i'}(1)\left\{\frac{p_{ii'}(0,1)}{N^2(1-\pi_i) \pi_{i'}} + q_{ii'} - \frac{1}{N^2} \right\} \right]. \nonumber
\end{align}
    \label{prop_decomp2}
\end{proposition}
The proof of Proposition follows from that of Theorem 2 in \cite{mukerjee2018using}. See the Appendix for an alternative proof. Under a CRD, $\Tilde{V}_d(\bm{Q})$ boils down to the standard Neymanian decomposition in Equation \ref{eq_neyman} for $\bm{Q} = (\bm{I} - (1/N)\bm{J})/\{N(N-1)\}$, where $\bm{I}$ is the identity matrix of order $N$ and $\bm{J}$ is the $N \times N$ matrix of all 1's. By construction of $\bm{Q}$, $\{\bm{Y}(1) - \bm{Y}(0)\}^\top \bm{Q} \{\bm{Y}(1) - \bm{Y}(0)\} \geq 0$, and under treatment effect homogeneity, $\{\bm{Y}(1) - \bm{Y}(0)\}^\top\bm{Q}\{\bm{Y}(1) - \bm{Y}(0)\} = \bm{0}$. Hence, an unbiased estimator of $\Tilde{V}_d(\bm{Q})$ is conservative for $\Var_d(\hat{\tau})$ in general, and unbiased under homogeneity (in fact, for some choices of $\bm{Q}$, it is unbiased under a weaker condition than homogeneity). We call such an estimator a Neymanian decomposition-based estimator.

Now, regardless of $\bm{Q}$, $\Tilde{V}_d(\bm{Q})$ can be estimated unbiasedly using a Horvitz-Thompson-type estimator if all the pairwise probabilities of treatment assignments are strictly positive. Akin to survey sampling \cite{kish1965survey}, we call this design condition the measurability condition and the corresponding design a measurable design. More formally, a design $d$ is called \textit{measurable}, if for all $i,i' \in \{1,2,...,N\}$ such that $i \neq i'$, and for $w,w' \in \{0,1\}$, \textcolor{black}{$p_{ii'}(w,w')>0$.} 

However, for non-measurable designs, $\Tilde{V}_d(\bm{Q})$ is not estimable for all $\bm{Q}$, and a judicious choice of $\bm{Q}$ is needed to ensure estimability. For instance, suppose there exists $i\neq i'$ such that $\Pr_d(W_i = 1,W_{i'} = 1) = 0$. Then, from Proposition \ref{prop_decomp2}, it follows that $\bm{Q}$ must satisfy $q_{ii'} = 1/{N^2}$. In general, with non-measurable designs, ensuring the existence of a $\bm{Q}$ that fulfills all these conditions is not straightforward. Even if such a $\bm{Q}$ exists, constructing it may pose challenges.

With this consideration, in the following two sections, we present and analyze two alternative approaches to Neymanian inference, namely, the contrast approach and the imputation approach. We discuss the conditions under which the corresponding variance estimators are conservative and unbiased (or close to unbiased). We also discuss connections of these approaches to the Neymanian decomposition approach and the Neymanian estimator.

\section{The contrast approach}
\label{sec_substitution}
\subsection{Motivating idea}
\label{sec_motivating_substitution}

In this section, we illustrate the key idea of the contrast approach using a toy example.
Consider a randomized experiment with $N = 4$ units and a design $d$ that selects one of the four assignment vectors in the set $\mathcal{W} = \{(1,1,0,0)^\top, (0,0,1,1)^\top, (1,0,0,1)^\top, (0,1,1,0)^\top\}$ with probability $0.25$ each. 
For this design, the group sizes are equal, and each unit has a propensity score $\pi_i = 0.5$. Consequently, the Horvitz-Thompson estimator is algebraically equivalent to the difference-in-means estimator.
Moreover, this design is not measurable because $\Pr(W_1 = w, W_3 = w) = 0$ and $\Pr(W_2 = w, W_4 = w) = 0$, for $w \in \{0,1\}$. 

The contrast approach \textcolor{black}{is based on the observation that $\Var_d(\hat{\tau})$ can be expressed as a weighted sum of squared contrasts of potential outcomes, which, in this example, takes the form}

\begin{align}
 \Var_d(\hat{\tau}) & = \frac{1}{8}\left[\left\{\frac{Y_1(0) + Y_1(1)}{2} + \frac{Y_2(0) + Y_2(1)}{2} - \frac{Y_3(0) + Y_3(1)}{2}  - \frac{Y_4(0) + Y_4(1)}{2}\right\}^2 \right. \nonumber\\
&\hspace{0.5in} \left. + \left\{\frac{Y_1(0) + Y_1(1)}{2} - \frac{Y_2(0) + Y_2(1)}{2} - \frac{Y_3(0) + Y_3(1)}{2}  + \frac{Y_4(0) + Y_4(1)}{2} \right\}^2 \right]. \label{eq_example1_1}
\end{align}
To find an estimator of this variance, first, we consider the case where the treatment effect is homogeneous across units, i.e., $Y_i(1) - Y_i(0) = \tau$ for all $i$. In this case, the variance expression in Equation \ref{eq_example1_1} simplifies to $\Var_d(\hat{\tau}) = (1/8)[\{Y_1(0) + Y_2(0) - Y_3(0) - Y_4(0)\}^2 + \{Y_1(0) - Y_2(0) - Y_3(0) + Y_4(0)\}^2]$.
The first \textcolor{black}{squared contrast} can be written as \begin{align}
    \{Y_1(0) + Y_2(0) - Y_3(0) - Y_4(0)\}^2 &= \{(Y_1(0) + \tau) + Y_2(0) - Y_3(0) - (Y_4(0) + \tau)\}^2 \nonumber\\
    &= \{Y_1(1) + Y_2(0) - Y_3(0) - Y_4(1)\}^2, \label{eq_example1_3}
\end{align}
\textcolor{black}{which} is unbiasedly estimable \textcolor{black}{using the Horvitz-Thompson estimator} $\mathbbm{1}\{\bm{W} = (1,0,0,1)^\top
\}(Y^{\text{obs}}_1 + Y^{\text{obs}}_2 - Y^{\text{obs}}_3 - Y^{\text{obs}}_4)^2/0.25$. Similarly, \textcolor{black}{the second squared contrast can be written as} $\{Y_1(0) -Y_2(0) - Y_3(0) + Y_4(0)\}^2 = \{Y_1(1) - Y_2(1) - Y_3(0) + Y_4(0)\}^2$, \textcolor{black}{which is estimated unbiasedly by the Horvitz-Thompson estimator} $\mathbbm{1}\{\bm{W} = (1,1,0,0)^\top\}(Y^{\text{obs}}_1 - Y^{\text{obs}}_2 - Y^{\text{obs}}_3 + Y^{\text{obs}}_4)^2/0.25.$ Thus, under homogeneity, we can unbiasedly estimate the variance of $\hat{\tau}$, even though the design is not measurable. 

We call this approach the contrast approach, since here, a contrast of potential outcomes corresponding to an assignment vector is substituted by a contrast of observed outcomes under another assignment vector. For instance, in Equation \ref{eq_example1_3}, the contrast $\{Y_1(0) + Y_2(0) - Y_3(0) - Y_4(0)\}^2$ corresponds to the assignment vector $(1,1,0,0)^\top$ (or, equivalently, $(0,0,1,1)^\top$) in that this contrast is same (up to a proportionality constant) as $\{\hat{\tau}(\bm{w}) - \tau\}^2$ when $\bm{w} = (1,1,0,0)^\top$. In the contrast approach, we substitute this contrast by $\{Y_1(1)+ Y_2(0)-Y_3(0) - Y_4(1)\}^2$, which is same as $(Y^{\text{obs}}_1-Y^{\text{obs}}_2-Y^{\text{obs}}_3+Y^{\text{obs}}_4)^2$ if $\bm{W} = (1,0,0,1)^\top$. In this case, the assignment vector $(1,0,0,1)^\top$ acts as a \textcolor{black}{``substitute''} for the vector $(1,1,0,0)^\top$. \textcolor{black}{A formal definition of a substitute for an assignment vector is provided later in Section \ref{sec_substitution_general}.}

Now, \textcolor{black}{under the contrast approach, we can further exploit} the symmetry of the design to obtain a variance estimator that is \textcolor{black}{not only unbiased under homogeneity, but also conservative in general. The key idea is to use two complementary substitutes for each contrast.}
To illustrate, we first note that the assignment vector $(0,1,1,0)^\top$ also acts as a substitute for $(1,1,0,0)^\top$, since, under homogeneity, $\{Y_1(0) + Y_2(0) - Y_3(0) - Y_4(0)\}^2 = \{Y_1(0) + Y_2(1) - Y_3(1) - Y_4(0)\}^2$. In fact, $(0,1,1,0)^\top$ and $(1,0,0,1)^\top$ are the only two substitutes of $(1,1,0,0)^\top$. Combining the contrasts from these two substitutes, we can write $\{Y_1(0) + Y_2(0) - Y_3(0) - Y_4(0)\}^2 = (1/2)[\{Y_1(1) + Y_2(0) - Y_3(0) - Y_4(1)\}^2 + \{Y_1(0) + Y_2(1) - Y_3(1) - Y_4(0)\}^2]$, where the right-hand side is unbiasedly estimable \textcolor{black}{using the Horvitz-Thompson estimators discussed above}. Moreover, by Jensen's inequality, $[\{Y_1(0) + Y_1(1)\}/2 + \{Y_2(0) + Y_2(1)\}/{2} - \{Y_3(0) + Y_3(1)\}/{2}  - \{Y_4(0) + Y_4(1)\}/{2}]^2 \leq (1/2)\{(Y_1(1) + Y_2(0) - Y_3(0) - Y_4(1))^2 + (Y_1(0) + Y_2(1) - Y_3(1) - Y_4(0))^2\},$
 and this upper bound is attained under homogeneity. A similar argument applies to the second \textcolor{black}{squared contrast in} Equation \ref{eq_example1_1}. 
 Therefore, using the contrast approach, we obtain an estimator of $\Var_d(\hat{\tau})$ that is \textcolor{black}{both} conservative in general and unbiased under homogeneity.

\subsection{General formulation and properties}
\label{sec_substitution_general}

We now formalize the contrast approach to a more general class of experimental designs. To this end, we consider designs that assign units to two groups of equal size and have constant propensity scores across units.
\begin{assumption}[Equal sized groups with constant propensity score] \normalfont
    For design $d$, $\sum_{i=1}^{N}W_i = N/2$ and $\pi_i$ is constant across $i \in \{1,2,...,N\}$. \label{assump_equalsize}
\end{assumption}
Assumption \ref{assump_equalsize} implies that the propensity score is half for each unit and hence, the Horvitz-Thompson estimator is algebraically equivalent to the difference-in-means estimator.
This assumption holds for any design with fixed (non-random) group sizes that is symmetric with respect to the labeling of the groups, i.e., $p_{\bm{w}} = p_{\bm{1} - \bm{w}}$. Common examples include complete and stratified randomized designs with equal allocation, matched-pair designs, rerandomization with Mahalanobis distance (or any imbalance criteria) and equal allocation \cite{morgan2012rerandomization}. 
In Section \ref{sec_substitution_extension} in the Appendix, we discuss the contrast approach for a class of designs with unequal (and possibly random) group sizes.

The contrast approach relies on two key conditions on the design $d$. First, for every assignment vector $\bm{w} \in \mathcal{W}$, there should exist a substitute assignment vector $\tilde{\bm{w}} \in \mathcal{W}$. 
The presence of at least one substitute ensures that the resulting variance estimator is unbiased for $\Var_d(\hat{\tau})$ under homogeneity. Second, if $\tilde{w} \in \mathcal{W}$ is a substitute for $\bm{w}$, then $\bm{1} - \bm{w}$ should also be a substitute for $\bm{w}$. In the toy example, both $(1,0,0,1)^\top$ and $(0,1,1,0)^\top$ are substitutes for $(1,1,0,0)^\top$. The presence of this pair of substitutes ensures that the resulting estimator is conservative for $\Var_d(\hat{\tau})$ in general. Below we formalize these two conditions in Assumptions \ref{assump_substitution} and \ref{assump_closed}.
\begin{assumption}[Substitution condition] \normalfont
Fix a design $d$ with support $\mathcal{W}$. For $\bm{w} \in \mathcal{W}$, suppose units $\{i_1,...,i_{N/2}\}$ are assigned to treatment and units $\{j_1,...,j_{N/2}\}$ are assigned to control, where $\{i_1,...,i_{N/2}\} \cup \{j_1,...,j_{N/2}\} = \{1,...,N\}$. Then, there exists $\{i_{r_1},...i_{r_{N/4}}\} \subset \{i_1,...,i_{N/2}\}$ and $\{j_{s_1},...j_{s_{N/4}}\} \subset \{j_1,...,j_{N/2}\}$, such that under $d$, units $\{i_{r_1},...i_{r_{N/4}}\} \cup \{j_{s_1},...j_{s_{N/4}}\}$ are assigned to treatment with positive probability.
\label{assump_substitution}  
\end{assumption}
\textcolor{black}{Assumption~\ref{assump_substitution} holds for complete and paired randomized designs with equal and even group sizes (see Theorems~\ref{thm_substitute_neyman} and~\ref{thm_substitute_matched}), as well as for a class of stratified and clustered randomized designs with equal group sizes (see the Appendix for details).}
Assumption \ref{assump_substitution} also allows us to formalize the notion of a substitute. Fix $\bm{w}\in \mathcal{W}$ and let $\tilde{\bm{w}} = (\tilde{w}_1,...,\tilde{w}_N)^\top$ be a vector of assignments such that, $\tilde{w}_i = 1$ if $i \in \{i_{r_1},...i_{r_{N/4}}\} \cup \{j_{s_1},...j_{s_{N/4}}\}$ and $\tilde{w}_i = 0$ otherwise. $\tilde{\bm{w}}$ is called a substitute for $\bm{w}$. If Assumption \ref{assump_substitution} holds, we have $p_{\tilde{\bm{w}}}>0$, i.e., $\tilde{\bm{w}} \in \mathcal{W}$. Thus, it is possible to replace a contrast of potential outcomes corresponding to $\bm{w}$ by a contrast of observed outcomes under $\tilde{\bm{w}}$, as shown for the toy example in Section \ref{sec_motivating_substitution}. 

\begin{assumption}[Closed under label switching] \normalfont
For design $d$ with support $\mathcal{W}$, $\bm{w} \in \mathcal{W} \iff \bm{1} - \bm{w} \in \mathcal{W}$. \label{assump_closed}   
\end{assumption}
\textcolor{black}{Assumption~\ref{assump_closed} is satisfied by any design that remains invariant under relabeling of treatment and control groups. Examples include Bernoulli designs; complete, stratified, or clustered randomized designs with equal group sizes; and rerandomization based on Mahalanobis distance with equal-sized groups.  
Under Assumption~\ref{assump_closed}, if $\tilde{\bm{w}}$ is a substitute for $\bm{w}$, then $\bm{1} - \tilde{\bm{w}}$ is also a substitute for $\bm{w}$.}
Thus, under Assumptions \ref{assump_substitution} and \ref{assump_closed}, there exists more than one substitute for $\bm{w} \in \mathcal{W}$. We denote $\mathcal{G}(\bm{w})$ as an arbitrary set of substitutes of $\bm{w}$. $\mathcal{G}(\bm{w})$ is said to be closed under label switching, if $\tilde{\bm{w}} \in \mathcal{G}(\bm{w}) \iff \bm{1} - \tilde{\bm{w}} \in \mathcal{G}(\bm{w})$. 

Given a set of substitutes $\mathcal{G}(\bm{w})$ for $\bm{w} \in \mathcal{W}$, we obtain the following closed-form expression for the estimator of $\Var_d(\hat{\tau})$ under the contrast approach. 
\begin{equation}
\hat{V}_{\text{sub}} = \frac{4}{N^2}\sum_{\bm{w}:\bm{W} \in \mathcal{G}(\bm{w})} \frac{p_{\bm{w}}}{p_{\bm{W}}} \frac{1}{|\mathcal{G}(\bm{w})|}\{\bm{l}(\bm{w})^\top \bm{Y}^{\text{obs}}\}^2,
\label{eq_subsgen}
\end{equation}
where $\bm{l}(\bm{w}) = (l_1(\bm{w}),...,l_N(\bm{w}))^\top$ is such that $l_i(\bm{w}) = 1$ if $w_i = 1$ and $l_i(\bm{w}) = -1$ if $w_i = 0$.
\textcolor{black}{Note that $p_{\bm{w}}$ is the probability of realizing a specific assignment vector $\bm{w}$, whereas
$p_{\bm{W}}$ is the probability of realizing the observed assignment vector $\bm{W}$. 
The construction of $\hat{V}_{\text{sub}}$ directly extends the procedure outlined for the motivating example in Section \ref{sec_motivating_substitution}. In essence, rather than decomposing $\Var_d(\hat{\tau})$ into a quadratic form involving the squares and products of the potential outcomes, it is first expressed as a weighted sum of squared contrasts of potential outcomes (akin to Equation \ref{eq_example1_1}). These potentially unobservable contrasts are then substituted by observable contrasts of the form $\{\bm{l}(\bm{w})^\top \bm{Y}^{\text{obs}}\}$, together with an adjustment factor based on inverse probability weighting (or Horvitz-Thompson technique), resulting in the weights $\frac{1}{p_{\bm{W}}}$.}

Theorem \ref{thm_substitution} formalizes conditions under which the estimator is conservative in general and unbiased under homogeneity.   
\begin{theorem} \normalfont
Let $d$ be a design with support $\mathcal{W}$, satisfying Assumptions \ref{assump_equalsize} and \ref{assump_substitution}. Consider the estimator in Equation \ref{eq_subsgen}. Under treatment effect homogeneity, $\E_d(\hat{V}_{\text{sub}}) = \Var_d(\hat{\tau})$. Moreover, if Assumption \ref{assump_closed} holds for all $\bm{w}$, and $\mathcal{G}(\bm{w})$ is closed under label switching, then $\E_d(\hat{V}_{\text{sub}}) \geq \Var_d(\hat{\tau})$. 
Finally, by construction, $\hat{V}_{\text{sub}}$ is non-negative, thereby avoiding the negativity issues associated with standard Horvitz-Thompson-type variance estimators. 
\label{thm_substitution}
\end{theorem}
Theorem \ref{thm_substitution} shows that using the contrast approach, we can conduct Neymanian inference for a class of designs satisfying Assumptions \ref{assump_equalsize}, \ref{assump_substitution}, and \ref{assump_closed}. 
Note that the variance estimator $\hat{V}_{\text{sub}}$ is conservative in general (and unbiased under homogeneity) for any set of substitutes $\mathcal{G}(\bm{w})$ that is closed under label switch. Thus, there can be multiple choices of such $\mathcal{G}(\bm{w})$, and as a result, multiple variance estimators under the contrast approach. Also, if unbiasedness under homogeneity is the only requirement, then we no longer require $\mathcal{G}(\bm{w})$ to be closed under label switch, and hence, this class of variance estimators can be enlarged further. In this case, the estimators are valid under weaker restrictions on the design in that Assumption \ref{assump_closed} is no longer required.

Now, let $\mathcal{G}^*(\bm{w})$ be the set of all substitutes of $\bm{w}$. The corresponding variance estimator is likely to be more informative than other variance estimators of this type since it utilizes the most information from the design. Moreover, when $\mathcal{G}(\bm{w}) = \mathcal{G}^*(\bm{w})$ for all $\bm{w} \in \mathcal{W}$, the variance estimator can be further simplified as
$\hat{V}_{\text{sub}} = (4/N^2)\sum_{\bm{w} \in \mathcal{G}^*(\bm{W})} (p_{\bm{w}}/p_{\bm{W}})\{\bm{l}(\bm{w})^\top \bm{Y}^{\text{obs}}\}^2/|\mathcal{G}^*(\bm{w})|$,  
where the above equality holds because $\bm{W} \in \mathcal{G}^*(\bm{w}) \iff \bm{w} \in \mathcal{G}^*(\bm{W})$.

We conclude this section by focusing on the contrast approach for CRD and the toy example in Section \ref{sec_motivating_substitution}.
In a CRD, Assumption \ref{assump_closed} holds when the group sizes are equal ($N_t = N_c$), and Assumption \ref{assump_substitution} holds when $N$ is a multiple of four. Moreover, for every $\bm{w} \in \mathcal{W}$, there are multiple substitutes. To illustrate, for $\bm{w} = (\bm{1}^\top,\bm{0}^\top)$, any assignment vector that treats an arbitrary subset of $N/4$ units among the first $N/2$ units and another arbitrary subset of $N/4$ units among the last $N/2$ units is a substitute. Now, if we consider the full set of substitutes $\mathcal{G}^*(\bm{w})$ for each $\bm{w}$, then the resulting variance estimator is shown to be algebraically same as the Neymanian variance estimator. Theorem \ref{thm_substitute_neyman} formalizes this result. 
\begin{theorem} \normalfont
Consider a completely randomized design with support $\mathcal{W}$, and assume that $N = 4m$ for some positive integer $m$. Moreover, consider the variance estimator $\hat{V}_{\text{sub}}$ with the full set of substitutes, i.e., $\mathcal{G}(\bm{w}) = \mathcal{G}^*(\bm{w})$ for all $\bm{w} \in \mathcal{W}$. Then, $\hat{V}_{\text{sub}} = \hat{V}_{\text{Neyman}}.$ \label{thm_substitute_neyman}
\end{theorem}
Theorem \ref{thm_substitute_neyman} thus connects the contrast approach to the approach based on traditional Neymanian decomposition and provides an alternative interpretation of the Neymanian variance estimator $\hat{V}_{\text{Neyman}}$. Concretely, under a CRD with equal group sizes, the $\hat{V}_{\text{Neyman}}$ is the variance estimator that uses the largest class of substitutes. 

The contrast approach is also connected to the general Neymanian decomposition approach, as described in Proposition \ref{prop_decomp2}. To see this, we revisit the contrast approach for the toy example. The corresponding design satisfies Assumptions \ref{assump_equalsize}, \ref{assump_substitution}, and \ref{assump_closed}. By considering the full set of substitutes, the resulting variance estimator is shown to be algebraically equivalent to a Neymanian decomposition-based variance estimator. Proposition \ref{prop_subs_example} formalizes this result.
\begin{proposition} \normalfont
Let $d$ be the design in the toy example in Section \ref{sec_motivating_substitution}. The variance estimator under the contrast approach is the same as the Neymanian decomposition-based estimator with
$\bm{Q} = (1/16)\begin{psmallmatrix}
    1 & -1 & 1 & -1 \\
    -1 & 1 & -1 & -1 \\
    1 & -1 & 1 & -1 \\
    -1 & 1 & -1 & 1 \\
\end{psmallmatrix}.$
\label{prop_subs_example}
\end{proposition}
It is straightforward to verify that the matrix $\bm{Q}$ in Proposition \ref{prop_subs_example} satisfies all the required conditions in Proposition \ref{prop_decomp2}. Thus, in this example, the contrast approach provides an alternative way to obtain a Neymanian decomposition-based variance estimator by implicitly constructing a suitable $\bm{Q}$. In general, however, the contrast approach can yield estimators that fall outside the class of Neymanian decomposition-based variance estimators; see the Appendix for details.

Finally, the design in this example is equivalent to a matched-pair design, with units $(1,3)$ and $(2,4)$ forming the two pairs. It is straightforward to verify that, in this case, the variance estimator under the contrast approach coincides with the standard variance estimator for matched-pair experiments (see, e.g., \citealt{imbens2015causal}, Chapter 10). In fact, Theorem \ref{thm_substitute_matched} shows that this equivalence holds for more general matched-pair designs. In a matched-pair design, the standard estimator of $\Var_d(\hat{\tau})$ is given by  $\hat{V}_{\text{pair}} = \frac{4}{N(N-2)}\sum_{j=1}^{N/2}\{(Y^{\text{obs}}_{jt}-Y^{\text{obs}}_{jc})- (\bar{Y}_t - \bar{Y}_c)\}^2$, where $Y^{\text{obs}}_{jt}$ and $Y^{\text{obs}}_{jc}$ are observed outcomes of the treated and control units in pair $j$, respectively.
\begin{theorem} \normalfont
Consider a matched-pair design with $N = 4k$ for an integer $k\geq 1$. Moreover, consider the variance estimator $\hat{V}_{\text{sub}}$ with the full set of substitutes, i.e., $\mathcal{G}(\bm{w}) = \mathcal{G}^*(\bm{w})$ for all $\bm{w} \in \mathcal{W}$. Then, $\hat{V}_{\text{sub}} = \hat{V}_{\text{pair}}.$ 
\label{thm_substitute_matched}
\end{theorem}
Thus, similar to the Neymanian variance estimator $\hat{V}_{\text{Neyman}}$ in completely randomized designs, under matched-pair designs, $\hat{V}_{\text{pair}}$ can be regarded as a special case of variance estimators under the contrast approach.

\textcolor{black}{The contrast approach allows us to construct variance estimators with desirable properties for certain non-measurable designs. However, it does not extend to all such designs, for example, to more general forms of stratified or cluster randomized designs, because some underlying conditions, such as the substitution condition (Assumption~\ref{assump_substitution}), may be violated. To address this limitation, in Section~\ref{sec_imputation}, we formally propose and analyze the imputation approach to variance estimation.}


\section{The imputation approach}
\label{sec_imputation}

We begin this section by deriving a general closed-form expression of the variance of $\hat{\tau}$ in terms of the potential outcomes, as presented in Proposition \ref{prop_impute_gen}.
\begin{proposition} \normalfont
Let $c_i  = (1-\pi_i)Y_i(1) + \pi_i Y_i(0)$. For an arbitrary design $d$,
    \begin{align}
    \Var_d(\hat{\tau}) = \frac{1}{N^2}\sum_{\bm{w}} p_{\bm{w}} \left(\sum_{i:w_{i} = 1}\frac{c_{i}}{\pi_{i}} - \sum_{i:w_{i} = 0}\frac{c_{i}}{1-\pi_{i}} \right)^2 =: \textcolor{black}{\psi_d}(\bm{c}). 
\end{align}
\label{prop_impute_gen}
\end{proposition}
Proposition \ref{prop_impute_gen} shows that, the variance of $\hat{\tau}$ under $d$ depends on the potential outcomes through $\bm{c} = (c_1,...,c_N)^\top$. 

As a starting point for the imputation approach, we focus on a general class of designs that satisfy $\pi_i = 0.5$ for all $i$, i.e., the designs assign each unit to treatment or control with equal probability. Under this class, $c_i = \{Y_i(1)+Y_i(0)\}/2$, the simple average of the two potential outcomes. 
\textcolor{black}{This class of designs} is larger than the class accommodated by the contrast approach discussed in Section \ref{sec_substitution_general} since, among others, the substitution condition (Assumption \ref{assump_substitution}) is not required. 
We discuss the extensions of this approach to arbitrary designs in Section \ref{sec_imputation_extension}.

\subsection{\textcolor{black}{Basic} formulation and properties}
\label{sec_imputation_general}
\textcolor{black}{By the fundamental problem of causal inference \citep{holland1986statistics}, for each unit $i$, one of the two potential outcomes is unobserved/missing. The imputation approach is built on the idea of imputing these missing potential outcomes across units and using the observed and imputed potential outcomes to compute the variance of $\hat{\tau}$, in a manner similar to the Fisher randomization test.} In essence, it conducts Neymanian inference by leveraging techniques from Fisherian randomization-based inference. 

\textcolor{black}{Formally, the imputation approach proceeds as follows. First, the investigator specifies a vector $(\beta_1, \beta_2,..., \beta_N)^\top$, where $\beta_i$ represents the postulated unit-level causal effect for unit $i$. The choice of these $(\beta_1, \beta_2,..., \beta_N)^\top$ is discussed later in this section. Next, the potential outcomes for each unit $i$ are imputed as follows.} 
\begin{equation}
    \hat{Y}_i(1) = W_iY^{\text{obs}}_i + (1-W_i)(Y^{\text{obs}}_i + \beta_i), \hspace{0.5cm} \hat{Y}_i(0) = W_i (Y^{\text{obs}}_i - \beta_i) + (1-W_i) Y^{\text{obs}}_i. \label{eq_ai}
\end{equation}
In other words, the potential outcomes are imputed as if the unit-level treatment effect for unit $i$ is $\beta_i$. 
Note that, although we observe one potential outcome for unit $i$, we can still conceptualize imputing both the potential outcomes, where the observable potential outcome of unit $i$ is simply imputed by $Y^{\text{obs}}_i$.

After imputing these potential outcomes, $\Var_d(\hat{\tau})$ is estimated simply by plugging in the imputed potential outcomes in its expression, as provided in \textcolor{black}{Proposition \ref{prop_impute_gen}}. \textcolor{black}{Formally},
 let $\hat{\bm{c}} = (\hat{c}_1,...,\hat{c}_N)^\top$ be the corresponding imputed vector of average potential outcomes, where $\hat{c}_i = \{\hat{Y}_i(0) + \hat{Y}_i(1)\}/{2}$.
The imputation estimator of $\Var_d(\hat{\tau})$ simply plugs in the imputed $\bm{c}$, i.e., $\widehat{\Var}_d(\hat{\tau}) = \textcolor{black}{\psi_d}(\hat{\bm{c}}).$ 
Notice that, by construction, \textcolor{black}{and regardless of the choice of the $\beta_i$s}, $\widehat{\Var}_d(\hat{\tau})$ is non-negative.

\textcolor{black}{It is straightforward to use the imputation approach when the closed-form expression of $\Var_d(\hat{\tau})$ is simple and easy to compute. For instance, under a CRD with equal group sizes, the variance estimator under this approach is given by $\widehat{\Var}_d(\hat{\tau}) = \frac{\hat{S}^2_1}{N/2}+ \frac{\hat{S}^2_0}{N/2}  - \frac{\hat{S}^2_{\tau}}{N}$, where $\hat{S}^2_1$, $\hat{S}^2_0$, and $\hat{S}^2_{\tau}$ are defined analogously to $S^2_1$, $S^2_0$, and $S^2_{\tau}$, respectively (see Section \ref{sec_neyman_decomposition}), with $Y_i(1)$ and $Y_i(0)$ replaced by $\hat{Y}_i(1)$ and $\hat{Y}_i(0)$, respectively.}

Even when the expression of $\Var_d(\hat{\tau})$ is \textcolor{black}{complex or not readily available} in closed form, the variance estimator can be obtained using Monte Carlo methods, akin to those used to approximate p-values in Fisher randomization tests (see, e.g., \cite{imbens2015causal}, Chapter 5). In particular, after imputing the missing potential outcomes, we can estimate the variance as follows.
\begin{enumerate}
    \item For unit $i$, set its potential outcomes as
        $\tilde{Y}_i(1) = W_i Y^{\text{obs}}_i + (1-W_i) \hat{Y}_i(1)$, $\tilde{Y}_i(0) = W_i\hat{Y}_i(0) + (1-W_i)Y^{\text{obs}}_i$.  

    \item Draw $M$ independent vectors of assignments $\bm{W}^{(1)},...,\bm{W}^{(M)}$ under $d$.
    
    \item For $m \in \{1,2,...,M\}$, compute the estimator $\hat{\tau}_m = (2/N)\sum_{i:W^{(m)}_i = 1}\tilde{Y}_i(1) - (2/N)\sum_{i:W^{(m)}_i = 0}\tilde{Y}_i(0)$.

\item Estimate $\Var_d(\hat{\tau})$ by the sample variance of the $\hat{\tau}_m$s, i.e.,  $1/(M-1)\sum_{m=1}^{M}(\hat{\tau}_m - \bar{\hat{\tau}})^2$, where $\bar{\hat{\tau}} = (1/M)\sum_{m = 1}^{M}\hat{\tau}_m$.
\end{enumerate}
Unlike the Neymanian decomposition and the contrast approach, the imputation approach does not necessarily require knowledge of the assignment mechanism. Thus, for a complex design where the joint probabilities of treatment assignments are difficult to obtain, the imputation approach can provide a computationally simpler alternative to variance estimation, even if the design is measurable.

\textcolor{black}{What are admissible choices of the $\beta_i$s? To this end, Proposition \ref{thm_imputation1} shows that as long as $\beta_i$s are deterministic, the resulting imputation-based variance estimator is guaranteed to be conservative. Moreover, in the special case where all $\beta_i$s are set to a common deterministic value (say, $\beta$), the proposition also provides an explicit characterization of the upward bias of the variance estimator under treatment effect homogeneity.}
\begin{proposition} \normalfont
    Let $d$ be a design satisfying $\pi_i = 0.5$ for all $i$. Consider the imputation approach with an arbitrary deterministic vector $(\beta_1,...,\beta_N)^\top$
    , where the potential outcomes are imputed as in Equation \ref{eq_ai}. Then,
    \begin{enumerate}[label=(\alph*)]
        \item The imputation-based variance estimator $\textcolor{black}{\psi_d}(\hat{\bm{c}})$
    satisfies $\E_d\{\textcolor{black}{\psi_d}(\hat{\bm{c}})\} \geq \Var_d(\hat{\tau})$.
    \item Assume that treatment effects are homogeneous, i.e., $Y_i(1) - Y_i(0) = \tau$. If $\beta_i = \beta$ for some arbitrary deterministic $\beta$, then $\E_d\{\textcolor{black}{\psi_d}(\hat{\bm{c}})\} = \Var_d(\hat{\tau}) + (\tau - \beta)^2\E_d\{\textcolor{black}{\psi_d}(\bm{W})\}.$
    \end{enumerate}
\label{thm_imputation1}
\end{proposition}
\textcolor{black}{See the Appendix for a proof. In essence, the proof of part (a) relies on the convexity of $\psi_d(\cdot)$ and on the fact that, with deterministic $\beta_i$s, $\hat{\bm{c}}$ is unbiased for $\bm{c}$. The latter follows from the symmetry of the design, i.e., $\pi_i = 0.5$ for all $i$. With arbitrary deterministic $\beta_i$s, although $\hat{Y}_i(1)$ and $\hat{Y}_i(0)$ may individually be biased for $Y_i(1)$ and $Y_i(0)$, due to this symmetry, their biases are equal in magnitude and opposite in sign. These opposing biases cancel out when averaged, yielding an unbiased estimator of $c_i$.}

Therefore, for any design $d$ satisfying $\pi_i = 0.5$ for all $i$, the imputation approach always yields a conservative variance estimator. This result is true regardless of the value of $(\beta_1,...,\beta_N)^\top$. As a special case, the imputation estimator is conservative when $\beta_i = \beta$ for all $i$, i.e., when the potential outcomes are imputed assuming the sharp null hypothesis of constant treatment effect $\beta$, even though the true treatment effects may be heterogeneous.
Throughout the rest of the section, we \textcolor{black}{focus on this special case where} $\beta_i = \beta$.

By Proposition \ref{thm_imputation1}\textcolor{black}{(b)}, it follows that, when the true unit-level treatment effect is $\tau$, $\E_d\{\textcolor{black}{\psi_d}(\hat{\bm{c}})\}$ can be decomposed into the true variance and a bias term that is quadratic in the difference between the true effect $\tau$ and the \textcolor{black}{postulated} effect $\beta$. Indeed, the bias decreases as $\beta$ gets closer to $\tau$ and vanishes when $\beta = \tau$. Therefore, in practice, if researchers have prior information regarding the magnitude of the average treatment effect (e.g., $\Gamma_1\leq \tau \leq \Gamma_2$ for constants $\Gamma_1$ and $\Gamma_2$), they can leverage the information in the choice of $\beta$ to ensure that the bias of $\textcolor{black}{\psi_d}(\hat{\bm{c}})$ is small. 

The (upward) bias of $\E_d\{\textcolor{black}{\psi_d}(\hat{\bm{c}})\}$ in Proposition \ref{thm_imputation1} under homogeneity also depends on the term $\E_d\{\textcolor{black}{\psi_d}(\bm{W})\}$. If $\E_d\{\textcolor{black}{\psi_d}(\bm{W})\}$ converges to zero as $N$ gets large, then the additive bias of the imputation estimator, $\E_d\{\textcolor{black}{\psi_d}(\hat{\bm{c}})\} - \Var_d(\hat{\tau})$, also goes to zero. Proposition \ref{prop_limit} provides a sufficient design condition for $\E_d\{\textcolor{black}{\psi_d}(\bm{W})\}$ to converge to zero and shows that the condition is satisfied under completely randomized designs.
\begin{proposition} \normalfont
Denote $\pi_{ij} = \Pr_d(W_i = 1, W_j  =1)$ and let the imputation approach be as in Proposition \ref{thm_imputation1}\textcolor{black}{(b)}. For every $\bm{w} \in \mathcal{W}$, suppose $\sum_{i=1}^{N}w_i = N/2$ and $(16/N^2)\mathop{\sum\sum}_{i<j:w_i = w_j}[\pi_{ij} - (N-2)/\{4(N-1)\}] + 1/(N-1) = o(1)$. Then, \textcolor{black}{under treatment effect homogeneity,} $\E_d\{\textcolor{black}{\psi_d}(\hat{\bm{c}})\} - \Var_d(\hat{\tau}) = (\tau - \beta)^2 o\left(1\right).$
In particular, for a completely randomized design with equal group sizes, $\E_d\{\textcolor{black}{\psi_d}(\hat{\bm{c}})\} - \Var_d(\hat{\tau}) = (\tau - \beta)^2 /(N-1).$
     \label{prop_limit}
\end{proposition}
In other words, the additive bias of $\textcolor{black}{\psi_d}(\hat{\bm{c}})$ goes to zero if design $d$ admits two groups of equal size and a condition on the pairwise probabilities of treatment. 
This condition is satisfied in a CRD, where $\pi_{ij} = (N-2)/\{4(N-1)\}$ for all $i \neq j$. The term $(16/N^2)\mathop{\sum\sum}_{i<j:w_i = w_j}[\pi_{ij} - (N-2)/\{4(N-1)\}]$ measures an average difference between the pairwise probabilities under design $d$ and those under a CRD. 
In this sense, this condition can also be interpreted as a form of exchangeability condition on pairs. Overall, Proposition \ref{prop_limit} shows that for a class of designs with equal-sized groups that satisfy this exchangeability condition, the bias of the imputation-based variance estimator tends to zero \textcolor{black}{under treatment effect homogeneity} even if the true unit-level treatment effect $\tau$ is imputed incorrectly by $\beta$.

\textcolor{black}{Moreover, we can also explicitly characterize the upward bias of $\psi_d(\hat{\bm{c}})$ for certain standard designs, even under treatment effect heterogeneity. Propositions~\ref{prop_app_imp_crd} and \ref{prop_app_imp_match} in the Appendix formalize these results for CRD and matched-pair designs, respectively, and derive necessary and sufficient conditions under which $\psi_d(\hat{\bm{c}})$ is no more conservative than the standard variance estimators for these designs, namely $\hat{V}_{\text{Neyman}}$ and $\hat{V}_{\text{pair}}$. Equivalently, these results provide a direct comparison between the imputation approach and the contrast approach, since the latter reduces to these standard variance estimators (by Theorems~\ref{thm_substitute_neyman} and \ref{thm_substitute_matched}). In both cases, $\psi_d(\hat{\bm{c}})$ is shown to have an equal or smaller bias when $\beta$ lies within an interval of the form $[\tau - s, \tau + s]$, where $s$ represents a form of scaled standard deviation of the unit-level treatment effects. Thus, for these designs, when the postulated treatment effect $\beta$ is sufficiently close to the true average treatment effect $\tau$, the resulting imputation-based estimator can be less conservative than the contrast-based estimator, even under heterogeneous treatment effects.}    

Since Proposition \ref{thm_imputation1} \textcolor{black}{(as well as Propositions~\ref{prop_app_imp_crd} and \ref{prop_app_imp_match})} hints at choosing a $\beta$ that is close to $\tau$ (or is a reasonable guess for $\tau$), it is tempting to use $\beta = \hat{\tau}$. The resulting approach differs from the imputation approach discussed earlier in that the missing potential outcomes are now imputed by a random (as opposed to fixed) quantity. More importantly, under this approach, the imputed $\hat{\bm{c}}$ is no longer unbiased for $\bm{c}$ and hence, there is no guarantee that the resulting variance estimator $\textcolor{black}{\psi_d}(\hat{\bm{c}})$ is conservative for $\Var_d(\hat{\tau})$. 
However, as we discuss below, this approach still yields reasonable variance estimators for a large class of designs.
In particular, Proposition \ref{prop_imputation_neyman} shows that under CRD, $\textcolor{black}{\psi_d}(\hat{\bm{c}})$ is asymptotically equivalent to the Neymanian estimator.
\begin{proposition} \normalfont
 Consider the imputation approach for a completely randomized design with equal group sizes, where the missing potential outcomes are imputed as if the true unit-level treatment effect is $\hat{\tau}$. Then the corresponding imputation-based variance estimator $\textcolor{black}{\psi_d}(\hat{\bm{c}})$ satisfies $\textcolor{black}{\psi_d}(\hat{\bm{c}}) = \hat{V}_{\text{Neyman}}\times (N-2)/(N-1).$
 \label{prop_imputation_neyman}
\end{proposition}
Proposition \ref{prop_imputation_neyman} connects the imputation approach to the Neymanian estimator under CRD and shows that the imputation estimator $\textcolor{black}{\psi_d}(\hat{\bm{c}})$ (based on $\hat{\tau}$) is algebraically equivalent to a scaled version of the Neymanian estimator, where the scaling factor is $(N-2)/(N-1)<1$. 
Thus, under treatment effect homogeneity, $\E_d\{\textcolor{black}{\psi_d}(\hat{\bm{c}})\} <\Var_d(\hat{\tau})$, and hence $\textcolor{black}{\psi_d}(\hat{\bm{c}})$ is slightly anti-conservative in finite samples. 
\textcolor{black}{By applying the correction factor $(N-1)/(N-2)$, we recover $\hat{V}_{\text{Neyman}}$, thereby obtaining conservative variance estimator in finite samples.}
In large samples, however, $\textcolor{black}{\psi_d}(\hat{\bm{c}})$ is equivalent to the Neymanian estimator. Thus, for sufficiently large $N$, $\textcolor{black}{\psi_d}(\hat{\bm{c}})$ is approximately unbiased for $\Var_d(\hat{\tau})$ under homogeneity. 
Finally, this asymptotic equivalence also shows that the Neymanian estimator can be alternatively derived from a Fisherian mode of inference. See \cite{samii2012equivalencies} for an equivalent result, where they established a connection between the homoskedastic variance estimator from OLS regression and $\Var(\hat{\tau})$ under constant treatment effects equal to $\hat{\tau}$.

For designs beyond CRD, $\textcolor{black}{\psi_d}(\hat{\bm{c}})$ \textcolor{black}{(and likewise $\frac{N-1}{N-2}\psi_d(\hat{\bm{c}})$)} is not guaranteed to be anti-conservative (or conservative) in finite samples. However, as shown in Proposition \ref{prop_imp_tauhat_asymp}, the additive bias of $\textcolor{black}{\psi_d}(\hat{\bm{c}})$ goes to zero under mild conditions on the design and the potential outcomes.
\begin{proposition} \normalfont
Assume that the control potential outcomes satisfy $|Y_i(0)|\leq B$ for some $B>0$ and consider a design $d$ such that $\hat{\tau} - \tau = o_P(1)$. Then, under treatment effect homogeneity, \textcolor{black}{the imputation estimator based on $\hat{\tau}$ satisfies} $\E_d\{\textcolor{black}{\psi_d}(\hat{\bm{c}})\} - \Var_d(\hat{\tau}) = o(1).$
    \label{prop_imp_tauhat_asymp}
\end{proposition}
Therefore, if the potential outcomes are bounded and if $\hat{\tau}$ is consistent for $\tau$ under design $d$, then the variance estimator under the imputation approach (based on $\hat{\tau}$) has negligible bias in large samples. This result complements our observations from Proposition \ref{thm_imputation1} by providing a justification for setting $\beta = \hat{\tau}$.

\subsection{Extensions to general experimental designs: direct imputation}
\label{sec_imputation_extension}

In this section, we extend the imputation approach to a general design $d$, \textcolor{black}{where $\pi_i$ may differ from 0.5 for some units.} 
\textcolor{black}{For such designs, however, the imputation approach described in the previous section may not yield conservative variance estimators, even with deterministic $\beta_i$s. The key reason for this is that, unless $\pi_i = 0.5$, the imputed $c_i$, i.e., $\hat{c}_i = (1-\pi_i)\hat{Y}_i(1) + \pi_i \hat{Y}_i(0)$ is not unbiased for $c_i = (1-\pi_i)Y_i(1) + \pi_i Y_i(0)$.} 
Nevertheless, under treatment effect homogeneity and with some additional design conditions, we can \textcolor{black}{attain conservativeness}; see Appendix \ref{appsec_bias_imputation} for details.

 Now, since $\Var_d(\hat{\tau})$ is a function of the potential outcomes only through $\bm{c}$, we can alternatively consider directly imputing $\bm{c}$, bypassing the steps to impute the potential outcomes $\bm{Y}(1)$ and $\bm{Y}(0)$. We term this the \textit{direct imputation} approach.
More concretely, we estimate the variance by $\textcolor{black}{\psi_d}(\hat{\bm{c}})$, where $\hat{\bm{c}} = (\hat{c}_1,...,\hat{c}_N)^\top$ is some unbiased estimator of $\bm{c}$. Since $\textcolor{black}{\psi_d}(\cdot)$ is a convex function, it follows that $\E_d\{\textcolor{black}{\psi_d}(\hat{\bm{c}})\} \geq \textcolor{black}{\psi_d}\{\E_d(\hat{\bm{c}})\} = \Var_d(\hat{\tau}).$
Thus, the direct imputation approach leads to a conservative estimator of the variance of $\hat{\tau}$. 
Moreover, even when $\textcolor{black}{\psi_d}(\bm{c})$ is not known or difficult to obtain in closed form, the variance estimator using this approach can be obtained using Monte Carlo methods.

To obtain a suitable estimator of $c_i$, we consider a class of linear estimators of the form,
    \begin{equation}
  \hat{c}_i =
    \begin{cases}
    
      \frac{1-\pi_i}{\pi_i} Y^{\text{obs}}_i - (1-\pi_i)\gamma_i & \text{if $W_i = 1$}\\
    \frac{\pi_i}{1-\pi_i} Y^{\text{obs}}_i + \pi_i\gamma_i & \text{if $W_i = 0$},
    \end{cases}  
    \label{eq_ci}
\end{equation}
where $\gamma_i$ is a (deterministic/random) number to be set by the investigator. In fact, if $\gamma_i$s are deterministic, then the above is the unique class of linear estimators that is unbiased for $c_i$ (see Proposition \ref{prop_altimpute} in the Appendix). 
Moreover, when $\pi_i = 0.5$, the direct imputation approach boils down to the standard imputation approach discussed in Section \ref{sec_imputation_general}, with $\beta_i = \gamma_i$.


What is a reasonable choice of $\gamma_i$ in practice? As shown in Section \ref{sec_imputation_general}, even for designs with $\pi_i = 0.5$, choosing $\gamma_i = \hat{\tau}$ may lead to anti-conservative variance estimators in finite samples. 
In this section, we propose a fix to this problem using a Jackknife approach. Roughly speaking, instead of setting $\gamma_i=\hat{\tau}$, this approach sets $\gamma_i$ as a leave-one-out version of $\hat{\tau}$ that excludes unit $i$, which in turn allows us to unbiasedly estimate $c_i$.

To formalize, consider the class of linear estimators in $\hat{c}_i$ in Equation \ref{eq_ci}. Proposition \ref{prop_bias} derives the bias of $\hat{c}_i$.
\begin{proposition} \normalfont
For an arbitrary and possibly random $\gamma_i$, the bias of $\hat{c}_i$ is $\E_d(\hat{c}_i) - c_i  = \pi_i(1-\pi_i)\{\E_d(\gamma_i|W_i = 0) - \E_d(\gamma_i|W_i = 1)\}.$
\label{prop_bias}
\end{proposition}
Thus, a necessary and sufficient condition for the bias to be zero is that $\E_d(\gamma_i|W_i = 0) = \E_d(\gamma_i|W_i = 1)$, i.e., $\gamma_i$ is mean-independent of $W_i$.
As alluded to earlier, the bias term is zero if $\gamma_i$ is deterministic but may not be zero when $\gamma_i$ is random. In particular, under CRD with equal group sizes, setting $\gamma_i = \hat{\tau}$ implies $\E_d(\gamma_i|W_i = 0) - \E_d(\gamma_i|W_i = 1) = \{4/(N-1)\}(\bar{c} - c_i),$ where $\bar{c}$ is the mean of $c_i$ across the $N$ units. Thus, under CRD, the bias of this $\hat{c}_i$ vanishes if and only if $c_i = \bar{c}$, i.e., the average potential outcome is constant across units.   

To obtain a suitable unbiased estimator of $c_i$, \textcolor{black}{we propose setting} $\gamma_i = \hat{\theta}_{(-i)}$, where
\begin{equation}
    \hat{\theta}_{(-i)} = \frac{1}{|\mathcal{S}_1(i)|}\sum_{j\in \mathcal{S}_1(i)}\frac{W_jY^{\text{obs}}_j(1-\pi_j)}{\tilde{\pi}_j \pi_j} - \frac{1}{|\mathcal{S}_0(i)|}\sum_{j\in\mathcal{S}_0(i)}\frac{(1-W_j)Y^{\text{obs}}_j {\pi}_j}{(1-\tilde{\pi}_j)(1-\pi_j)}. \label{eq_jack3}
\end{equation}
Here, for $z\in \{0,1\}$, $\mathcal{S}_z(i) = \{j\neq i: \Pr_d(W_j=z|W_i = w)>0 \text{ for } w\in \{0,1\}\}$; $\tilde{\pi}_j = \Pr_d(W_j = 1|W_i = 1)$ if $W_i = 1$, and $\tilde{\pi}_j = \Pr_d(W_j = 1|W_i = 0)$ if $W_i = 0$.
\textcolor{black}{Appendix \ref{appsec_jack} provides a heuristic justification for this choice of $\gamma_i$ by connecting the direct imputation method to an implicit imputation of $Y_i(1)$ and $Y_i(0)$}. \textcolor{black}{Essentially, $\hat{\theta}_{(-i)}$ represents} a leave-one-out version of \textcolor{black}{a Horvitz-Thompson-type estimator} $\hat{\theta}$ that excludes unit $i$, \textcolor{black}{where $\hat{\theta} = (1/N)\sum_{j}W_jY^{\text{obs}}_j(1-\pi_j)/{\pi^2_j} - (1/N)\sum_{j}{(1-W_j)Y^{\text{obs}}_j\pi_j}/{(1-\pi_j)^2}$.}
\textcolor{black}{Note that when computing $\hat{\theta}_{(-i)}$ from $\hat{\theta}$, the weights are adjusted based on the treatment assignment of unit $i$. Specifically, rather than weighting unit $j$ by the inverse of $\Pr_d(W_j = 1)$, we weight it by the inverse of $\Pr_d(W_j = 1 \mid W_i)$. Moreover, the summations are taken over the subsets of units $\mathcal{S}_1(i)$ and $\mathcal{S}_0(i)$ to ensure that the denominators are nonzero.}
As a special case, under a CRD with equal group sizes, $\hat{\theta}_{(-i)}$ boils down to the standard difference-in-means statistic, leaving out unit $i$.

Now, Proposition \ref{prop_jack_general2} shows that this choice of $\gamma_i$ indeed leads to an unbiased estimator of $c_i$.
\begin{proposition}\normalfont
Let $d$ be an arbitrary design with $\pi_i \in (0,1)$. Consider a direct imputation estimator $\hat{c}_i$ with $\gamma_i = \hat{\theta}_{(-i)}$, where $\hat{\theta}_{(-i)}$ is as defined in Equation \ref{eq_jack3}. It follows that,
$$\textcolor{black}{\E_d(\hat{\theta}_{(-i)}|W_i = 1) = \E_d(\hat{\theta}_{(-i)}|W_i = 0) = \frac{1}{|\mathcal{S}_1(i)|}\sum_{j \in \mathcal{S}_1(i)}\frac{1-\pi_j}{\pi_j}Y_j(1) - \frac{1}{|\mathcal{S}_0(i)|}\sum_{j \in \mathcal{S}_0(i)}\frac{\pi_j}{1-\pi_j}Y_j(0),}$$
and hence $\E_d(\hat{c}_i) = c_i$ \textcolor{black}{and $\E_d\{\textcolor{black}{\psi_d}(\hat{\bm{c}})\} \geq \Var_d(\hat{\tau})$}.
    \label{prop_jack_general2}    
\end{proposition}
Thus, by Proposition \ref{prop_jack_general2}, $\textcolor{black}{\psi_d}(\hat{\bm{c}})$ based on the Jackknife estimator $\hat{\theta}_{(-i)}$ is conservative. As shown in Theorem \ref{thm_jack2_bias} below, the upward bias of this estimator can be explicitly characterized under complete randomization and treatment effect homogeneity. 
\begin{theorem} \normalfont
Let $d$ be a completely randomized design with equal group sizes. Denote $\hat{V}_{\text{Jack}}$ as the variance estimator under the direct imputation approach with $\gamma_i = \hat{\theta}_{(-i)}$, where $\hat{\theta}_{(-i)}$ is as defined in Equation \ref{eq_jack3}. It follows that, if $Y_i(1) - Y_i(0) = \tau$ for all $i \in \{1,2,...,N\}$, $\E_d(\hat{V}_{\text{Jack}}) = \Var_d(\hat{\tau}) \times (N-1)/(N-2).$
\label{thm_jack2_bias}
\end{theorem}
We recall that, when $\gamma_i = \hat{\tau}$, the resulting variance estimator $\hat{V}$ satisfied $\E_d(\hat{V}) = \Var_d(\hat{\tau}){(N-2)}/{(N-1)}$. Therefore, using the Jackknifed version of $\hat{\tau}$ reverses the scaling factor $(N-2)/(N-1)$ and produces a conservative estimator for $\Var_d(\hat{\tau})$. 

The jackknife-based imputation estimator can be readily implemented using the Monte Carlo method, even when $\psi_d(\bm{c})$ is not available in closed form. 
The steps are outlined below.
\begin{enumerate}
    \item For unit $i \in \{1,2,...,N\}$,
    \begin{enumerate}
        \item Compute $\gamma_i = \hat{\theta}_{(-i)}$.
        \item Compute $\beta_i = W_i \{\frac{2\pi_i - 1}{\pi^2_i} Y^{\text{obs}}_i + \frac{1-\pi_i}{\pi_i}\gamma_i\} + (1-W_i) \{ \frac{2\pi_i - 1}{(1-\pi_i)^2} Y^{\text{obs}}_i + \frac{\pi_i}{1-\pi_i}\gamma_i\}$.
        \item Impute the potential outcomes for unit $i$ as $\tilde{Y}_i(1) = W_i Y^{\text{obs}}_i + (1-W_i) (Y^{\text{obs}}_i + \beta_i)$, $\tilde{Y}_i(0) = W_i(Y^{\text{obs}}_i - \beta_i) + (1-W_i)Y^{\text{obs}}_i$.  
    \end{enumerate}

    \item Draw $M$ independent vectors of assignments $\bm{W}^{(1)},...,\bm{W}^{(M)}$ under $d$.
    
    \item For $m \in \{1,2,...,M\}$, compute the estimator, $\hat{\tau}_m = (1/N)\sum_{i=1}^{N}W^{(m)}_i \tilde{Y}_i(1)/\pi_i - (1/N)\sum_{i=1}^{N}(1-W^{(m)}_i) \tilde{Y}_i(0)/(1-\pi_i).$.

    \item Estimate $\Var_d(\hat{\tau})$ by the sample variance of the $\hat{\tau}_m$s, i.e.,  $1/(M-1)\sum_{m=1}^{M}(\hat{\tau}_m - \bar{\hat{\tau}})^2$, where $\bar{\hat{\tau}} = (1/M)\sum_{m = 1}^{M}\hat{\tau}_m$.
\end{enumerate}
In 1(b), $\beta_i$s represent the implicit unit-level causal effect postulated by the direct imputation approach; see Appendix \ref{appsec_jack} for details. 

In Appendix \ref{appsec_alternative}, we discuss alternative choices of $\gamma_i$ that lead to conservative variance estimators. In a simulation study in Appendix \ref{sec_simulation}, we compare the relative biases of the variance estimators for different choices of $\gamma_i$ under completely randomized designs\textcolor{black}{, as well as stratified randomized designs}. The results indicate that the jackknifed variance estimator with $\gamma_i = \hat{\theta}_{(-i)}$ performs reasonably well across scenarios, especially when treatment effects are homogeneous.

In the remainder of this section, we evaluate the performance of this jackknifed variance estimator $\hat{V}_{\text{Jack}}$ with other variance estimators using two simulation studies (A and B). For comparison, we consider the commonly used design-based variance estimator of \cite{aronow2013class}, $\hat{V}_{\text{AM}}$, which uses a variance expansion (similar to those discussed in Section \ref{sec_neyman_decomposition}) and bounds each term that is non-identifiable using Young's inequality. In addition, we consider the imputation-based variance estimator with $\gamma_i = \hat{\tau}$, denoted by $\hat{V}_{\hat{\tau}}$. 

In simulation study A, we consider a rerandomized design with
$N = 12$, $N_t = N_t = 6$, and a single covariate $X_i$ such that $X_i \overset{iid}{\sim} \mathcal{N}(10,1)$ for $i \in \{1,2\}$, and $X_i \overset{iid}{\sim} \mathcal{N}(0,1)$ for $i \in \{3,...,N\}$. To implement the design, we randomly draw an assignment vector under complete randomization and rerandomize until the mean imbalance in $X$ is small enough, in particular, until the absolute standardized mean difference (ASMD) in $X$, $|\bar{X}_t - \bar{X}_c|/\sqrt{(s^2_t + s^2_c)/2}$ is smaller than 0.2.\footnote{Here, $\bar{X}_t$ and $\bar{X}_c$ denote the covariate means, and $s^2_t$ and $s^2_c$ the variances, in the treated and control groups, respectively.} The choice of the covariate and the balancing criterion is chosen specifically to ensure that $\Pr(W_1 = 1,W_2 = 1) = 0$, making the design non-measurable.

In simulation study B, we consider a rerandomized design with
$N=50$, $N_t = N_c = 25$, and $6$ covariates, generated according to the simulation design in   \cite{hainmueller2012balancing} (see also, \cite{chattopadhyay2022balanced}).
\begin{equation}
\begin{psmallmatrix}
X_{1}\\
X_{2}\\
X_3
\end{psmallmatrix} \sim
\mathcal{N}_3\left\{\begin{psmallmatrix}
0\\
0\\
0
\end{psmallmatrix},\begin{psmallmatrix}
2 & 1 & -1\\
1 & 1 & -0.5\\
-1 & -0.5 & 1
\end{psmallmatrix}\right\},\hspace{0.1cm}
 X_4 \sim \text{Unif}(-3,3), \hspace{0.1cm} X_5 \sim \chi^2_1, \hspace{0.1cm} X_6 \sim \text{Bernoulli}(0.5). \label{dgp}
\end{equation}
Here, $X_4$, $X_5$, and $X_6$ are mutually independent and separately independent of $(X_1,X_2,X_3)^\top$. To implement the design, we randomly draw an assignment vector under complete randomization and rerandomize until the maximum of the ASMDs across the six covariates is smaller than 0.2.

Under each simulation study, we consider four different generative models for the potential outcomes, as shown below.
 \begin{enumerate}
        \item No effect whatsoever: $Y_i(0) \overset{iid}{\sim} \text{Unif}(0,10)$, $Y_i(1) = Y_i(0)$ for all $i$. 

        \item Constant effect = 5: $Y_i(0) \overset{iid}{\sim} \text{Unif}(0,10)$, $Y_i(1) = Y_i(0) + 5$ for all $i$.

        \item Constant effect, varying across simulations: $Y_i(0) \overset{iid}{\sim} \text{Unif}(0,10)$, $Y_i(1) = Y_i(0) + \tau$ for all $i$, where $\tau \sim \text{Unif}(-5,5)$ \textcolor{black}{and is independent of $\bm{Y}(0)$}. 

        \item Heterogeneous effects: $Y_i(0) \overset{iid}{\sim} \text{Unif}(0,10)$, $Y_i(1) = Y_i(0) + \tau_i$ for all $i$, where $\tau_i \overset{iid}{\sim} \text{Unif}(-5,5)$ \textcolor{black}{and are independent of $\bm{Y}(0)$}. 
    \end{enumerate}
For each of the above scenarios, \textcolor{black}{we generate the set of potential outcomes}, and for a variance estimator $\hat{V}$, compute its relative bias \textcolor{black}{with respect to the true design-based variance}, $\frac{\E_d(\hat{V}) - \Var_d(\hat{\tau})}{\Var_d(\hat{\tau})}$, and its standard deviation. This process is then replicated, each time independently generating the potential outcomes according to the specified data-generating process, \textcolor{black}{thereby producing a distribution of these performance metrics under that process}. 
Figure \ref{fig:simu1} and \ref{fig:simu2} displays the resulting distributions of relative bias and standard deviation of each estimator under simulation studies A and B, respectively.
\begin{figure}[!ht]
    \centering
    \includegraphics[width=\linewidth]{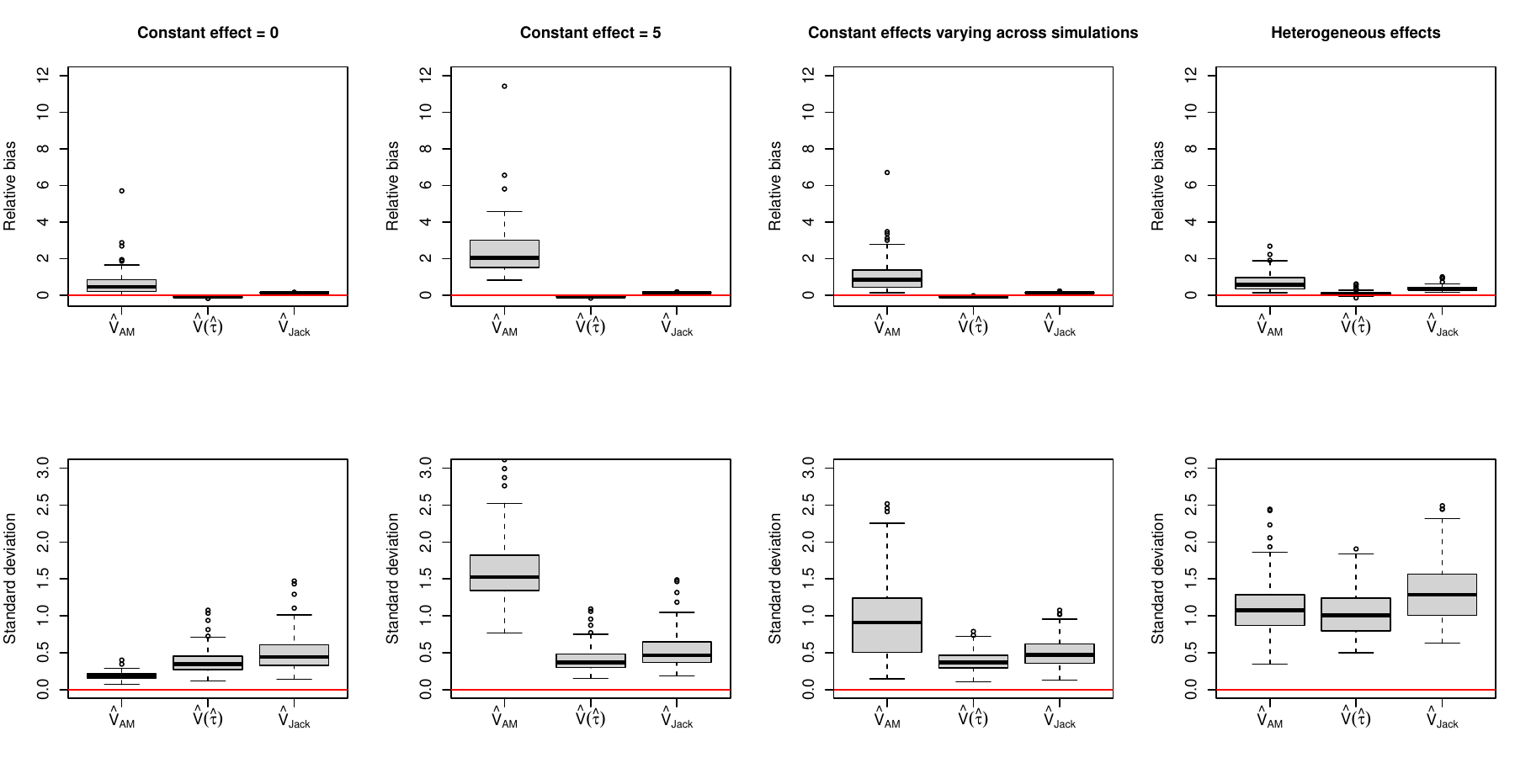}
    \caption{Relative bias and standard error of variance estimators across different scenarios under simulation study A}
    \label{fig:simu1}
\end{figure}

\begin{figure}[!ht]
    \centering
    \includegraphics[width=\linewidth]{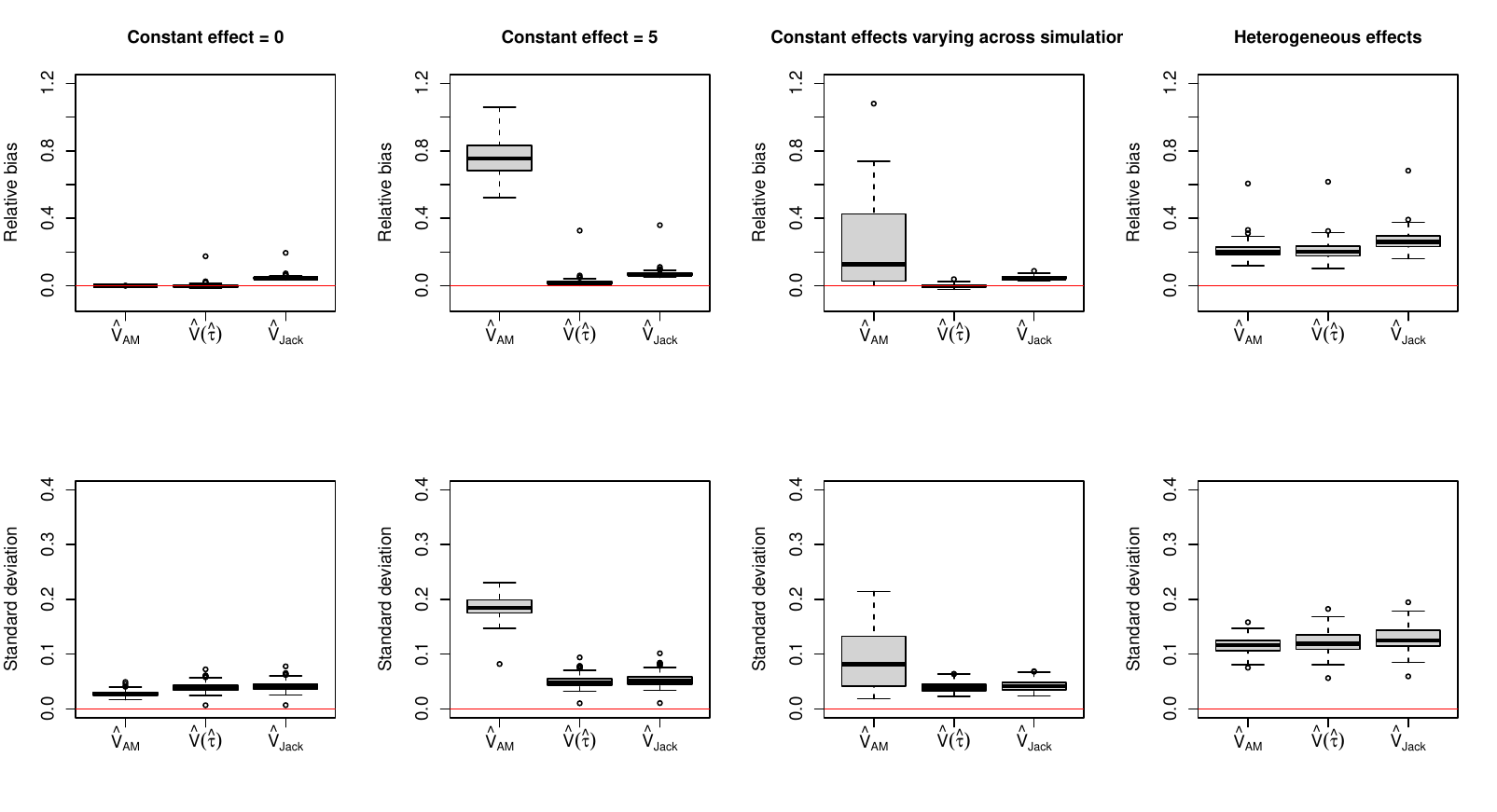}
    \caption{Relative bias and standard error of variance estimators across different scenarios under simulation study B}
    \label{fig:simu2}
\end{figure}
Figures \ref{fig:simu1} and \ref{fig:simu2} show that, as expected, both $\hat{V}_{\text{AM}}$ and $\hat{V}_{\text{Jack}}$ are conservative estimators (i.e., their relative bias is non-negative), whereas $\hat{V}_{\hat{\tau}}$ is not necessarily conservative. Across all scenarios, $\hat{V}_{\text{Jack}}$ performs well in terms of both relative bias and standard deviation, outperforming the other estimators when treatment effects are homogeneous. 
\textcolor{black}{Roughly speaking, in the construction of $\hat{V}_{\text{AM}}$, the use of Young's inequality for each non-identifiable term introduces local upper bounds that tend to drive a larger amount of conservativeness than $\hat{V}_{\text{Jack}}$. In Proposition \ref{prop_app_am} in the Appendix, we derive the bias of $\hat{V}_{\text{AM}}$ and show that the bias depends on the sum of squares of unit-level treatment effects, as well as on the sum of squares of pairwise sums and differences of potential outcomes over all non-identifiable pairs. These components collectively contribute to the higher conservativeness of the estimator. As a special case, when treatment effects are homogeneous and equal to $\tau$, and the design is measurable (such as CRD), $\E_d(\hat{V}_{\text{AM}}) - \Var_d(\hat{\tau}) = \frac{\tau^2}{N}$. In comparison, for a CRD with equal group sizes and under homogeneity, Theorem~\ref{thm_jack2_bias} shows that $\E_d(\hat{V}_{\text{Jack}}) - \Var_d(\hat{\tau}) = \frac{4S^2_1}{N(N-2)}$, which decays at a faster rate than the bias of $\hat{V}_{\text{AM}}$ when $\tau$ is different from zero.}

\textcolor{black}{Finally, in Appendix \ref{sec_app_imp_standard}, we compare the imputation-based estimators with standard variance estimators (equivalently, the contrast-based estimators) for common designs, such as $\hat{V}_{\text{Neyman}}$ (for CRD) and $\hat{V}_{\text{pair}}$ (for matched-pair designs). Our results show that under both designs, the two estimators perform very similarly in terms of relative bias and standard deviation, with the contrast-based estimators exhibiting marginally smaller bias and standard deviation than the imputation-based estimators for smaller population sizes. The additional structure in the contrast-based estimators translates into unbiasedness under homogeneity, and slightly less conservativeness than the imputation-based estimator under heterogeneity.}

\section{Final thoughts}
\label{sec_final}

In this paper, we explored the problem of Neymanian inference for average treatment effects in general experimental designs, offering insights into existing principles and introducing two novel approaches: the contrast approach and the imputation approach, both of which offer new perspectives to interpret Neyman's original approach. Each of these approaches produces a class of variance estimators that, although distinct in general, align with Neyman's estimator under complete randomization.
We analyzed the theoretical properties of both approaches, demonstrating that each yields conservative variance estimators that are unbiased (exactly or approximately) under homogeneous treatment effects for a broad class of designs.
While our focus in this paper was on elucidating the principles that allow us to conduct Neymanian inference for designs beyond complete randomization, the choice of an estimator matters in practice. In this paper, through simulation studies, we found that the Jackknife imputation-based variance estimator performs robustly across different scenarios.
An important direction for future work is to identify optimal variance estimators within each approach (along the lines of \cite{harshaw2021optimized}).
\textcolor{black}{Moreover, while our analysis in this paper focused on the finite-sample behavior of the proposed estimators (for example, their conservativeness) and the large-sample behavior of their bias, a relevant direction for future exploration is to study other important asymptotic properties, such as consistency. In a design-based framework, establishing consistency under a general design is nontrivial, since it depends on both the assignment mechanism and the behavior of the hypothetical sequence of underlying finite populations. Nevertheless, for standard designs such as completely randomized and matched-pair designs, these estimators are consistent under mild regularity conditions (see, for example, \cite{li2017general}), since in these cases they reduce to scaled versions of the corresponding standard variance estimators (see Theorems \ref{thm_substitute_neyman}, \ref{thm_substitute_matched}, and Proposition \ref{prop_imputation_neyman}).}

\textcolor{black}{Our discussion in this paper does not directly incorporate pre-treatment covariates. In randomized experiments, one may estimate $\tau$ using a regression-adjusted estimator, such as the interacted regression method of \cite{lin2013agnostic}, in place of the unbiased Horvitz–Thompson estimator $\hat{\tau}$ considered in this paper. However, the design-based variances of regression-adjusted estimators generally do not have a tractable closed-form expression in finite samples, which limits the direct use of the contrast approach. In addition, these variances depend separately on the two potential outcomes rather than through their average, so establishing conservativeness of the resulting variance estimator under the imputation approach is nontrivial. 
In a similar vein, for estimating the variance of $\hat{\tau}$, covariate adjustment can be introduced within the imputation step, for example, by imputing potential outcomes using regression-based estimates of treatment effects. Understanding how such covariate-adjusted imputations affect the bias of the constructed $\hat{c}_i$, and whether they yield conservative variance estimators under common designs, is a promising direction for future research.}


\bibliographystyle{unsrtnat}
\bibliography{mybibliography23.bib}

@article{ding2017paradox,
  title={A paradox from randomization-based causal inference},
  author={Ding, Peng},
  journal={Statistical science},
  pages={331--345},
  year={2017},
  publisher={JSTOR}
}

@article{li2017general,
  title={General forms of finite population central limit theorems with applications to causal inference},
  author={Li, Xinran and Ding, Peng},
  journal={Journal of the American Statistical Association},
  volume={112},
  number={520},
  pages={1759--1769},
  year={2017},
  publisher={Taylor \& Francis}
}

@article{lin2013agnostic,
  title={Agnostic notes on regression adjustments to experimental data: Reexamining Freedman's critique},
  author={Lin, Winston},
  journal={The Annals of Applied Statistics},
  pages={295--318},
  year={2013},
  publisher={JSTOR}
}

@book{lohr2021sampling,
  title={Sampling: design and analysis},
  author={Lohr, Sharon L},
  year={2021},
  publisher={Chapman and Hall/CRC}
}

@article{wu2021randomization,
  title={Randomization tests for weak null hypotheses in randomized experiments},
  author={Wu, Jason and Ding, Peng},
  journal={Journal of the American Statistical Association},
  volume={116},
  number={536},
  pages={1898--1913},
  year={2021},
  publisher={Taylor \& Francis}
}

@article{ding2018randomization,
  title={A randomization-based perspective on analysis of variance: a test statistic robust to treatment effect heterogeneity},
  author={Ding, Peng and Dasgupta, Tirthankar},
  journal={Biometrika},
  volume={105},
  number={1},
  pages={45--56},
  year={2018},
  publisher={Oxford University Press}
}

@article{kempthorne1955randomization,
  title={The randomization theory of experimental inference},
  author={Kempthorne, Oscar},
  journal={Journal of the American Statistical Association},
  volume={50},
  number={271},
  pages={946--967},
  year={1955},
  publisher={Taylor \& Francis}
}

@article{higgins2015blocking,
  title={Blocking estimators and inference under the Neyman-Rubin model},
  author={Higgins, Michael J and S{\"a}vje, Fredrik and Sekhon, Jasjeet S},
  journal={arXiv preprint arXiv:1510.01103},
  year={2015}
}

@article{imai2008variance,
  title={Variance identification and efficiency analysis in randomized experiments under the matched-pair design},
  author={Imai, Kosuke},
  journal={Statistics in medicine},
  volume={27},
  number={24},
  pages={4857--4873},
  year={2008},
  publisher={Wiley Online Library}
}

@article{gadbury2001randomization,
  title={Randomization inference and bias of standard errors},
  author={Gadbury, Gary L},
  journal={The American Statistician},
  volume={55},
  number={4},
  pages={310--313},
  year={2001},
  publisher={Taylor \& Francis}
}

@article{abadie2008estimation,
  title={Estimation of the conditional variance in paired experiments},
  author={Abadie, Alberto and Imbens, Guido W},
  journal={Annales d'Economie et de Statistique},
  pages={175--187},
  year={2008},
  publisher={JSTOR}
}

@article{fogarty2018mitigating,
  title={On mitigating the analytical limitations of finely stratified experiments},
  author={Fogarty, Colin B},
  journal={Journal of the Royal Statistical Society Series B: Statistical Methodology},
  volume={80},
  number={5},
  pages={1035--1056},
  year={2018},
  publisher={Oxford University Press}
}

@article{pashley2021insights,
  title={Insights on variance estimation for blocked and matched pairs designs},
  author={Pashley, Nicole E and Miratrix, Luke W},
  journal={Journal of Educational and Behavioral Statistics},
  volume={46},
  number={3},
  pages={271--296},
  year={2021},
  publisher={SAGE Publications Sage CA: Los Angeles, CA}
}

@article{wilk1955randomization,
  title={The randomization analysis of a generalized randomized block design},
  author={Wilk, Martin B},
  journal={Biometrika},
  volume={42},
  number={1/2},
  pages={70--79},
  year={1955},
  publisher={JSTOR}
}

@article{aronow2013class,
  title={A class of unbiased estimators of the average treatment effect in randomized experiments},
  author={Aronow, Peter M and Middleton, Joel A},
  journal={Journal of Causal Inference},
  volume={1},
  number={1},
  pages={135--154},
  year={2013},
  publisher={De Gruyter}
}

@article{chattopadhyay2022balanced,
  title={Balanced and robust randomized treatment assignments: The finite selection model for the health insurance experiment and beyond},
  author={Chattopadhyay, Ambarish and Morris, Carl N and Zubizarreta, Jos{\'e} R},
  journal={arXiv preprint arXiv:2205.09736},
  year={2022}
}

@article{aronow2013conservative,
  title={Conservative variance estimation for sampling designs with zero pairwise inclusion probabilities},
  author={Aronow, Peter M and Samii, Cyrus},
  journal={Survey Methodology},
  volume={39},
  number={1},
  pages={231--241},
  year={2013},
  publisher={Statistics Canada}
}

@phdthesis{luo2023interval,
  title={Interval Estimation for the Average Treatment Effect in Randomized Experiments and Mata-Analysis},
  author={Luo, Xiaokang},
  year={2023},
  school={Rutgers The State University of New Jersey, School of Graduate Studies}
}

@article{robins1988confidence,
  title={Confidence intervals for causal parameters},
  author={Robins, James M},
  journal={Statistics in medicine},
  volume={7},
  number={7},
  pages={773--785},
  year={1988},
  publisher={Wiley Online Library}
}

@article{aronow2014sharp,
  title={Sharp bounds on the variance in randomized experiments},
  author={Aronow, Peter M and Green, Donald P and Lee, Donald KK},
journal = {Annals of Statistics},
volume = {42},
number = {3},
pages = {850-871},
year={2014}
}

@article{nutz2022directional,
  title={The directional optimal transport},
  author={Nutz, Marcel and Wang, Ruodu},
  journal={The Annals of Applied Probability},
  volume={32},
  number={2},
  pages={1400--1420},
  year={2022},
  publisher={Institute of Mathematical Statistics}
}

@article{mukerjee2018using,
  title={Using standard tools from finite population sampling to improve causal inference for complex experiments},
  author={Mukerjee, Rahul and Dasgupta, Tirthankar and Rubin, Donald B},
  journal={Journal of the American Statistical Association},
  volume={113},
  number={522},
  pages={868--881},
  year={2018},
  publisher={Taylor \& Francis}
}

@article{harshaw2021optimized,
  title={Optimized variance estimation under interference and complex experimental designs},
  author={Harshaw, Christopher and Middleton, Joel A and S{\"a}vje, Fredrik},
  journal={arXiv preprint arXiv:2112.01709},
  year={2021}
}

@article{dasgupta2015causal,
  author = {T. Dasgupta and N. Pillai and D. R. Rubin},
  journal = {Journal of the Royal Statistical Society: Series B},
  number = 4,
  pages = {727--753},
  title = {Causal Inference for 2$^K$ factorial designs by using potential outcomes},
  volume = 77,
  year = 2015
}

@book{fisher1935design,
  address = {London},
  author = {Fisher, R. A.},
  publisher = {Oliver \& Boyd},
  title = {The design of experiments},
  year = 1935
}

@article{hainmueller2012balancing,
  author = {Hainmueller, Jens},
  journal = {Political Analysis},
  number = 1,
  pages = {25-46},
  title = {Entropy balancing for causal effects: a multivariate reweighting method to produce balanced samples in observational studies},
  volume = 20,
  year = 2012
}

@article{holland1986statistics,
  title={Statistics and causal inference},
  author={Holland, Paul W},
  journal={Journal of the American statistical Association},
  volume={81},
  number={396},
  pages={945--960},
  year={1986},
  publisher={Taylor \& Francis}
}

@book{imbens2015causal,
  title={Causal inference in statistics, social, and biomedical sciences},
  author={Imbens, Guido W and Rubin, Donald B},
  year={2015},
  publisher={Cambridge University Press}
}

@book{kish1965survey,
  title={Survey sampling},
  author={Kish, Leslie},
year={1965},
  publisher={New York: John Wiley \& Sons.}
}

@article{morgan2012rerandomization,
  author = "Morgan, Kari Lock and Rubin, Donald B.",
  doi = "10.1214/12-AOS1008",
  journal = "Annals of Statistics",
  month = "04",
  number = "2",
  pages = "1263--1282",
  publisher = "The Institute of Mathematical Statistics",
  title = "Rerandomization to improve covariate balance in experiments",
  url = "http://dx.doi.org/10.1214/12-AOS1008",
  volume = "40",
  year = "2012"
}

@article{neyman1923application,
  author = {Neyman, J.},
  journal = {Statistical Science},
  number = 4,
  pages = {465--472},
  title = {On the application of probability theory to agricultural experiments},
  volume = 5,
  year = {1923, 1990}
}

@article{rubin1974estimating,
  title={Estimating causal effects of treatments in randomized and nonrandomized studies.},
  author={Rubin, Donald B},
  journal={Journal of Educational Psychology},
  volume={66},
  number={5},
  pages={688},
  year={1974},
  publisher={American Psychological Association}
}

@article{rubin1980randomization,
  title={Randomization analysis of experimental data: the Fisher randomization test comment},
  author={Rubin, Donald B},
  journal={Journal of the American Statistical Association},
  volume={75},
  number={371},
  pages={591--593},
  year={1980},
  publisher={JSTOR}
}

@article{samii2012equivalencies,
  title={On equivalencies between design-based and regression-based variance estimators for randomized experiments},
  author={Samii, Cyrus and Aronow, Peter M},
  journal={Statistics \& Probability Letters},
  volume={82},
  number={2},
  pages={365--370},
  year={2012},
  publisher={Elsevier}
}

\section*{Appendix}

\addcontentsline{toc}{section}{Supplementary Materials}

\setcounter{section}{0}
\def\thesection{\Alph{section}}
\renewcommand*{\thetheorem}{\Alph{section}\arabic{theorem}}
\renewcommand*{\thelemma}{\Alph{section}\arabic{lemma}}

\setcounter{table}{0}
\renewcommand{\thetable}{A\arabic{table}}
\setcounter{figure}{0}
\renewcommand\thefigure{A\arabic{figure}}
\setcounter{theorem}{0}
\renewcommand\thetheorem{A\arabic{theorem}}
\setcounter{lemma}{0}
\renewcommand\thelemma{A\arabic{lemma}}
\setcounter{equation}{0}
\renewcommand\theequation{A\arabic{equation}}
\setcounter{corollary}{0}
\renewcommand\thecorollary{A\arabic{corollary}}

\section{Additional theoretical results}
\label{appsec_additional_theory}

\subsection{Bias of standard imputation under homogeneity}
\label{appsec_bias_imputation}

In this section, we characterize the bias of the imputation-based variance estimator $\textcolor{black}{\psi_d}(\hat{\bm{c}})$ for a general design $d$, when the potential outcomes are imputed by assuming that the unit-level effects are equal to a common, deterministic value $\beta$. This characterization also reveals the design conditions required to ensure that $\textcolor{black}{\psi_d}(\hat{\bm{c}})$ is conservative. Theorem \ref{thm_impute_general} formalizes this result.
\begin{theorem} \normalfont
Let $d$ be an arbitrary design with support $\mathcal{W}$. Also, for $\bm{w} \in \mathcal{W}$, let $N_t(\bm{w})$ and $N_c(\bm{w})$ be the sizes of the treatment and control groups, respectively. 
Consider the imputation approach, where the missing potential outcomes are imputed as if the true unit-level treatment effect is $\beta$. Then, under treatment effect homogeneity, the corresponding variance estimator $\textcolor{black}{\psi_d}(\hat{\bm{c}})$ satisfies,
$$\textcolor{black}{\psi_d}(\hat{\bm{c}}) = \Var_d(\hat{\tau}) + A_1 + A_2,$$
where, 
$$A_1  = \frac{(\tau - \beta)^2}{N^2}\sum_{\bm{w}}p_{\bm{w}}\left\{\left(\sum_{i:w_i = 1}\frac{W_i}{\pi_i} - \sum_{i:w_i = 0}\frac{W_i}{1-\pi_i}\right) - \left(\sum_{i:w_i = 1}\frac{1}{\pi_i} - N \right)\right\}^2,$$
and 
\begin{align*}
  A_2 & =   2\frac{(\tau - \beta)}{N^2}\sum_{\bm{w}}p_{\bm{w}}\left\{ \left(\sum_{i:w_i = 1}\frac{Y_i(0)}{\pi_i} - \sum_{i:w_i = 0}\frac{Y_i(0)}{1-\pi_i} \right) + \tau \left( \sum_{i:w_i = 1}\frac{1-\pi_i}{\pi_i} - N_c(\bm{w}) \right) \right\} \nonumber \\
 & \hspace{3cm} \times \left\{\left(\sum_{i:w_i = 1}\frac{W_i}{\pi_i} - \sum_{i:w_i = 0}\frac{W_i}{1-\pi_i}\right) - \left(\sum_{i:w_i = 1}\frac{1-\pi_i}{\pi_i} - N_c(\bm{w}) \right)\right\}.
\end{align*} \label{thm_impute_general}
\end{theorem}

\begin{proof}
We follow the notations as in the proof of Proposition \ref{prop_impute_gen}. If treatment effect homogeneity holds, i.e., if $Y_i(1) -Y_i(0) = \tau$ for all $i$, then $c_i = Y_i(0) + (1-\pi_i)\tau$. Hence,
\begin{align}
   \Var_d(\hat{\tau}) & =   \frac{1}{N^2}\sum_{\bm{w}} p_{\bm{w}} \left\{ \left(\sum_{i:w_{i} = 1}\frac{Y_{i}(0)}{\pi_{i}} - \sum_{i:w_{i} = 0}\frac{Y_{i}(0)}{1-\pi_{i}} \right) + \tau \left(\sum_{i:w_i=1}\frac{1 - \pi_i}{\pi_i} - N_c(\bm{w}) \right) \right\}^2, \label{eq_A_impute2}
\end{align}
where $N_c(\bm{w}) = \sum_{i=1}^{n}(1-w_i)$ is the number of control units corresponding to the assignment vector $\bm{w}$. Similarly, let $N_t(\bm{w}) = \sum_{i=1}^{n}w_i$ be the number of treated units corresponding to $\bm{w}$.
Now, under the imputation approach, the imputed potential outcomes $\hat{Y}_i(0)$ and $\hat{Y}_i(1)$ can be written as
\begin{equation}
  \hat{Y}_i(0) = 
    \begin{cases}
     Y_i(1) - \beta & \text{if $W_i = 1$}\\
    Y_i(0) & \text{if $W_i = 0$},
    \end{cases}       
\end{equation}
\begin{equation}
  \hat{Y}_i(1) =
    \begin{cases}
     Y_i(1) & \text{if $W_i = 1$}\\
    Y_i(0) + \beta & \text{if $W_i = 0$}.
    \end{cases}       
\end{equation}
So, the imputed $c_i$ is given by,
\begin{align}
    \hat{c}_i & = (1-\pi_i)\hat{Y}_i(1) + \pi_i \hat{Y}_i(0) \nonumber \\
    & = W_i(Y_i(1) -\pi_i \beta) + (1-W_i)\{Y_i(0) + (1-\pi_i)\beta\} \nonumber \\
    & = Y_i(0) + W_i(\tau - \beta) + (1 - \pi_i)\beta \quad \text{(under homogeneity)}.
\end{align}
Now, for the estimated variance of $\hat{\tau}$ under the imputation approach is given by,
\begin{align}
    \textcolor{black}{\psi_d}(\hat{\bm{c}}) & = \frac{1}{N^2}\sum_{\bm{w}} p_{\bm{w}} \left(\sum_{i:w_{i} = 1}\frac{\hat{c}_{i}}{\pi_{i}} - \sum_{i:w_{i} = 0}\frac{\hat{c}_{i}}{1-\pi_{i}} \right)^2 \nonumber \\
    & = \frac{1}{N^2} \sum_{\bm{w}}p_{\bm{w}} \left\{ \left(\sum_{i:w_i = 1}\frac{Y_i(0)}{\pi_i} - \sum_{i:w_i = 0}\frac{Y_i(0)}{1-\pi_i} \right) + (\tau - \beta) \left(\sum_{i:w_i = 1}\frac{W_i}{\pi_i} - \sum_{i:w_i = 0}\frac{W_i}{1-\pi_i} \right) \right. \nonumber \\
    &+ \left. \beta \left(\sum_{i:w_i = 1}\frac{1-\pi_i}{\pi_i} - N_c(\bm{w}) \right)\right\}^2 \nonumber \\
    & = \frac{1}{N^2} \sum_{\bm{w}}p_{\bm{w}} \left\{ \left(\sum_{i:w_i = 1}\frac{Y_i(0)}{\pi_i} - \sum_{i:w_i = 0}\frac{Y_i(0)}{1-\pi_i} \right) + \tau \left( \sum_{i:w_i = 1}\frac{1-\pi_i}{\pi_i} - N_c(\bm{w}) \right) \right. \nonumber \\
 & + \left. (\tau - \beta) \left(\sum_{i:w_i = 1}\frac{W_i}{\pi_i} - \sum_{i:w_i = 0}\frac{W_i}{1-\pi_i}\right) - \left(\sum_{i:w_i = 1}\frac{1-\pi_i}{\pi_i} - N_c(\bm{w}) \right)\right\}^2 \nonumber \\
 & = \textcolor{black}{\psi_d}(\bm{c}) + \frac{(\tau - \beta)^2}{N^2}\sum_{\bm{w}}p_{\bm{w}}\left\{\left(\sum_{i:w_i = 1}\frac{W_i}{\pi_i} - \sum_{i:w_i = 0}\frac{W_i}{1-\pi_i}\right) - \left(\sum_{i:w_i = 1}\frac{1-\pi_i}{\pi_i} - N_c(\bm{w}) \right)\right\}^2 \nonumber \\
 & + 2\frac{(\tau - \beta)}{N^2}\sum_{\bm{w}}p_{\bm{w}}\left\{ \left(\sum_{i:w_i = 1}\frac{Y_i(0)}{\pi_i} - \sum_{i:w_i = 0}\frac{Y_i(0)}{1-\pi_i} \right) + \tau \left( \sum_{i:w_i = 1}\frac{1-\pi_i}{\pi_i} - N_c(\bm{w}) \right) \right\} \nonumber \\
 & \hspace{3cm} \times \left\{\left(\sum_{i:w_i = 1}\frac{W_i}{\pi_i} - \sum_{i:w_i = 0}\frac{W_i}{1-\pi_i}\right) - \left(\sum_{i:w_i = 1}\frac{1-\pi_i}{\pi_i} - N_c(\bm{w}) \right)\right\} \nonumber \\
 & = \textcolor{black}{\psi_d}(\bm{c}) + A_1 + A_2, \label{eq_A_impute3}
\end{align}
where,
\begin{align}
A_1 &=  \frac{(\tau - \beta)^2}{N^2}\sum_{\bm{w}}p_{\bm{w}}\left\{\left(\sum_{i:w_i = 1}\frac{W_i}{\pi_i} - \sum_{i:w_i = 0}\frac{W_i}{1-\pi_i}\right) - \left(\sum_{i:w_i = 1}\frac{1-\pi_i}{\pi_i} - N_c(\bm{w}) \right)\right\}^2    \nonumber \\
& = \frac{(\tau - \beta)^2}{N^2}\sum_{\bm{w}}p_{\bm{w}}\left\{\left(\sum_{i:w_i = 1}\frac{W_i}{\pi_i} - \sum_{i:w_i = 0}\frac{W_i}{1-\pi_i}\right) - \left(\sum_{i:w_i = 1}\frac{1}{\pi_i} - N) \right)\right\}^2,
\end{align}
and $A_2 = \textcolor{black}{\psi_d}(\hat{\bm{c}}) -\textcolor{black}{\psi_d}(\bm{c}) - A_1$. 
This completes the proof.
\end{proof}

Using Theorem \ref{thm_impute_general}, we can decompose the bias of $\textcolor{black}{\psi_d}(\hat{\bm{c}})$ in two terms $\E_d(A_1)$ and $\E_d(A_2)$. While the former is non-negative, the latter can take arbitrary values depending on the design. Assumption \ref{assump_fixed} provides a sufficient condition under which $\E_d(A_2)$ vanishes.
\begin{assumption}[Fixed total weight condition] \normalfont
    Fix a design $d$ with support $\mathcal{W}$. For every $\bm{w} \in \mathcal{W}$, the total inverse-probability weight satisfies $\sum_{i:w_i = 1}(1/\pi_i) + \sum_{i:w_i = 0}1/(1-\pi_i) = 2N.$ \label{assump_fixed}
\end{assumption}
Assumption \ref{assump_fixed} is satisfied if, e.g., $\pi_i = 0.5$ for all $i$. It is also satisfied for EPSEM designs with fixed (and possibly unequal) group sizes, e.g., a CRD with $N_t \neq N_c$. Now, if Assumption \ref{assump_fixed} holds, then under homogeneity, we can show that the imputation estimator is indeed conservative for $\Var_d(\hat{\tau})$.
\begin{corollary} \normalfont
Let $d$ be a design satisfying Assumption \ref{assump_fixed}. Then, the imputation estimator \textcolor{black}{based on $\beta$}, $\textcolor{black}{\psi_d}(\bm{c})$, satisfies
   $$\E_d\{\textcolor{black}{\psi_d}(\hat{\bm{c}})\} = \Var_d(\hat{\tau}) + (\tau - \beta)^2\E_d\{\textcolor{black}{\psi_d}(\bm{W})\}.$$ \label{corollary_imputation}
\end{corollary}

\begin{proof}
    
Using the notations as in the proof of Theorem \ref{thm_impute_general} we get
\begin{align}
    \E_d(A_2) & = 2\frac{(\tau - \beta)}{N^2}\sum_{\bm{w}}p_{\bm{w}}\left\{ \left(\sum_{i:w_i = 1}\frac{Y_i(0)}{\pi_i} - \sum_{i:w_i = 0}\frac{Y_i(0)}{1-\pi_i} \right) + \tau \left( \sum_{i:w_i = 1}\frac{1-\pi_i}{\pi_i} - N_c(\bm{w}) \right) \right\} \nonumber \\
 & \hspace{3cm} \times \left\{\left(N_t(\bm{w}) - \sum_{i:w_i = 0}\frac{\pi_i}{1-\pi_i}\right) - \left(\sum_{i:w_i = 1}\frac{1-\pi_i}{\pi_i} - N_c(\bm{w}) \right)\right\} \nonumber \\
 & = 2\frac{(\tau - \beta)}{N^2}\sum_{\bm{w}}p_{\bm{w}}\left\{ \left(\sum_{i:w_i = 1}\frac{Y_i(0)}{\pi_i} - \sum_{i:w_i = 0}\frac{Y_i(0)}{1-\pi_i} \right) + \tau \left( \sum_{i:w_i = 1}\frac{1-\pi_i}{\pi_i} - N_c(\bm{w}) \right) \right\} \nonumber \\
 & \hspace{3cm} \times \left\{2N - \left(\sum_{i:w_i = 1}\frac{1}{\pi_i} + \sum_{i:w_i = 0}\frac{1}{1-\pi_i} \right)\right\}.
\end{align}
By the fixed total weight condition, $\sum_{i:w_i = 1}(1/\pi_i) + \sum_{i:w_i = 0}1/(1-\pi_i) = 2N.$. This implies, $\E_d(A_2) = 0$.
Therefore,
\begin{align}
    \E_d\{\textcolor{black}{\psi_d}(\hat{\bm{c}})\} & = \textcolor{black}{\psi_d}(\bm{c}) + \E_d(A_1) \geq \E_d(A_1),
\end{align}
since $A_1 \geq 0$. Thus, the variance estimator is conservative under homogeneity. This completes the proof.
\end{proof}

Corollary \ref{corollary_imputation} directly generalizes Proposition \ref{thm_imputation1} to a class of designs satisfying Assumption \ref{assump_fixed}. As before, the upward bias in $\textcolor{black}{\psi_d}(\hat{\bm{c}})$ decreases as $\beta$ gets closer to the true treatment effect $\tau$.

Note that, even if a design satisfies Assumption \ref{assump_fixed}, the imputation estimator is not guaranteed to be conservative unless treatment effects are homogeneous across units. As discussed in Section \ref{sec_imputation_extension}, this happens because the imputed values of $\hat{\bm{c}}$ may not be unbiased for $\bm{c}$.

\subsection{On the class of linear direct imputation estimators}

In Proposition \ref{prop_altimpute}, we consider a class of linear imputation estimators of $\bm{c}$ and provide necessary and sufficient conditions under which these estimators are unbiased. 

\begin{proposition} \normalfont
Consider the following class of linear imputation estimators of $\bm{c}$.
\begin{equation}
  \hat{c}_i =
    \begin{cases}
      \alpha_i Y^{\text{obs}}_i + \zeta_i & \text{if $W_i = 1$}\\
    \tilde{\alpha}_i Y^{\text{obs}}_i + \tilde{\zeta}_i & \text{if $W_i = 0$},
    \end{cases}       
\end{equation}
where $\alpha_i,\zeta_i,\tilde{\alpha}_i,\tilde{\zeta}_i$ are constants. $\hat{c}_i$ is unbiased for $c_i$ if and only if $\alpha_i = \frac{1-\pi_i}{\pi_i}$, $\tilde{\alpha}_i = \frac{\pi_i}{1-\pi_i}$, and $\pi_i\zeta_i = -(1-\pi_i)\tilde{\zeta}_i$.
\label{prop_altimpute}
\end{proposition}

\begin{proof}
    Let $Y_i(0),Y_i(1), c_i$ be defined as before. To find an unbiased estimator of $c_i$, we consider the following class of linear estimators.
\begin{equation}
  \tilde{c}_i = 
    \begin{cases}
     \alpha_i Y^{\text{obs}}_i + \zeta_i & \text{if $W_i = 1$}\\
    \tilde{\alpha}_i Y^{\text{obs}}_i + \tilde{\zeta}_i & \text{if $W_i = 0$},
    \end{cases}       
\end{equation}
where are $\alpha_i,\zeta_i,\tilde{\alpha}_i,\tilde{\zeta}_i$ are constants to be determined. Since we want the estimator $\tilde{c}_i$ to be unbiased for $c_i$, we require
\begin{align}
    &\E(\tilde{c}_i) = c_i \nonumber \\
   &\iff \alpha_i\pi_iY_i(1) + (1-\pi_i)\tilde{\alpha}_iY_i(0) + \pi_i\zeta_i + (1-\pi_i)\tilde{\zeta}_i = (1-\pi_i)Y_i(1) + \pi_iY_i(0).
\end{align}
Comparing the coefficients on both sides, we get the following necessary conditions for unbiasedness.
\begin{enumerate}
    \item $\pi_i\zeta_i + (1-\pi_i)\tilde{\zeta}_i = 0$
    \item $\alpha_i = \frac{1-\pi_i}{\pi_i}$
    \item $\tilde{\alpha}_i = \frac{\pi_i}{1 - \pi_i}$.
\end{enumerate}
The estimator considered in Section \ref{sec_imputation_extension} sets $\alpha_i = \frac{1-\pi_i}{\pi_i}$, $\tilde{\alpha}_i = \frac{\pi_i}{1 - \pi_i}$, and $\tilde{\zeta}_i  = \pi_i\gamma_i$, and $\zeta_i = -(1-\pi_i)\gamma_i$.
\end{proof}
{\color{black}\subsection{Heuristic justification for the jackknife imputation estimator}
\label{appsec_jack}

In this section, we provide a heuristic justification for the choice of $\gamma_i$ as $\gamma_i = \hat{\theta}_{(-i)}$. 
To this end, we first note that the direct imputation approach can be conceptualized as imputing $Y_i(0)$ and $Y_i(1)$ implicitly. More concretely, suppose the missing potential outcomes are imputed as if the true unit-level effect for unit $i$ is $\beta_i$, where
    \begin{equation}
 {\beta}_i =
    \begin{cases}
      \frac{2\pi_i - 1}{\pi^2_i} Y^{\text{obs}}_i + \frac{1-\pi_i}{\pi_i}\gamma_i & \text{if $W_i = 1$}\\
    \frac{2\pi_i - 1}{(1-\pi_i)^2} Y^{\text{obs}}_i + \frac{\pi_i}{1-\pi_i}\gamma_i & \text{if $W_i = 0$}, \label{eq_betai}
    \end{cases}       
\end{equation}
It is straightforward to see that, with this choice of $\beta_i$, the resulting $\hat{c}_i$ is algebraically equivalent to that in Equation \ref{eq_ci}. As a special case, when $\pi_i = 0.5$, $\beta_i$ boils down to $\gamma_i$. Now, the expected value of this \textit{assumed} unit-level effect is $\E_d(\beta_i) = \{Y_i(1) - Y_i(0)\} + \E_d(\gamma_i - [\{(1-\pi_i)/\pi_i\}Y_i(1) - \{\pi_i/(1-\pi_i)\}Y_i(0)])$, i.e., in expectation, the assumed unit-level treatment effect equals the true unit-level effect $\{Y_i(1)-Y_i(0)\}$ and a residual term $ \E_d(\gamma_i - [\{(1-\pi_i)/\pi_i\}Y_i(1) - \{\pi_i/(1-\pi_i)\}Y_i(0)])$. Thus, one may choose $\gamma_i$ to be a reasonable estimator of $$\theta = (1/N)\sum_{i=1}^{N}\left[\{(1-\pi_i)/\pi_i\}Y_i(1) - \{\pi_i/(1-\pi_i)\}Y_i(0)\right].$$
Notice that, when $\pi_i = 0.5,$ $\theta = \tau$. A natural estimator of $\theta$ is the Horvitz-Thompson estimator 
$$\hat{\theta} = (1/N)\sum_{j}W_jY^{\text{obs}}_j(1-\pi_j)/{\pi^2_j} - (1/N)\sum_{j}{(1-W_j)Y^{\text{obs}}_j\pi_j}/{(1-\pi_j)^2}.$$ 
To ensure unbiasedness of $\hat{c}_i$, we set $\gamma_i$ to be a leave-one-out version of $\hat{\theta}$ that excludes unit $i$, i.e., $\gamma_i  = \hat{\theta}_{(-i)}$, where
\begin{equation}
    \hat{\theta}_{(-i)} = \frac{1}{|\mathcal{S}_1(i)|}\sum_{j\in \mathcal{S}_1(i)}\frac{W_jY^{\text{obs}}_j(1-\pi_j)}{\tilde{\pi}_j \pi_j} - \frac{1}{|\mathcal{S}_0(i)|}\sum_{j\in\mathcal{S}_0(i)}\frac{(1-W_j)Y^{\text{obs}}_j {\pi}_j}{(1-\tilde{\pi}_j)(1-\pi_j)}. 
, \label{eq_app_jack3}
\end{equation}
, $\tilde{\pi}_j = \Pr_d(W_j = 1|W_i = 1)$ if $W_i = 1$, and $\tilde{\pi}_j = \Pr_d(W_j=1|W_i = 0)$ if $W_i = 0$. Here, the summations are taken over the subsets of units $\mathcal{S}_1(i)$ and $\mathcal{S}_0(i)$ to ensure that the denominators are nonzero.

}
\subsection{Alternative jackknife imputation estimators}
\label{appsec_alternative}

An alternative Jackknife estimator is given by $\gamma_i = \hat{\tau}_{(-i)}$ where $\hat{\tau}_{(-i)}$ is the Horvitz-Thompson estimator, leaving out the $i$th unit, i.e., 
    \begin{equation}
  \hat{\tau}_{(-i)} = \frac{1}{N-1}\sum_{j \neq i}\frac{W_jY^{\text{obs}}_j}{\tilde{\pi}_{j}} - \frac{1}{N-1}\sum_{j \neq i}\frac{(1-W_j)Y^{\text{obs}}_j}{1-\tilde{\pi}_{j}},
    \label{eq_jack2}
\end{equation}
where,
\begin{equation}
     \tilde{\pi}_j =
    \begin{cases}
      \Pr_d(W_j = 1|W_i = 1) & \text{if $W_i = 1$}\\
    \Pr_d(W_j=1|W_i = 0) & \text{if $W_i = 0$}.
    \end{cases}  
\end{equation}
When $\pi = 0.5,$ this estimator is equivalent to the Jackknifed estimator $\hat{\theta}_{(-i)}$ in Section \ref{sec_imputation_extension}. Moreover, under a CRD, this estimator boils down to the Jackknifed difference-in-means statistic.

Proposition \ref{prop_jack_general} shows that, conditional on $W_i = w\in \{0,1\}$, the Jackknife estimator $\hat{\tau}_{(-i)}$ is unbiased for the average treatment effect in a population of $N-1$ units that excludes unit $i$. 
\begin{proposition}\normalfont
Let $d$ be an arbitrary design with $\pi_i \in (0,1)$. Consider a direct imputation estimator $\hat{c}_i$ with $\gamma_i = \hat{\tau}_{(-i)}$, where $\hat{\tau}_{(-i)}$ is as defined in Equation \ref{eq_jack2}. It follows that,
$$\E_d(\hat{\tau}_{(-i)}|W_i = 1) = \E_d(\hat{\tau}_{(-i)}|W_i = 0) = \frac{1}{N-1}\sum_{j \neq i}(Y_j(1) - Y_j(0)),$$
and hence $\E_d(\hat{c}_i) = c_i.$
    \label{prop_jack_general}    
\end{proposition}
\begin{proof}
    
\begin{align}
\E_d(\hat{\tau}_{(-i)}|W_i = 1) &= \E_d \left\{\frac{1}{N-1}\sum_{j \neq i}\frac{W_jY_j(1)}{\Pr_d(W_j = 1|W_i=1)} - \frac{1}{N-1}\sum_{j \neq i}\frac{(1-W_j)Y_j(0)}{\Pr_d(W_j = 0|W_i=1)}  \Bigg|W_i = 1 \right\}   \nonumber\\
& = \frac{1}{N-1}\sum_{j \neq i}(Y_j(1) - Y_j(0)).
\end{align}
Similarly, we can show that
\begin{align}
\E_d(\hat{\tau}_{(-i)}|W_i = 0) = \frac{1}{N-1}\sum_{j \neq i}(Y_j(1) - Y_j(0)).   
\end{align}
This completes the proof.
\end{proof}
Thus, setting $\gamma_i$ as the Jackknifed version of the Horvitz-Thompson estimator, we can obtain a conservative variance estimator for an arbitrary design.
{\color{black}\subsection{Additional insights on the imputation approach for standard designs}
\label{sec_app_imp_standard}
In this section, we first revisit the setup of Section~\ref{sec_imputation_general} and examine the imputation approach based on a deterministic value of $\beta$. We compare the resulting variance estimator with commonly used variance estimators for standard designs, such as $\hat{V}_{\text{Neyman}}$ for completely randomized designs and $\hat{V}_{\text{pair}}$ for matched-pair designs.

Proposition~\ref{prop_app_imp_crd} characterizes the conditions under which the imputation-based variance estimator is less conservative than the Neymanian variance estimator in the case of completely randomized designs.

\begin{proposition} \normalfont
Consider a CRD with equal group sizes. The bias of the imputation estimator based on a deterministic $\beta$ is no more than that of $\hat{V}_{\text{Neyman}}$ if and only if
$$\tau -\sqrt{\frac{S^2_{\tau}}{N}} \leq \beta \leq \tau + \sqrt{\frac{S^2_{\tau}}{N}}.$$
    \label{prop_app_imp_crd}
\end{proposition}
\begin{proof}
   The imputed potential outcomes are given by
   \begin{equation}
    \hat{Y}_i(1) = W_iY_i(1) + (1-W_i)(Y_i(0) + \beta), \hspace{0.5cm} \hat{Y}_i(0) = W_i (Y_i(1) - \beta) + (1-W_i) Y_i(0). 
\end{equation}
 Denote $\overline{\hat{Y}}(1)$ and $\overline{\hat{Y}}(0)$ be the averages of the imputed potential outcomes over the population. We note that,
\begin{align}
  \overline{\hat{Y}}(1) = \frac{\bar{Y}_t(1) + \bar{Y}_t(0)}{2} + \frac{\beta}{2},  
\end{align} 
and 
\begin{align}
  \overline{\hat{Y}}(0) = \frac{\bar{Y}_t(1) + \bar{Y}_t(0)}{2} - \frac{\beta}{2},  
\end{align} 
where $\bar{Y}_t(1) = \frac{1}{N/2}\sum_{i:W_i = 1}Y_i(1)$ and $\bar{Y}_t(0) = \frac{1}{N/2}\sum_{i:W_i = 0}Y_i(0)$. The Horvitz-Thompson estimator $\hat{\tau}$ can be written as $\hat{\tau} = \bar{Y}_t(1) - \bar{Y}_t(0)$.
The imputation-based variance estimator is given by
\begin{align}
&    \hat{V}_{\beta}\nonumber\\
& = \frac{1}{(N/2)(N-1)}\sum_{i=1}^{N}\{\hat{Y}_i(1) -\overline{\hat{Y}}(1)\}^2 + \frac{1}{(N/2)(N-1)}\sum_{i=1}^{N}\{\hat{Y}_i(0) -\overline{\hat{Y}}(0)\}^2 \nonumber\\
    & \quad - \frac{1}{N(N-1)}\sum_{i=1}^{N}[\{\hat{Y}_i(1)-\bar{Y}_i(0)\} -\{\overline{\hat{Y}}(1) - \overline{\hat{Y}}(0)\}]^2 \nonumber\\
    & = \frac{1}{(N/2)(N-1)} \left[ \sum_{i:W_i = 1}\{{Y}_i(1) -\overline{\hat{Y}}(1)\}^2 + \sum_{i:W_i = 0}\{{Y}_i(0) + \beta - \overline{\hat{Y}}(1)\}^2  \right. \nonumber\\
    & \quad \quad \quad \left. + \sum_{i:W_i =1}^{N}\{{Y}_i(1) - \beta - \overline{\hat{Y}}(0)\}^2 + \sum_{i:W_i =0}^{N}\{{Y}_i(0) - \overline{\hat{Y}}(0)\}^2 \right] \nonumber\\
    & = \frac{1}{(N/2)(N-1)} \left[ 2(N/2)(N/2-1)\hat{V}_{\text{Neyman}} + (N/2)\{\bar{Y}_t(1) -\overline{\hat{Y}}(1)\}^2 + (N/2)\{\bar{Y}_t(0) -\overline{\hat{Y}}(0)\}^2  \right. \nonumber \\
    & \quad \quad \quad \left. (N/2)\{\bar{Y}_t(0) +\beta -\overline{\hat{Y}}(1)\}^2 + (N/2)\{\bar{Y}_t(1) -\beta -\overline{\hat{Y}}(0)\}^2  \right] \nonumber \\
    & = \frac{N-2}{N-1}\hat{V}_{\text{Neyman}} +  \frac{1}{(N/2)(N-1)} (N/2) (\hat{\tau} - \beta)^2 = \frac{N-2}{N-1}\hat{V}_{\text{Neyman}} + \frac{1}{N-1}(\hat{\tau} - \beta)^2. \label{eq_app_imp1}
 \end{align}
Now, the bias of Neyman's variance estimator is given by
\begin{align}
    \E_d(\hat{V}_{\text{Neyman}}) -\Var_d(\hat{\tau}) = \frac{S^2_{\tau}}{N}. 
\end{align}
The bias of the imputation estimator is given by
\begin{align}
\E_d\{\hat{V}_{\beta}\} -\Var_d(\hat{\tau})  = \frac{N-2}{N-1}\frac{S^2_{\tau}}{N} + \frac{1}{N-1}(\beta - \tau)^2. 
\end{align}
Therefore, the imputation estimator has a bias no more than the Neymanian estimator iff
\begin{align}
 &  \frac{N-2}{N-1}\frac{S^2_{\tau}}{N} + \frac{1}{N-1}(\beta - \tau)^2 \leq  \frac{S^2_{\tau}}{N} \nonumber \\
  & \iff (\beta - \tau)^2 \leq \frac{S^2_{\tau}}{N} \nonumber\\
  & \iff \tau -\sqrt{\frac{S^2_{\tau}}{N}} \leq \beta \leq \tau + \sqrt{\frac{S^2_{\tau}}{N}}.
\end{align}
 
\end{proof}
Next, we consider the case of a matched pair design. We label the two units in each pair $j$ as $A$ and $B$, with potential outcomes $Y_{jA}(w)$ and $Y_{jB}(w)$, respectively, $w\in \{0,1\}$. Also, let $W_j$ be the indicator that unit A receives treatment in pair $j$. Finally, let $Y^{\text{obs}}_{jt}$ and $Y^{\text{obs}}_{jc}$ denote the observed outcomes of the treated and control units in pair $j$, respectively. For this design, the Horvitz-Thompson estimator $\hat{\tau}$ is same as the difference in means estimator $\bar{Y}_t - \bar{Y}_c$. We consider the standard variance estimator $\hat{V}_{\text{pair}} = \frac{4}{N(N-2)}\sum_{j=1}^{N/2}\{(Y^{\text{obs}}_{jt}-Y^{\text{obs}}_{jc})- (\bar{Y}_t - \bar{Y}_c)\}^2$.

Proposition~\ref{prop_app_imp_match} characterizes the conditions under which the imputation-based variance estimator is less conservative than $\hat{V}_{\text{pair}}$.

\begin{proposition}\normalfont
Consider a matched pair design with $N$ units. The bias of the imputation estimator based on a deterministic $\beta$ is no more than that of $\hat{V}_{\text{pair}}$ if and only if
$$\tau -\sqrt{R} \leq \beta \leq \tau + \sqrt{R},$$
where $$R = \frac{4}{N(N-2)}\sum_{j=1}^{N/2}\left\{\frac{(Y_{jA}(1) - Y_{jA}(0)) +(Y_{jB}(1) - Y_{jB}(0))}{2} - \tau\right\}^2.$$
 \label{prop_app_imp_match}
\end{proposition}
\begin{proof}
Under a matched pair design, 
\begin{align}
    \Var_d(\hat{\tau})&= \frac{1}{N^2}\sum_{j=1}^{N/2}\{(Y_{jA}(1) + Y_{jA}(0)) - (Y_{jB}(1) + Y_{jB}(0))\}^2
\end{align}
\end{proof}
See \cite{imbens2015causal}, Chapter 10, for a proof. Thus, the imputation-based variance estimator can be written as:
\begin{align}
    \hat{V}_{\beta} &= \frac{1}{N^2}\sum_{j=1}^{N/2}[W_j\{(Y_{jA}(1) + Y_{jA}(1)-\beta) - (Y_{jB}(0) + \beta + Y_{jB}(0))\} \nonumber\\
    & \quad \quad + (1-W_j)\{(Y_{jA}(0)+\beta + Y_{jA}(0)) - (Y_{jB}(1) + Y_{jB}(1)-\beta)\}]^2 \nonumber\\
    & = \frac{4}{N^2}\sum_{j=1}^{N/2}\{W_j(Y_{jA}(1)-Y_{jB}(0)) + (1-W_i)(Y_{jA}(0)-Y_{jB}(1)) - (2W_i-1)\beta\}^2 \nonumber\\
    & = \frac{4}{N^2}\sum_{j=1}^{N/2}(Y^{\text{obs}}_{jt}-Y^{\text{obs}}_{jc}-\beta)^2 \nonumber\\
    & = \frac{4}{N^2}\left[\sum_{j=1}^{N/2}\{(Y^{\text{obs}}_{jt}-Y^{\text{obs}}_{jc}) - (\bar{Y}_t - \bar{Y}_c)\}^2 + \frac{N}{2}(\hat{\tau} - \beta)^2\right]\nonumber\\
    & = \frac{N-2}{N}\hat{V}_{\text{pair}} + \frac{2}{N}(\hat{\tau} - \beta)^2.
\end{align}
Now, the bias of $\hat{V}_{\text{pair}}$ is 
\begin{align}
    \E_d(\hat{V}_{\text{pair}}) - \Var_d(\hat{\tau}) = \frac{4}{N(N-2)}\sum_{j=1}^{N/2}\left\{\frac{(Y_{jA}(1) - Y_{jA}(0)) +(Y_{jB}(1) - Y_{jB}(0))}{2} - \tau\right\}^2 = R.
\end{align}
See \cite{imbens2015causal}, Chapter 10, for a proof. The bias of $\hat{V}_{\beta}$ is
\begin{align}
    \E_d(\hat{V}_{\beta}) - \Var_d(\hat{\tau}) &= \frac{N-2}{N}R + \frac{2}{N}(\tau - \beta)^2
\end{align}
Therefore, the imputation approach has a bias no more than $\hat{V}_{\text{pair}}$ iff
\begin{align}
&    \frac{N-2}{N}R + \frac{2}{N}(\tau - \beta)^2 \leq R \nonumber\\
& \iff (\beta - \tau)^2 \leq R \nonumber\\
& \tau -\sqrt{R} \leq \beta \leq \tau + \sqrt{R}.
\end{align}
This completes the proof.

So far, our discussion has focused on the imputation approach with a deterministic $\beta$. We now evaluate the performance of the jackknife-based imputation estimator relative to the standard variance estimators. To this end, we first consider two completely randomized designs, one with $N = 12$ and the other with $N = 40.$ We compare the relative bias and the standard deviation of the Neymanian variance estimator ($\hat{V}_{\text{Neyman}}$), the imputation estimator with $\beta =0$ ($\hat{V}_0$), and the jackknife imputation estimator $\hat{V}_{\text{Jack}}$. 
We consider three different generative models for the potential outcomes, as shown below.
 \begin{enumerate}
        \item Constant effect = 5: $Y_i(0) \overset{iid}{\sim} \text{Unif}(0,10)$, $Y_i(1) = Y_i(0) + 5$ for all $i$.

        \item Heterogeneous effects, average = 0: $Y_i(0) \overset{iid}{\sim} \text{Unif}(0,10)$, $Y_i(1)\overset{iid}{\sim} \text{Unif}(0,10)$ for all $i$, and $\bm{Y}(1)$s and $\bm{Y}(0)$s are independent.

        \item Heterogeneous effects, average = 5: $Y_i(0) \overset{iid}{\sim} \text{Unif}(0,10)$, $Y_i(1)\overset{iid}{\sim} \text{Unif}(5,15)$ for all $i$, and $\bm{Y}(1)$s and $\bm{Y}(0)$s are independent. 
    \end{enumerate}
The results are shown in Figures \ref{fig:app_crd_12} and \ref{fig:app_crd_40}.

We observe from both the figures that $\hat{V}_{\text{Neyman}}$ and $\hat{V}_{\text{Jack}}$ perform similarly across all scenarios, with $\hat{V}_{\text{Neyman}}$ exhibiting slightly smaller bias and standard deviation than $\hat{V}_{\text{Jack}}$ for the smaller population size. By construction, $\hat{V}_{\text{Neyman}}$ is unbiased when the treatment effect is constant, whereas $\hat{V}_{\text{Jack}}$ is slightly conservative. 
$\hat{V}_0$ is more conservative than both the other two estimators when the treatment effect equals five (either identically, or on average), which is expected since the imputation in this case is based on an incorrect postulated treatment effect (i.e., zero).
\begin{figure}[H]
    \centering
    \includegraphics[width=0.9\linewidth]{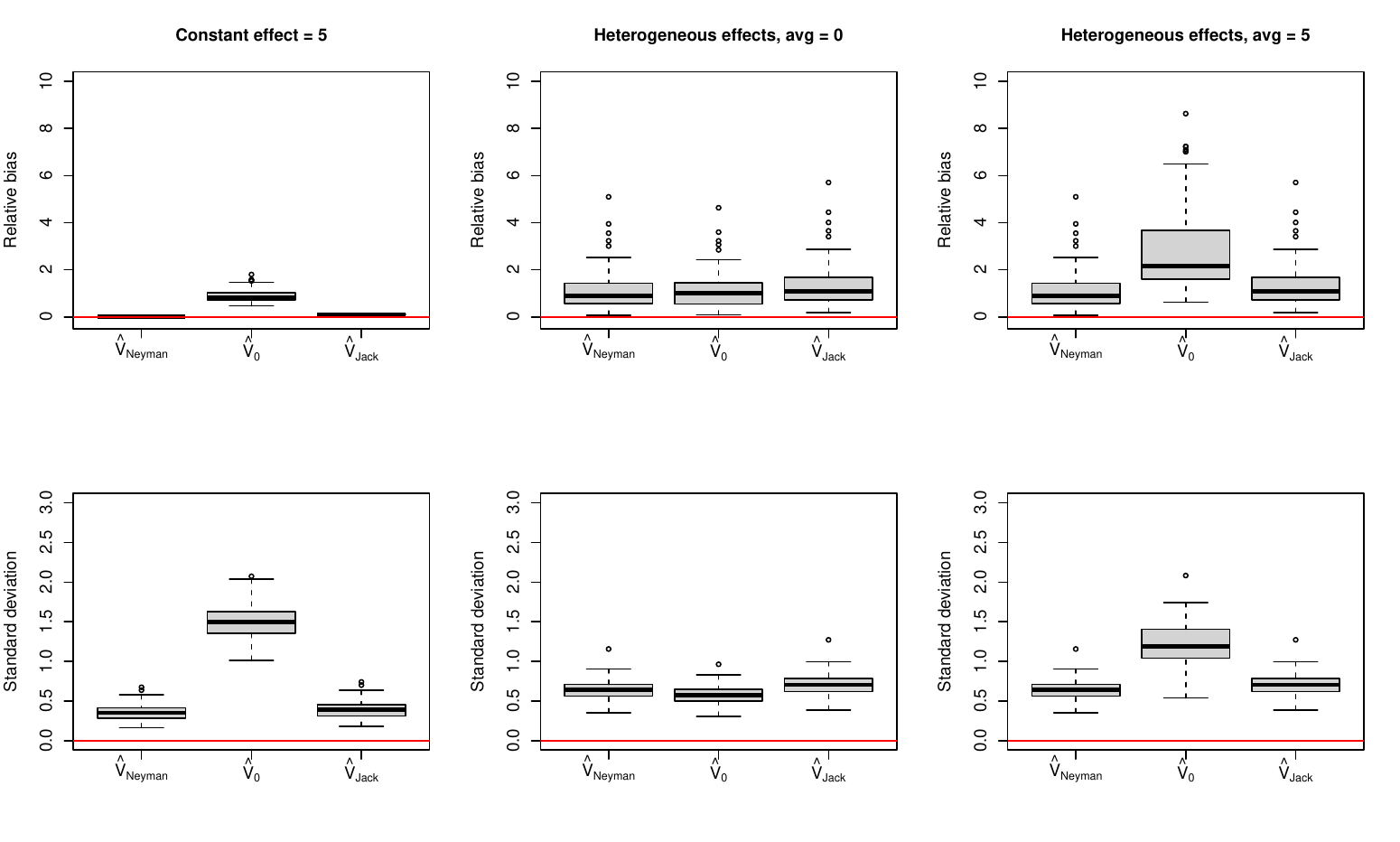}
    \caption{\textcolor{black}{Relative biases and standard errors of variance estimators under a CRD with $N = 12$ and equal group sizes.}}
    \label{fig:app_crd_12}
\end{figure}
\begin{figure}[H]
    \centering
    \includegraphics[width=0.9\linewidth]{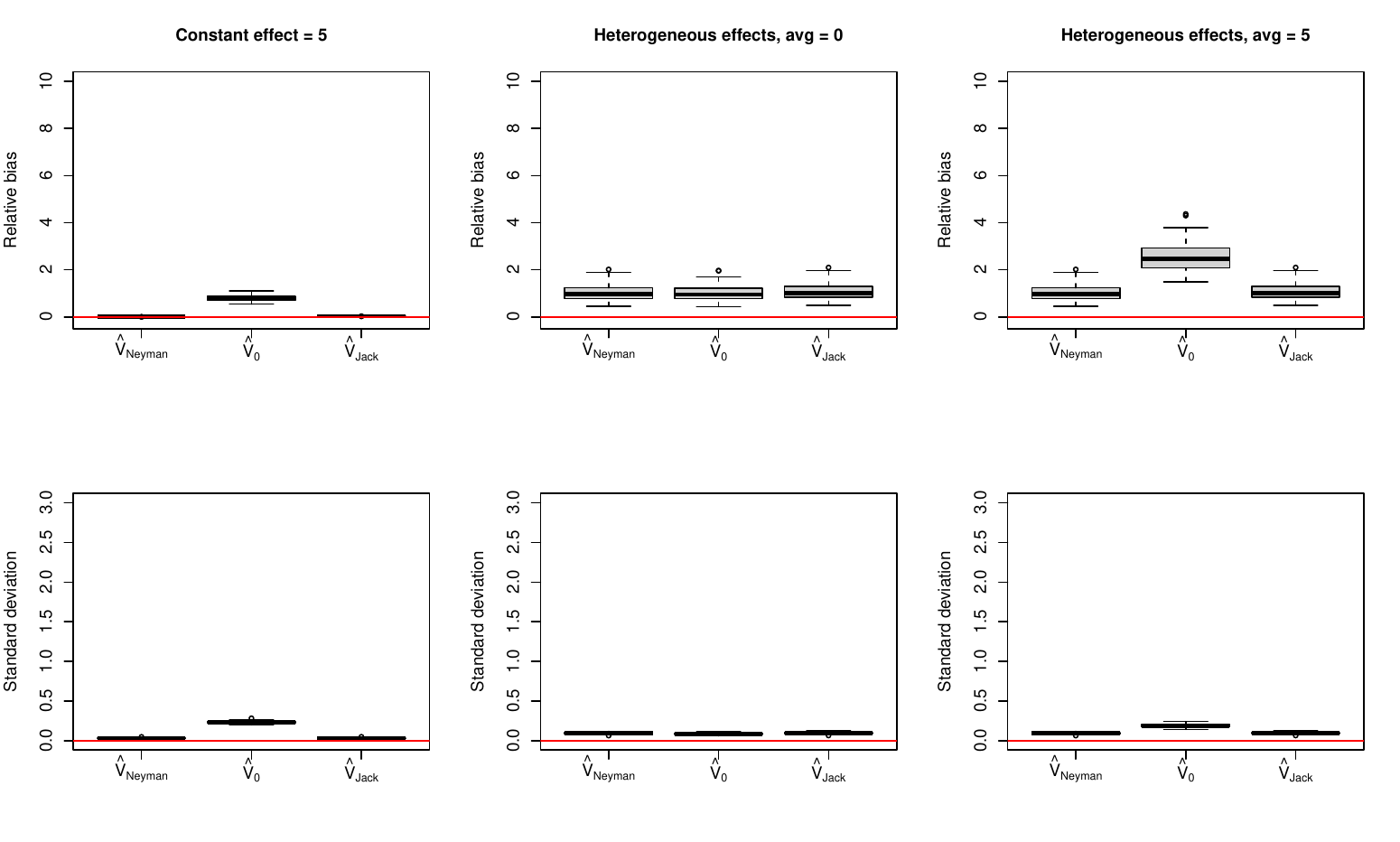}
    \caption{\textcolor{black}{Relative biases and standard errors of variance estimators under a CRD with $N = 40$ and equal group sizes.}}
    \label{fig:app_crd_40}
\end{figure}
We perform a similar comparison for two matched-pair designs, one with $N = 12$, and the other with $N = 40$, with three variance estimators, namely $\hat{V}_{\text{pair}}$, $\hat{V}_0$, and $\hat{V}_{\text{Jack}}$. The results are shown in Figures \ref{fig:app_match_12} and \ref{fig:app_match_40}.
\begin{figure}[H]
    \centering
    \includegraphics[width=0.9\linewidth]{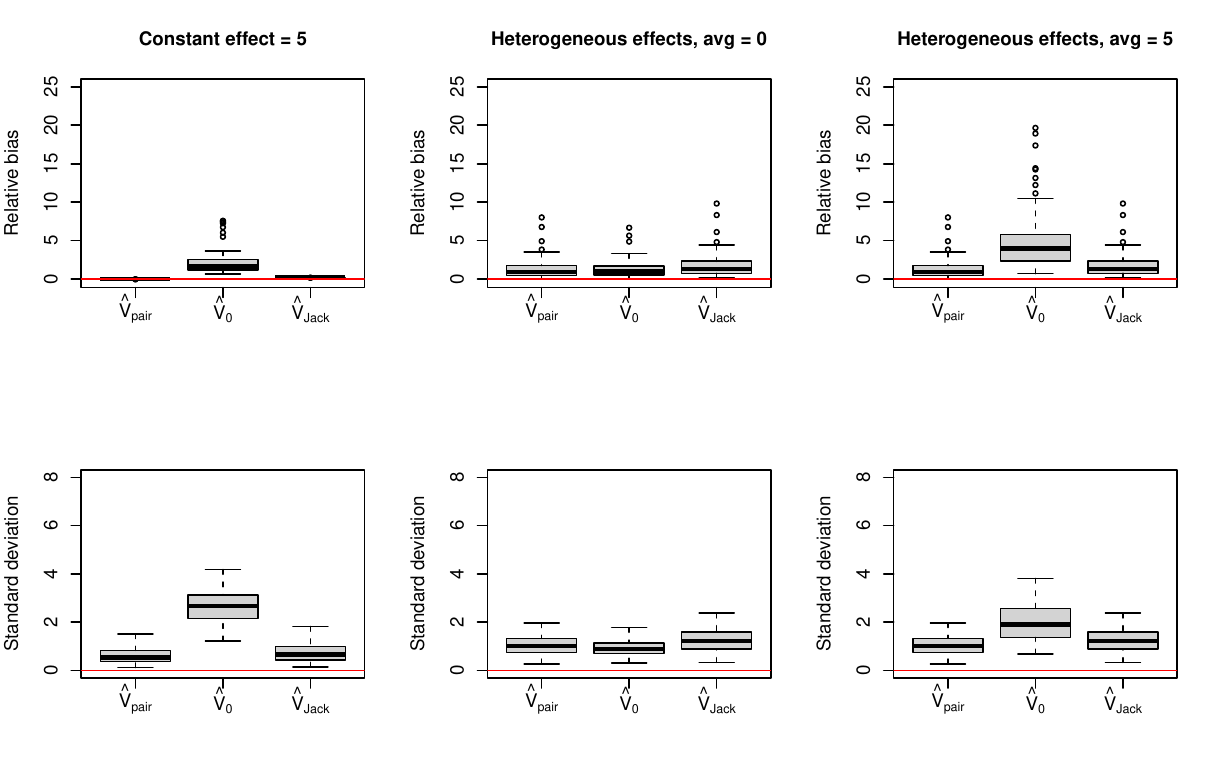}
    \caption{\textcolor{black}{Relative biases and standard errors of variance estimators under a matched-pair design with $N = 12$.}}
    \label{fig:app_match_12}
\end{figure}
\begin{figure}[H]
    \centering
    \includegraphics[width=0.9\linewidth]{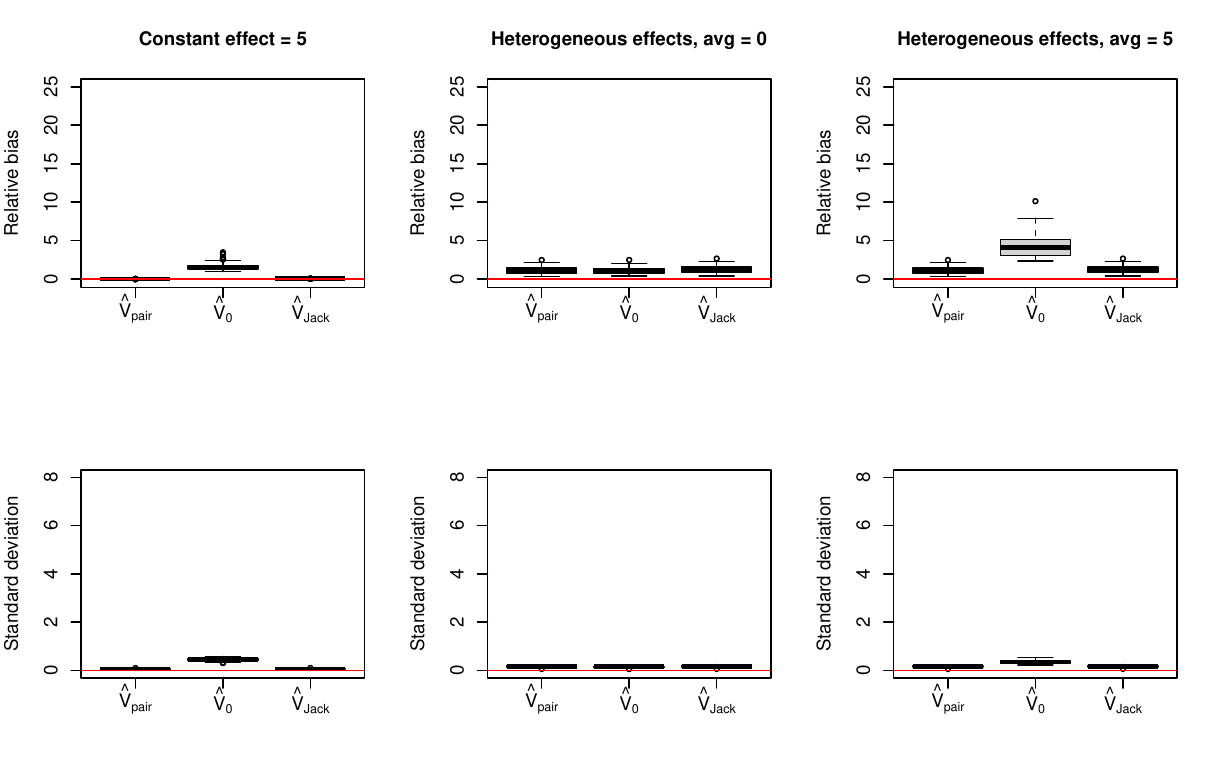}
    \caption{\textcolor{black}{Relative biases and standard errors of variance estimators under a matched-pair design with $N = 40$.}}
    \label{fig:app_match_40}
\end{figure}
The figures indicate that, similar to the CRD case, $\hat{V}_{\text{pair}}$ and $\hat{V}_{\text{Jack}}$ perform comparably across all scenarios, with $\hat{V}_{\text{pair}}$ having slightly smaller relative bias and standard deviation than $\hat{V}_{\text{Jack}}$. In addition, both measures tend to be smaller for all the designs, as compared to the $N = 12$ case.}
{\color{black}
\subsection{Substitution condition for stratified and clustered designs}
\label{sec_substitution_stratified}

We first consider a stratified randomized design with $H$ equal sized strata, where the $h$th stratum contains $M$ units; $M\times H = N$. Assume that both $H$ and $M$ are even, and there is equal allocation per stratum. Within stratum $h$, we label the units as $(1,h),...,(M,h)$. 

Now, consider a fixed $\bm{w} \in \mathcal{W}$. Without loss of generality, assume that under $\bm{w}$, units $(1,h),(2,h),...,(\frac{M}{2},h)$ (in stratum $h$) are assigned to treatment, and units $(\frac{M}{2}+1,h),(\frac{M}{2}+2,h),...,(M,h)$ are assigned to control. 

The substitution condition requires the existence of another assignment vector $\tilde{\bm{w}} \in \mathcal{W}$ that selects half of the $\frac{N}{2}$ treated units and half of the $\frac{N}{2}$ control units under $\bm{w}$ and assigns them to the treatment group. In this case, we construct such a $\tilde{\bm{w}}$ as follows.
For strata $h \in \{1, 2, \ldots, \frac{H}{2}\}$, we select the control units (under $\bm{w}$) $(\frac{M}{2}+1, h), (\frac{M}{2}+2, h), \ldots, (M, h)$, yielding $\frac{N}{4}$ units in total.  
For strata $h \in \{\frac{H}{2}+1, \frac{H}{2}+2, \ldots, H\}$, we select the treated units (under $\bm{w}$) $(1, h), (2, h), \ldots, (\frac{M}{2}, h)$, also yielding $\frac{N}{4}$ units.  
These $\frac{N}{2}$ units are assigned to treatment. It follows that the resulting assignment vector $\tilde{\bm{w}}$ belongs to the support $\mathcal{W}$, and hence, the substitution condition is satisfied.

Next, we consider a cluster randomized design with $K$ clusters, each of size $m$, where $K$ is a multiple of four, and $N = Km$. We label the clusters as $1,2,...,K$ and assume that half of the clusters receive treatment. 

Now, consider a fixed $\bm{w}\in \mathcal{W}$. Without loss of generality, assume that under $\bm{w}$, clusters $1,2,...,\frac{K}{2}$ receive treatment, and clusters $\frac{K}{2}+1,\frac{K}{2}+2,...,K$ receive control. 
As before, the substitution condition requires the existence of another assignment vector $\tilde{\bm{w}} \in \mathcal{W}$ that selects half of the $\frac{N}{2}$ treated units and half of the $\frac{N}{2}$ control units under $\bm{w}$ and assigns them to the treatment group. In this case, we construct such a $\tilde{\bm{w}}$ as follows.
We select clusters $1,2,...,\frac{K}{4}$ (which are treated clusters under $\bm{w}$), yielding $\frac{N}{4}$ units in total, and clusters $\frac{K}{2}+1, \frac{K}{2}+2,...,\frac{3K}{4}$ (which are control clusters under $\bm{w}$), yield $\frac{N}{4}$ units. These combined set of 
$\frac{K}{2}$ clusters are assigned to treatment. Equivalently, the $\frac{N}{2}$ units in these clusters are assigned to treatment. 
By construction, the resulting assignment vector $\tilde{\bm{w}}$ belongs to the support $\mathcal{W}$, and hence, the substitution condition is satisfied.

Finally, we note that for the two classes of designs discussed above, Assumptions~\ref{assump_equalsize} and~\ref{assump_closed} are automatically satisfied, which allows the contrast approach to be implemented.

\subsection{On the bias of \texorpdfstring{$\hat{V}_{\text{AM}}$}{\^V}}

For simplicity, let us denote $a_i = Y_i(0)$ and $b_i = Y_i(1).$ For a general design $d$, Proposition \ref{prop_app_am} characterizes the bias of the variance estimator $\hat{V}_{\text{AM}}$.
\begin{proposition}\normalfont
The bias of $\hat{V}_{\text{AM}}$ under a general design $d$ is given by
\begin{align*}
    \E_d(\hat{V}_{\text{AM}}) - \Var_d(\hat{\tau}) & = \frac{1}{N^2}\left\{\sum_{i=1}^{N}(b_i - a_i)^2 + \frac{1}{2}\mathop{\sum\sum}_{i \neq j: p_{ij}(1,1) = 0}(b_i+b_{j})^2 + \frac{1}{2}\mathop{\sum\sum}_{i \neq j: p_{ij}(0,0) = 0}(a_i+a_{j})^2 \right. \\
    & \quad\quad \left. + \mathop{\sum\sum}_{i \neq j: p_{ij}(1,0) = 0}(b_i - a_{j})^2 \right\}.
\end{align*}
    \label{prop_app_am}
\end{proposition}
\begin{proof}
The variance of $\hat{\tau}$ is given by,
\begin{align}
    \Var_d(\hat{\tau}) = \frac{1}{N^2}(V_1 + V_2 - 2V_3),
\end{align}
where 
\begin{align*}
    V_1 &=\sum_{i=1}^{N}\frac{1-\pi_i}{\pi_i}b^2_i + \mathop{\sum\sum}_{i\neq j}b_ib_j \frac{p_{ij}(1,1) - \pi_i \pi_j}{\pi_i \pi_j},
\end{align*}
\begin{align*}
    V_1 &=\sum_{i=1}^{N}\frac{\pi_i}{1-\pi_i}a^2_i + \mathop{\sum\sum}_{i\neq j}a_ia_j \frac{p_{ij}(0,0) - (1-\pi_i)(1- \pi_j)}{(1-\pi_i)(1- \pi_j)},
\end{align*}
and
\begin{align*}
    V_3 &= - \sum_{i=1}^{N}a_ib_i +  \mathop{\sum\sum}_{i\neq j}b_ia_j \frac{p_{ij}(1,0) - \pi_i(1- \pi_j)}{\pi_i(1- \pi_j)}.
\end{align*}
Let us denote $\tilde{W}_i = 1-W_i$. The variance estimator $\hat{V}_{\text{AM}}$ is given by
\begin{align}
    \hat{V} _{\text{AM}} = \frac{1}{N^2}(\hat{V}_1 + \hat{V}_2 - 2\hat{V}_3),
\end{align}
where 
\begin{align*}
    \hat{V}_1 &=\sum_{i=1}^{N}W_i\frac{1-\pi_i}{\pi^2_i}(Y^{\text{obs}}_i)^2 + \mathop{\sum\sum}_{\substack{i\neq j:\\ p_{ij}(1,1)>0}}W_iW_jY^{\text{obs}}_i Y^{\text{obs}}_j \frac{p_{ij}(1,1) - \pi_i \pi_j}{p_{ij}(1,1)\pi_i \pi_j}\\
    & \quad + \frac{1}{2}\mathop{\sum\sum}_{\substack{i\neq j:\\ p_{ij}(1,1)=0}}\left\{\frac{W_i(Y^{\text{obs}}_i)^2}{\pi_i} + \frac{W_j(Y^{\text{obs}}_j)^2}{\pi_j}\right\}, 
\end{align*}
\begin{align*}
    \hat{V}_2 &=\sum_{i=1}^{N}\tilde{W}_i\frac{\pi_i}{(1-\pi_i)^2}(Y^{\text{obs}}_i)^2 + \mathop{\sum\sum}_{\substack{i\neq j:\\ p_{ij}(0,0)>0}}\tilde{W}_i\tilde{W}_j Y^{\text{obs}}_i Y^{\text{obs}}_j \frac{p_{ij}(0,0) - (1-\pi_i) (1-\pi_j)}{p_{ij}(0,0)(1-\pi_i) (1-\pi_j)} \\
    & \quad + \frac{1}{2}\mathop{\sum\sum}_{\substack{i\neq j:\\ p_{ij}(0,0)=0}}\left\{\frac{\tilde{W}_i(Y^{\text{obs}}_i)^2}{1-\pi_i} + \frac{\tilde{W}_j(Y^{\text{obs}}_j)^2}{1-\pi_j}\right\}, 
\end{align*}
and
\begin{align*}
    \hat{V}_3 &= -\frac{1}{2}\left(\sum_{i=1}^{N}\frac{{W}_i}{\pi_i}(Y^{\text{obs}}_i)^2 + \sum_{i=1}^{N}\frac{\tilde{W}_i}{1-\pi_i}(Y^{\text{obs}}_i)^2\right) + \mathop{\sum\sum}_{\substack{i\neq j:\\ p_{ij}(1,0)>0}}{W}_i\tilde{W}_j Y^{\text{obs}}_i Y^{\text{obs}}_j \frac{p_{ij}(1,0) - \pi_i (1-\pi_j)}{p_{ij}(1,0)\pi_i(1-\pi_j)} \\
    & \quad - \frac{1}{2}\mathop{\sum\sum}_{\substack{i\neq j:\\ p_{ij}(1,0)=0}}\left\{\frac{{W}_i(Y^{\text{obs}}_i)^2}{\pi_i} + \frac{\tilde{W}_j(Y^{\text{obs}}_j)^2}{1-\pi_j}\right\}, 
\end{align*}
Therefore,
\begin{align}
    \E_d(\hat{V}_1) - V_1 &= \frac{1}{2}\mathop{\sum\sum}_{\substack{i\neq j:\\ p_{ij}(1,1)=0}}(b^2_i + b^2_j) -\mathop{\sum\sum}_{\substack{i\neq j:\\ p_{ij}(1,1)=0}}b_ib_j = \frac{1}{2}\mathop{\sum\sum}_{\substack{i\neq j:\\ p_{ij}(1,1)=0}}(b_i+b_j)^2.
\end{align}
Similarly,
\begin{align}
    \E_d(\hat{V}_2) - V_2 &= \frac{1}{2}\mathop{\sum\sum}_{\substack{i\neq j:\\ p_{ij}(0,0)=0}}(a_i+a_j)^2.
\end{align}
Finally, 
\begin{align}
    \E_d(\hat{V}_3) - V_3 &= - \frac{1}{2}\sum_{i=1}^{N}(b^2_i + a^2_i) +\sum_{i=1}^{N}a_ib_i - \frac{1}{2}\mathop{\sum\sum}_{\substack{i\neq j:\\ p_{ij}(1,0)=0}}(b^2_i+a^2_j) + \mathop{\sum\sum}_{\substack{i\neq j:\\ p_{ij}(1,0)=0}}b_ia_j \nonumber\\
    & = -\frac{1}{2}\left\{\sum_{i=1}^{N}(b_i-a_i)^2 + \mathop{\sum\sum}_{\substack{i\neq j:\\ p_{ij}(1,0)=0}} (b_i-a_j)^2  \right\}.
\end{align}
Putting them all together, we get
\begin{align}
    \E_d(\hat{V}_{\text{AM}}) - \Var_d(\hat{\tau}) &= \frac{1}{N^2}\left\{ \frac{1}{2}\mathop{\sum\sum}_{\substack{i\neq j:\\ p_{ij}(1,1)=0}}(b_i+b_j)^2 + \frac{1}{2}\mathop{\sum\sum}_{\substack{i\neq j:\\ p_{ij}(0,0)=0}}(a_i+a_j)^2 + \sum_{i=1}^{N}(b_i-a_i)^2 + \mathop{\sum\sum}_{\substack{i\neq j:\\ p_{ij}(1,0)=0}} (b_i-a_j)^2\right\}.
\end{align}
    This completes the proof.
\end{proof}

}


\subsection{Extensions of the contrast approach}
\label{sec_substitution_extension}
A crucial design requirement in the contrast approach discussed thus far is Assumption \ref{assump_equalsize}, which implies that the Horvitz-Thompson estimator is equivalent to the simple difference-in-means statistic. Essentially, the contrast approach is tailored towards unweighted (or self-weighted) statistics such as the difference-in-means statistic. However, the equivalence between the Horvitz-Thompson estimator and the difference-in-means statistic does not hold in general for designs where the group sizes are unequal and/or the propensity scores vary across units. 
Thus, applying the contrast approach to estimate $\Var_d(\hat{\tau})$ for such designs is not straightforward. 

Nevertheless, when the design has constant propensity scores, the difference-in-means statistic is algebraically the same as the Hajek estimator, defined as,
\begin{align}
    \hat{\tau}_{\text{Hajek}} = \frac{\sum_{i:W_i = 1}Y^{\text{obs}}_i/\pi_i}{\sum_{i:W_i = 1}1/{\pi_i}} - \frac{\sum_{i:W_i = 0}Y^{\text{obs}}_i/(1- \pi_i)}{\sum_{i:W_i = 0}1/(1 - \pi_i)}.
\end{align}
While the Hajek estimator is biased in finite samples, the bias typically tends to zero as the sample size grows. Thus, for large enough sample size, the MSE of $\hat{\tau}_{\text{Hajek}}$ is approximately the same as $\Var(\hat{\tau}_{\text{Hajek}})$. Now, the contrast approach can be used to estimate the MSE of the Hajek estimator for a class of designs, assuming treatment effect homogeneity. The primary requirement for the designs is that the propensity scores are constant across units. Following terminologies from sample surveys, we call such designs EPSEM (equal probability of selection method) designs \cite{kish1965survey}.
\begin{assumption}[EPSEM] \normalfont
    For design $d$, $\pi_i$ is constant across $i \in \{1,2,..,,N\}$. \label{assump_epsem}
\end{assumption}
Note that Assumption \ref{assump_epsem} relaxes Assumption \ref{assump_equalsize} by allowing the group sizes to be different as well as random. In the special case where $d$ has fixed (i.e., non-random) treatment and control group sizes, then $N_t(\bm{w}) = N_t$ and $N_c(\bm{w}) = N_c$. 

Before presenting the general formulation of the contrast approach for EPSEM designs, we illustrate the approach using a simpler example. Throughout, we assume that treatment effect homogeneity holds.
Under homogeneity, the MSE of the Hajek estimator can be written as,
\begin{align}
\text{MSE}(\hat{\tau}) & = \sum_{\bm{w} \in \mathcal{W}} p_{\bm{w}} \left(\frac{Y_{i_1}(0) + Y_{i_2}(0) + ... + Y_{i_{N_t(\bm{w})}}(0)}{N_t(\bm{w})} - \frac{Y_{j_1}(0) + Y_{j_2}(0) + ... + Y_{j_{N_c(\bm{w})}}(0)}{N_c(\bm{w})} \right)^2, \label{eq_mse1}
\end{align}
where, under assignment $\bm{w}$, units $\{i_1,...,i_{N_t(\bm{w})}\}$ are assigned to treatment and units $\{j_1,...,j_{N_c(\bm{w})}\}$ are assigned to control, with $\{i_1,...,i_{N_t(\bm{w})}\} \cup \{j_1,...,j_{N_c(\bm{w})}\} = \{1,2,...,N\}$. Note that, here the group sizes are allowed to vary with $\bm{w}$. 
Now, for illustration, let $N = 18$, and consider a $\bm{w}\in \mathcal{W}$ that puts unit $\{1,2,...,6\}$ in the treatment group and units $\{7,8,...,18\}$ in the control group. The corresponding term in Equation \ref{eq_mse1} is $[\{Y_{1}(0) + Y_{2}(0) + ... + Y_{6}(0)\}/6 - \{Y_{7}(0) + Y_{8}(0) + Y_{9}(0) + ... + Y_{18}(0)\}/{12}]^2$.
To use the contrast approach on this contrast, we need to find an assignment vector $\tilde{w} \in \mathcal{W}$ that assigns two of the first 6 units to treatment and four of the last 12 units to treatment. Without loss of generality, suppose that one such assignment vector is $\tilde{\bm{w}} = (1,1,0,0,0,0,1,1,1,1,0,0,...,0)^\top$. Here, units $\{1,2,7,8,9,10\}$ are treated. Then, under homogeneity, we can write,
\begin{align}
  &  \left(\frac{Y_{1}(0) + Y_{2}(0) + Y_{3}(0) + ... + Y_{6}(0)}{6} - \frac{Y_{7}(0) + Y_{8}(0) + Y_{9}(0) + Y_{10}(0) + Y_{11}(0)+  Y_{12}(0) +... + Y_{18}(0)}{12} \right)^2  \nonumber \\
  & = \left(\frac{Y_{1}(1) + Y_{2}(1) + Y_{3}(0) + ... + Y_{6}(0)}{6} - \frac{Y_{7}(1) + Y_{8}(1) + Y_{9}(1) + Y_{10}(1) + Y_{11}(0) + Y_{12}(0) + ... + Y_{18}(0)}{12} \right)^2,
\end{align}
where the right-hand side is unbiasedly estimable. Applying a similar technique to all $\bm{w} \in \mathcal{W}$, we get an estimator of $\text{MSE}(\hat{\tau}_{\text{Hajek}})$ that is unbiased under homogeneity. 

For a general EPSEM design $d$, Assumption \ref{assump_substitution2} presents the analog of the substitution condition.
\begin{assumption}[Substitution condition] \normalfont
Fix an EPSEM design $d$ with support $\mathcal{W}$. For $\bm{w} \in \mathcal{W}$, suppose units $\{i_1,...,i_{N_t(\bm{w})}\}$ are assigned to treatment and units $\{j_1,...,j_{N_c(\bm{w})}\}$ are assigned to control, where $\{i_1,...,i_{N_t(\bm{w})}\} \cup \{j_1,...,j_{N_c(\bm{w})}\} = \{1,...,N\}$. Also, let $k = N^2_t(\bm{w})/N$.
Then, there exists $\{i_{r_1},...i_{r_{k}}\} \subset \{i_1,...,i_{N_t(\bm{w})}\}$ and $\{j_{s_1},...j_{s_{N_t(\bm{w}) - k}}\} \subset \{j_1,...,j_{N_c(\bm{w})}\}$, such that under $d$, units $\{i_{r_1},...i_{r_{k}}\} \cup \{j_{s_1},...j_{s_{N_t(\bm{w}) - k}}\}$ are assigned to treatment with positive probability. \label{assump_substitution2}  
\end{assumption}
Note that, a necessary condition for Assumption \ref{assump_substitution2} to hold is that $N^2_t(\bm{w})/{N}$ is an integer. Indeed, when $N_t(\bm{w}) = N_c(\bm{w}) = N/2$, Assumption \ref{assump_substitution2} boils down to Assumption \ref{assump_substitution}, and in this case, $k = N/4$. Now, denote $a_i = Y_i(0)$, $b_i = Y_i(1)$, and $p_{\bm{w}} = \Pr_d(\bm{W} = \bm{w})$.
The Hajek estimator corresponding to $d$ can be written as,
\begin{align}
\hat{\tau}_{\text{Hajek}} & = \frac{\sum_{i=1}^{N}W_ib_i}{N_t(\bm{W})} - \frac{\sum_{i=1}^{N}(1-W_i)a_i}{N_c(\bm{W})}    
\end{align}
Let $\bar{a}$ and $\bar{b}$ be the means of $a_i$ and $b_i$ across the $N$ units. Now, the MSE of $\hat{\tau}_{\text{Hajek}}$ is given by,
\begin{align}
 \E_d(\hat{\tau}_{\text{Hajek}} - \tau)^2 & = \sum_{\bm{w}} p_{\bm{w}} \left\{\frac{b_{i_1} + ... + b_{i_{N_t(\bm{w})}}}{N_t(\bm{w})} - \frac{a_{j_1} + ... + a_{j_{N_c(\bm{w})}}}{N_c(\bm{w})} - \bar{b} + \bar{a}\right\}^2,  
\end{align}
where, for assignment vector $\bm{w}$, units $\{i_1,...,i_{N_t(\bm{w})}\}$ receive treatment and units $\{j_1,...,j_{N_c(\bm{w})}\}$ receive control, where $\{i_1,...,i_{N_t(\bm{w})}\}\cup\{j_1,...,j_{N_c(\bm{w})}\} = \{1,2,...,N\}$. Denote $c_i(\bm{w})  = \frac{N_c(\bm{w})b_i + N_t(\bm{w})a_i}{N}$. Rearranging terms, we can rewrite the MSE of $\hat{\tau}_{\text{Hajek}}$ as 
\begin{align}
    \E_d(\hat{\tau}_{\text{Hajek}} - \tau)^2 &= \sum_{\bm{w}} p_{\bm{w}} \left\{\frac{c_{i_1} + ... + c_{i_{N_t(\bm{w})}}}{N_t(\bm{w})} - \frac{c_{j_1} + ... + c_{j_{N_c(\bm{w})}}}{N_c(\bm{w})}\right\}^2.
\end{align}
Here, for simplicity, we have omitted the argument $\bm{w}$ in $c_i(\bm{w})$. Denote $\bm{a} = (a_1,...,a_N)^\top$.
When treatment effect homogeneity holds, i.e., when $b_i - a_i = \tau$, then it follows that, 
\begin{align}
    \E_d(\hat{\tau}_{\text{Hajek}} - \tau)^2 &= \sum_{\bm{w}} p_{\bm{w}} \left\{\frac{a_{i_1} + ... + a_{i_{N_t(\bm{w})}}}{N_t(\bm{w})} - \frac{a_{j_1} + ... + a_{j_{N_c(\bm{w})}}}{N_c(\bm{w})}\right\}^2 \nonumber \\
    & = \sum_{\bm{w}} p_{\bm{w}} \{\bm{l}^*(\bm{w})^\top\bm{a}\}^2,
\end{align}
where $\bm{l}^*(\bm{w}) = \frac{1}{N_t(\bm{w})}$ if $w_i = 1$ and $\bm{l}^*(\bm{w}) = -\frac{1}{N_c(\bm{w})}$ if $w_i = 0$. 

Now, let $k(\bm{w}) := N^2_t(\bm{w})/N$. If the substitution condition holds, then $k(\bm{w})$ is an integer and there exists $\{i_{r_1},...,i_{r_{k(\bm{w})}}\} \subset \{i_1,...,i_{N_t(\bm{w})}\}$ and $\{j_{s_1},...,j_{s_{N_t - k(\bm{w})}} \} \subset \{j_1,...,j_{N_c(\bm{w})}\}$ such that the units $\{i_1,...,i_{N_t(\bm{w})}\} \cup \{j_1,...,j_{N_c(\bm{w})}\}$ receive treatment with positive probability. Let the corresponding assignment vector be $\tilde{\bm{w}}$. Then, we can write,
\begin{align}
\sum_{\bm{w}} p_{\bm{w}} \{\bm{l}^*(\bm{w})^\top\bm{a}\}^2 &= \sum_{\bm{w}} p_{\bm{w}} \left\{\frac{a_{i_1} + ... + a_{i_{N_t(\bm{w})}}}{N_t(\bm{w})} - \frac{a_{j_1} + ... + a_{j_{N_c(\bm{w})}}}{N_c(\bm{w})}\right\}^2 \nonumber \\
& = \sum_{\bm{w}} p_{\bm{w}} \left\{\frac{\sum_{r \in \{r_1,...,r_{k(\bm{w})}\}}b_{i_r} + \sum_{r \notin \{r_1,...,r_{k(\bm{w})}\}}a_{i_r}}{N_t(\bm{w})} \right. \nonumber \\
& \hspace{2cm}- \left. \frac{\sum_{s \in \{s_1,...,s_{N_t(\bm{w})- k(\bm{w})}\}}b_{j_s} + \sum_{s \notin \{s_1,...,s_{N_t(\bm{w}) - k(\bm{w})}\}}a_{j_s}}{N_c(\bm{w})}   \right\}^2  \nonumber\\
& = \sum_{\bm{w}} p_{\bm{w}} \{\bm{l}^*(\bm{w})^\top \bm{y}(\tilde{\bm{w}})\}^2,
\label{eq_A_substitution_unequal1}
\end{align}
The right-hand side of \ref{eq_A_substitution_unequal1} is unbiasedly estimable. Thus, if treatment effect homogeneity holds, then we can get an unbiased estimator of $\E_d(\hat{\tau}_{\text{Hajek}} - \tau)^2$.

In addition, if $N_t(\bm{w}) = N_c(\bm{w})$, and the design is closed, i.e., $\bm{w} \in \mathcal{W} \iff \bm{1} - \bm{w} \in \mathcal{W}$, then following the proof of the symmetric case, we can write,
\begin{align}
   \sum_{\bm{w}} p_{\bm{w}} \{\bm{l}^*(\bm{w})^\top\bm{a}\}^2 & = \sum_{\bm{w}}p_{\bm{w}}\frac{1}{2}\left[\{\bm{l}^*(\bm{w})^\top \bm{y}(\tilde{\bm{w}})\}^2 + \{\bm{l}^*(\bm{w})^\top \bm{y}(\bm{1} - \tilde{\bm{w}})\}^2 \right]. \label{eq_A_substitution_unequal2} 
\end{align}
The right-hand side of Equation \ref{eq_A_substitution_unequal2} is unbiasedly estimable. Moreover, in general (without assuming homogeneity), by Jensen's inequality, 
\begin{align}
    \sum_{\bm{w}}p_{\bm{w}}\frac{1}{2}\left[\{\bm{l}^*(\bm{w})^\top \bm{y}(\tilde{\bm{w}})\}^2 + \{\bm{l}^*(\bm{w})^\top \bm{y}(\bm{1} - \tilde{\bm{w}})\}^2 \right] \geq \sum_{\bm{w}} p_{\bm{w}} \{\bm{l}^*(\bm{w})^\top\bm{c}\}^2 = \E(\hat{\tau}-\tau)^2.
\end{align}
This implies that when $N_t(\bm{w}) = N_c(\bm{w})$, and the design is closed, we can find an estimator of the MSE of the Hajek estimator that
is conservative in general, and unbiased under homogeneity.
\qed

\section{Proofs of Propositions and Theorems}

\subsection{Proof of Proposition \ref{prop_decomp2}}

Let us denote $a_i = Y_i(0)$, $b_i = Y_i(1)$, $\bm{a} = \bm{Y}(1)$, and $\bm{a} = \bm{Y}(0)$. We start by proving the following Lemma, which directly extends the Neymanian decomposition in Equation \ref{eq_neyman}.
\begin{lemma} \normalfont
Denote $p_{ii'}(w,w') = \Pr_d(W_i = w,W_{i'} = w)$, where $w,w'\in\{0,1\}$. For a design $d$ satisfying Assumption \ref{assump_positivity}, $\Var_d(\hat{\tau}) = \Tilde{V}_d - \sum_{i=1}^{N}(Y_i(1) -Y_i(0) - \tau)^2/\{N(N-1)\}$, where 
\begin{align*}
\Tilde{V}_d & = \frac{1}{N^2}\left(\sum_{i=1}^{N}\frac{Y^2_i(1)}{\pi_i} + \sum_{i=1}^{N} \frac{Y^2_i(0) }{(1-\pi_i)} \right. + 2\mathop{\sum\sum}_{i<i'}\left[Y_i(1)Y_{i'}(1)\left\{\frac{p_{ii'}(1,1)}{\pi_i \pi_{i'}} - \frac{N}{N-1}\right\} \right. \nonumber\\
& \hspace{3.2in} \left. + Y_i(0)Y_{i'}(0)\left\{\frac{p_{ii'}(0,0)}{(1-\pi_i) (1-\pi_{i'})} - \frac{N}{N-1} \right\} \right] \nonumber \\
& \quad - \left. 2\mathop{\sum\sum}_{i<i'}\left[Y_i(1)Y_{i'}(0)\left\{\frac{p_{ii'}(1,0)}{\pi_i (1- \pi_{i'})} - \frac{N}{N-1}\right\}  + Y_i(0)Y_{i'}(1)\left\{\frac{p_{ii'}(0,1)}{(1-\pi_i) \pi_{i'}}   - \frac{N}{N-1} \right\} \right]\right).
\end{align*}
\label{prop_measurable}
\end{lemma}

\begin{proof}
Let us denote $a_i = Y_i(0)$ and $b_i = Y_i(1)$. We have 
\begin{align}
\sum_{i=1}^{N}(a_i - b_i - \tau)^2 &= \sum_{i=1}^{N}(a_i - b_i)^2 - N(\bar{b} - \bar{a})^2 \nonumber\\
& = \frac{N-1}{N}\sum_{i=1}^{N}(b^2_i + a^2_i -2a_ib_i) - \frac{2}{N}\mathop{\sum\sum}_{i<i'}(b_ib_{i'} + a_ia_{i'} - b_ia_{i'}- a_ib_{i'}).
\end{align}
Thus, 
\begin{align}
    \sum_{i=1}^{N}2a_ib_i & = \sum_{i=1}^{N}(b^2_i + a^2_i) - \frac{2}{N-1}\mathop{\sum\sum}_{i<i'}(b_ib_{i'} + a_ia_{i'} - b_ia_{i'}- a_ib_{i'}).
    \label{eq_A_2ab}
\end{align}
For an arbitrary design $d$, the variance of the Horvitz-Thompson estimator of the average treatment effect $\tau$ is given by,
\begin{align}
   & \Var_d(\hat{\tau}) \nonumber\\
   & = \Var_d\left\{\frac{1}{N}\sum_{i=1}^{N}W_i\left(\frac{b_i}{\pi_i} + \frac{a_i}{1-\pi_i}\right) \right\} \nonumber \\
    & = \frac{1}{N^2}\left\{\sum_{i=1}^{N}\pi_i(1-\pi_i)\left(\frac{b_i}{\pi_i} + \frac{a_i}{1 - \pi_i} \right)^2 + 2\mathop{\sum\sum}_{i<i'}\left(\frac{b_i}{\pi_i} + \frac{a_i}{1-\pi_i}\right) \left(\frac{b_{i'}}{\pi_{i'}} + \frac{a_{i'}}{1-\pi_{i'}}\right)\Cov_d(W_i,W_{i'})  \right\} \nonumber\\
    & = \frac{1}{N^2} \Big\{\sum_{i=1}^{N} b^2_i \left(\frac{1-\pi_i}{\pi_i} \right) + \sum_{i=1}^{N} a^2_i \left(\frac{\pi_i}{1-\pi_i} \right) + \sum_{i=1}^{N}2a_ib_i \nonumber\\
    & + 2\mathop{\sum\sum}_{i<i'}\Cov_d(W_i,W_{i'})\left(\frac{b_ib_{i'}}{\pi_i \pi_{i'}} + \frac{a_ia_{i'}}{(1-\pi_i) (1-\pi_{i'})}  + \frac{b_ia_{i'}}{\pi_i (1-\pi_{i'})} + \frac{a_ib_{i'}}{(1-\pi_i) \pi_{i'}} \right) \Big\} \nonumber \\
    & = \frac{1}{N^2}\left(\sum_{i=1}^{N} b^2_i \frac{1}{\pi_i} + \sum_{i=1}^{N} a^2_i \frac{1}{(1-\pi_i)} \right.\nonumber \\
& + 2\mathop{\sum\sum}_{i<i'}\left[b_ib_{i'}\left\{\frac{\Cov_d(W_i,W_{i'})}{\pi_i \pi_{i'}} - \frac{1}{N-1}\right\} + a_ia_{i'}\left\{\frac{\Cov_d(W_i,W_{i'})}{(1-\pi_i) (1-\pi_{i'})} - \frac{1}{N-1} \right\} \right] \nonumber \\
& + \left. 2\mathop{\sum\sum}_{i<i'}\left[b_ia_{i'}\left\{\frac{\Cov_d(W_i,W_{i'})}{\pi_i (1- \pi_{i'})} + \frac{1}{N-1}\right\} + a_ib_{i'}\left\{\frac{\Cov_d(W_i,W_{i'})}{(1-\pi_i) \pi_{i'}} + \frac{1}{N-1} \right\} \right]\right) \nonumber\\
& - \frac{1}{N(N-1)}\sum_{i=1}^{N}(b_i -a_i - \tau)^2 \nonumber \\
& = \frac{1}{N^2}\left(\sum_{i=1}^{N} b^2_i \frac{1}{\pi_i} + \sum_{i=1}^{N} a^2_i \frac{1}{(1-\pi_i)} \right.\nonumber \\
& + 2\mathop{\sum\sum}_{i<i'}\left[b_ib_{i'}\left\{\frac{\Pr_d(W_i = 1,W_{i'} = 1)}{\pi_i \pi_{i'}} - \frac{N}{N-1}\right\} + a_ia_{i'}\left\{\frac{\Pr_d(W_i = 0,W_{i'} = 0)}{(1-\pi_i) (1-\pi_{i'})} - \frac{N}{N-1} \right\} \right] \nonumber \\
& - \left. 2\mathop{\sum\sum}_{i<i'}\left[b_ia_{i'}\left\{\frac{\Pr_d(W_i = 1,W_{i'} = 0)}{\pi_i (1- \pi_{i'})} - \frac{N}{N-1}\right\} + a_ib_{i'}\left\{\frac{\Pr_d(W_i = 0,W_{i'} = 1)}{(1-\pi_i) \pi_{i'}} - \frac{N}{N-1} \right\} \right]\right) \nonumber\\
& - \frac{1}{N(N-1)}\sum_{i=1}^{N}(b_i -a_i - \tau)^2 \nonumber \\
& = \Tilde{V}_d - \frac{1}{N(N-1)}\sum_{i=1}^{N}(b_i -a_i - \tau)^2,
\end{align}
where the penultimate equality holds due to Equation \ref{eq_A_2ab}. This completes the proof. 
\end{proof}
Thus, by Lemma \ref{prop_measurable}, we can express the variance as
\begin{align}
\Var_d(\hat{\tau}) & = \tilde{V}_d - (\bm{b} - \bm{a})^\top \frac{1}{N(N-1)}(\bm{I} - \frac{1}{N}\bm{J}) (\bm{b} - \bm{a}) \nonumber\\
& = \left[\tilde{V}_d + (\bm{b} - \bm{a})^\top \left\{\bm{Q} - \frac{1}{N(N-1)}\left(\bm{I} - \frac{1}{N}\bm{J}\right)\right\}(\bm{b} - \bm{a})  \right] - (\bm{b} - \bm{a})^\top \bm{Q}(\bm{b} - \bm{a}) \nonumber \\
& = \left[\tilde{V}_d + (\bm{b} - \bm{a})^\top \bm{G}(\bm{b} - \bm{a})  \right] - (\bm{b} - \bm{a})^\top \bm{Q}(\bm{b} - \bm{a}),
\end{align}
where $\bm{G} = \bm{Q} - \frac{1}{N(N-1)}\left(\bm{I} - \frac{1}{N}\bm{J}\right)$. Let $g_{ii'}$ be the $(i,i')$th element of $\bm{G}$. 
Now, 
\begin{align}
& \tilde{V}_d + (\bm{b} - \bm{a})^\top \bm{G}(\bm{b} - \bm{a}) \nonumber\\
& = \tilde{V}_d + \sum_{i=1}^Ng_{ii}(b_i-a_i)^2 +  2\mathop{\sum\sum}_{i<i'}g_{ii'}(b_i - a_i)(b_{i'} - a_{i'}) \nonumber \\
& = \tilde{V}_d +  2\mathop{\sum\sum}_{i<i'}g_{ii'}(b_ib_{i'} + a_ia_{i'} - b_ia_{i'} - a_{i}b_{i'}) \nonumber \\
& = \frac{1}{N^2}\left\{\sum_{i=1}^{N} b^2_i \frac{1}{\pi_i} + \sum_{i=1}^{N} a^2_i \frac{1}{(1-\pi_i)} \right\} + 2\mathop{\sum\sum}_{i<i'}\left[b_ib_{i'}\left\{\frac{\Pr_d(W_i = 1,W_{i'} = 1)}{N^2\pi_i \pi_{i'}} +g_{ii'} - \frac{1}{N(N-1)}\right\} \right. \nonumber\\
& \hspace{2cm} \left. + a_ia_{i'}\left\{\frac{\Pr_d(W_i = 0,W_{i'} = 0)}{N^2(1-\pi_i) (1-\pi_{i'})} +g_{ii'}- \frac{1}{N(N-1)} \right\} \right] \nonumber \\
& \quad - 2\mathop{\sum\sum}_{i<i'}\left[b_ia_{i'}\left\{\frac{\Pr_d(W_i = 1,W_{i'} = 0)}{N^2\pi_i (1- \pi_{i'})} +g_{ii'}- \frac{1}{N(N-1)}\right\} \right. \nonumber\\
& \hspace{2cm} \left. + a_ib_{i'}\left\{\frac{\Pr_d(W_i = 0,W_{i'} = 1)}{N^2(1-\pi_i) \pi_{i'}} +g_{ii'} - \frac{1}{N(N-1)} \right\} \right], \label{eq_A_vd}
\end{align}
where the last equality holds since $g_{ii} = q_{ii} - \frac{1}{N^2} = 0$ by construction of $\bm{Q}$. Moreover, for $i\neq i'$, $g_{ii'} = q_{ii'} + \frac{1}{N^2(N-1)}$. The proof follows after substituting the expression of $g_{ii'}$ in Equation \ref{eq_A_vd}.
\qed

\subsection{Proof of Theorem \ref{thm_substitution}}

Denote $a_i = Y_i(0)$ and $b_i = Y_i(1)$. Let $N$ be a multiple of four, and consider a symmetric design $d$ that assigns the $N$ units into two groups of equal size. Denote $p_{\bm{w}} = \Pr_d(\bm{W} = \bm{w})$. 
 
Let $\mathcal{G}(\bm{w})$ be a set of substitutes of ${\bm{w}}$ that is closed under label switching, i.e., if $\tilde{\bm{w}} \in \mathcal{G}(\bm{w})$, $\bm{1} - \tilde{\bm{w}} \in \mathcal{G}(\bm{w})$. In this case we can split $\mathcal{G}(\bm{w})$ into $\mathcal{G}_1(\bm{w})$ and $\mathcal{G}_0(\bm{w})$ such that $\tilde{\bm{w}} \in \mathcal{G}_1(\bm{w}) \iff \bm{1}-\tilde{\bm{w}} \in \mathcal{G}_0(\bm{w})$. It follows that, $|\mathcal{G}_1(\bm{w})| = |\mathcal{G}_0(\bm{w})| = |\mathcal{G}(\bm{w})|/2$.

The general form of the variance estimator under the contrast approach is given by
    \begin{align}
        \hat{V}_{\text{sub}} &= \frac{4}{N^2}\sum_{\bm{w}:\bm{W} \in \mathcal{G}(\bm{w})} \frac{p_{\bm{w}}}{p_{\bm{W}}} \frac{1}{|\mathcal{G}(\bm{w})|}\{\bm{l}(\bm{w})^\top \bm{Y}^{\text{obs}}\}^2 \nonumber \\
        &= \frac{4}{N^2}\sum_{\bm{w}} \sum_{\tilde{\bm{w}} \in \mathcal{G}(\bm{w})} \mathbbm{1}(\bm{W} = \tilde{\bm{w}})\frac{p_{\bm{w}}}{p_{\tilde{\bm{w}}}} \frac{1}{|\mathcal{G}(\bm{w})|}\{\bm{l}(\bm{w})^\top \bm{Y}^{\text{obs}}\}^2
    \end{align}
Therefore, 
\begin{align}
    \E_d(\hat{V}_{\text{sub}}) &= \frac{4}{N^2}\sum_{\bm{w}} \sum_{\tilde{\bm{w}} \in \mathcal{G}(\bm{w})} \frac{p_{\bm{w}}}{|\mathcal{G}(\bm{w})|}\{\bm{l}(\bm{w})^\top \bm{y}(\tilde{\bm{w}})\}^2 \nonumber \\
    & = \frac{4}{N^2}\sum_{\bm{w}}\frac{p_{\bm{w}}}{|\mathcal{G}_1(\bm{w})|} \sum_{\tilde{\bm{w}} \in \mathcal{G}_1(\bm{w})} \frac{1}{2}\left[\{\bm{l}(\bm{w})^\top \bm{y}(\tilde{\bm{w}})\}^2 + \{\bm{l}(\bm{w})^\top \bm{y}(\bm{1} - \tilde{\bm{w}})\}^2\right].
    \label{eq_A_evhat}
\end{align}
Now, for a given $\bm{w}$, let $i_1,...,i_{N/2}$ be the treated units and $j_1,...,j_{N/2}$ be the control units, where $\{i_1,...i_{N/2}\} \cup \{j_1,...,j_{N/2}\} = \{1,2,...,N\}$. Now, denoting $c_i = \frac{a_i + b_i}{2}$, we have 
    \begin{align}
        \Var_d(\hat{\tau}) &= \sum_{\bm{w}}p_{\bm{w}} \left(\frac{b_{i_1}+ ...+ b_{i_{N/2}}}{N/2} - \frac{a_{i_1}+ ...+ a_{i_{N/2}}}{N/2} - \tau \right)^2 \nonumber \\
        &= \frac{4}{N^2}\sum_{\bm{w}}p_{\bm{w}} \left\{(c_{i_1}+ ...+ c_{i_{N/2}}) - (c_{j_1}+ ...+ c_{j_{N/2}})\right\}^2 \label{eq_A_sub1}\\
        & = \frac{4}{N^2}\sum_{\bm{w}}p_{\bm{w}} \{\bm{l}(\bm{w})^\top \bm{c}\}^2, \label{eq_A_sub2}
    \end{align}
where $\bm{l}(\bm{w}) = (l_1(\bm{w}),...,l_N(\bm{w}))^\top$, with $l_i(\bm{w}) = 1$ if $w_i = 1$, and $l_i(\bm{w}) = -1$ if $w_i = 0$. Also, let $\bm{y}(\bm{w})$ be the vector of observed outcomes, had the observed assignment vector been $\bm{w}$. So, $\bm{Y}^{\text{obs}} = \bm{y}(\bm{W})$.

Now, for the given $\bm{w}$, fix a $\tilde{\bm{w}} \in \mathcal{G}_1(\bm{w})$. By the substitution condition, $p_{\tilde{\bm{w}}}>0$ and by symmetry of the design, $p_{\bm{1} - \tilde{\bm{w}}}>0$. Moreover, let
$\{i_{r_1},...i_{r_{N/4}}\} \subset \{i_1,...,i_{N/2}\}$ and $\{j_{s_1},...j_{s_{N/4}}\} \subset \{j_1,...,j_{N/2}\}$, such that $\tilde{w}_{i} = 1$ if $i \in \{i_{r_1},...i_{r_{N/4}}\} \cup \{j_{s_1},...j_{s_{N/4}}\}$ and $\tilde{w}_i = 0$ otherwise.


Now, under treatment effect homogeneity, $c_i = a_i + \tau/2$. Therefore, from Equation \ref{eq_A_sub1}, we get,
\begin{align}
    \Var_d(\hat{\tau}) & = \frac{4}{N^2}\sum_{\bm{w}}p_{\bm{w}} \left\{(a_{i_1}+ ...+ a_{i_{N/2}}) - (a_{j_1}+ ...+ a_{j_{N/2}})\right\}^2 \nonumber \\
    & = \frac{4}{N^2}\sum_{\bm{w}}p_{\bm{w}} \{\bm{l}(\bm{w})^\top \bm{a}\}^2.
    \label{eq_A_homo1}
\end{align}
Following the approach in Section \ref{sec_motivating_substitution}, under homogeneity, we can write 
\begin{align}
&    (a_{i_1}+ ...+ a_{i_{N/2}}) - (a_{j_1}+ ...+ a_{j_{N/2}}) \nonumber \\
& = \left(\sum_{r \in \{r_1,r_2,...,r_{N/4}\}}a_{i_r} + \sum_{r \notin \{r_1,r_2,...,r_{N/4}\}}a_{i_r}\right) - \left(\sum_{s \in \{s_1,s_2,...,s_{N/4}\}}a_{j_s} + \sum_{s \notin \{s_1,s_2,...,s_{N/4}\}}a_{j_s}\right) \nonumber \\
&= \left\{\sum_{r \in \{r_1,r_2,...,r_{N/4}\}}(a_{i_r} + \tau) + \sum_{r \notin \{r_1,r_2,...,r_{N/4}\}}a_{i_r}\right\} - \left\{\sum_{s \in \{s_1,s_2,...,s_{N/4}\}}(a_{j_s} + \tau) + \sum_{s \notin \{s_1,s_2,...,s_{N/4}\}}a_{j_s}\right\} \nonumber \\
&= \left(\sum_{r \in \{r_1,r_2,...,r_{N/4}\}}b_{i_r} + \sum_{r \notin \{r_1,r_2,...,r_{N/4}\}}a_{i_r}\right) - \left(\sum_{s \in \{s_1,s_2,...,s_{N/4}\}}b_{j_s} + \sum_{s \notin \{s_1,s_2,...,s_{N/4}\}}a_{j_s}\right) \nonumber \\
&= \bm{l}(\bm{w})^\top\bm{y}(\tilde{\bm{w}}).
\end{align}
Similarly, we get, $(a_{i_1}+ ...+ a_{i_{N/2}}) - (a_{j_1}+ ...+ a_{j_{N/2}}) =  \bm{l}(\bm{w})^\top\bm{y}(\bm{1} - \tilde{\bm{w}})$. Thus, $\{\bm{l}(\bm{w})^\top\bm{a}\}^2 = \frac{1}{2}[\{\bm{l}(\bm{w})^\top\bm{y}(\tilde{\bm{w}})\}^2 + \{\bm{l}(\bm{w})^\top\bm{y}(\bm{1} - \tilde{\bm{w}})\}^2]$. Since this holds for every $\tilde{\bm{w}} \in \mathcal{G}_1(\bm{w})$, we get 
\begin{align}
    \frac{4}{N^2}\sum_{\bm{w}}\frac{p_{\bm{w}}}{|\mathcal{G}_1(\bm{w})|} \sum_{\tilde{\bm{w}} \in \mathcal{G}_1(\bm{w})} \frac{1}{2}\left[\{\bm{l}(\bm{w})^\top \bm{y}(\tilde{\bm{w}})\}^2 + \{\bm{l}(\bm{w})^\top \bm{y}(\bm{1}-\tilde{\bm{w}})\}^2\right] = \frac{4}{N^2}\sum_{\bm{w}}p_{\bm{w}} \{\bm{l}(\bm{w})^\top \bm{a}\}^2. \label{eq_A_homo2}
\end{align}
Equation \ref{eq_A_homo2}, combined with equations \ref{eq_A_evhat} and \ref{eq_A_homo1} implies that under treatment effect homogeneity, $\hat{V}_{\text{sub}}$ is unbiased for $\Var(\hat{\tau})$.

Next, we show that $\hat{V}_{\text{sub}}$ is conservative in general. To this end, we note that, for any $\tilde{\bm{w}}$, $\frac{1}{2}\{\bm{y}(\tilde{\bm{w}}) + \bm{y}(\bm{1}-\tilde{\bm{w}})\} = \frac{\bm{a} + \bm{b}}{2} = \bm{c}$. Therefore, by Jensen's inequality,
\begin{align}
    \E_d(\hat{V}_{\text{sub}}) &=   \frac{4}{N^2}\sum_{\bm{w}}\frac{p_{\bm{w}}}{|\mathcal{G}_1(\bm{w})|} \sum_{\tilde{\bm{w}} \in \mathcal{G}_1(\bm{w})} \frac{1}{2}\left[\{\bm{l}(\bm{w})^\top \bm{y}(\tilde{\bm{w}})\}^2 + \{\bm{l}(\bm{w})^\top \bm{y}(\bm{1}-\tilde{\bm{w}})\}^2\right] \nonumber \\
    & \geq \frac{4}{N^2}\sum_{\bm{w}}\frac{p_{\bm{w}}}{|\mathcal{G}_1(\bm{w})|} \sum_{\tilde{\bm{w}} \in \mathcal{G}_1(\bm{w})} \{\bm{l}(\bm{w})^\top \bm{c}\}^2 \nonumber\\
    & = \frac{4}{N^2} \sum_{\bm{w}} p_{\bm{w}} \{\bm{l}(\bm{w})^\top \bm{c}\}^2 = \Var_d(\hat{\tau}),
\end{align}
where the last equality holds from Equation \ref{eq_A_sub2}. This completes the proof.
\qed

\subsection{Proof of Theorem \ref{thm_substitute_neyman}}

For a completely randomized design (CRD) with equal group size, Neyman's variance estimator is given by,
\begin{align}
    \hat{V}_{\text{Neyman}} & = \frac{1}{\frac{N}{2}(\frac{N}{2}-1)}\left\{\sum_{i:W_i = 1}(Y^{\text{obs}}_i - \bar{Y}_t)^2+ \sum_{i:W_i = 0}(Y^{\text{obs}}_i - \bar{Y}_c)^2 \right\} \nonumber\\
    & = \frac{4}{N^2(N-2)}\left\{\mathop{\sum\sum}_{i,j:W_i,W_j = 1}(Y^{\text{obs}}_i - Y^{\text{obs}}_j)^2 + \mathop{\sum\sum}_{i,j:W_i,W_j = 0}(Y^{\text{obs}}_i - Y^{\text{obs}}_j)^2 \right\} \nonumber \\
    & = \frac{4}{N^2(N-2)}\left\{\mathop{\sum\sum}_{i\neq j:W_i,W_j = 1}(Y^{\text{obs}}_i - Y^{\text{obs}}_j)^2 + \mathop{\sum\sum}_{i\neq j:W_i,W_j = 0}(Y^{\text{obs}}_i - Y^{\text{obs}}_j)^2 \right\} \nonumber\\
    & = \frac{4}{N^2}\sum_{i=1}^{N}(Y^{\text{obs}}_i)^2 - \frac{8}{N^2(N-2)}\left(\mathop{\sum\sum}_{i \neq j: W_i,W_j = 1}Y^{\text{obs}}_i Y^{\text{obs}}_j + \mathop{\sum\sum}_{i \neq j: W_i,W_j = 0}Y^{\text{obs}}_i Y^{\text{obs}}_j \right). \label{eq_A_crd}
\end{align}
Now, consider the variance estimator under the contrast approach with $\mathcal{G}(\bm{w}) = \mathcal{G}^*(\bm{w})$. For a CRD, $|\mathcal{G}^*(\bm{w})| = \binom{N/2}{N/4}\binom{N/2}{N/4}$. The resulting estimator is given by,
\begin{align}
    \hat{V}_{\text{sub}} & = \frac{4}{N^2}\sum_{\bm{w}:\bm{W} \in \mathcal{G}^*(\bm{w})} \frac{p_{\bm{w}}}{p_{\bm{W}}} \frac{1}{|\mathcal{G}^*(\bm{w})|}\{\bm{l}(\bm{w})^\top \bm{Y}^{\text{obs}}\}^2 \nonumber \\
    &= \frac{4}{N^2} \sum_{\bm{w} \in \mathcal{G}^*(\bm{W})}\frac{p_{\bm{w}}}{p_{\bm{W}}} \frac{1}{|\mathcal{G}^*(\bm{w})|}\{\bm{l}(\bm{w})^\top \bm{Y}^{\text{obs}}\}^2 \nonumber \\
    & = \frac{4}{N^2} \frac{1}{\binom{N/2}{N/4}^2} \sum_{\bm{w} \in \mathcal{G}^*(\bm{W})} \{\bm{l}(\bm{w})^\top \bm{Y}^{\text{obs}}\}^2 \nonumber \\
    & = \frac{4}{N^2} \frac{1}{\binom{N/2}{N/4}^2} \sum_{\bm{w} \in \mathcal{G}^*(\bm{W})}\left\{(Y^{\text{obs}}_{i_1} +...+Y^{\text{obs}}_{i_{N/2}}) - (Y^{\text{obs}}_{j_1} +...+Y^{\text{obs}}_{j_{N/2}})\right\}^2 \nonumber \\
    & = \frac{4}{N^2} \frac{1}{\binom{N/2}{N/4}^2} \sum_{\bm{w} \in \mathcal{G}^*(\bm{W})} \left\{\sum_{i=1}^{N}(Y^{\text{obs}}_i)^2 + \mathop{\sum\sum}_{i\neq j: w_i=w_j}Y^{\text{obs}}_iY^{\text{obs}}_j - \mathop{\sum\sum}_{i\neq j: w_i\neq w_j}Y^{\text{obs}}_iY^{\text{obs}}_j \right\}
    \label{eq_A_vhat_crd}
\end{align}
where for a generic assignment vector $\bm{w}$, units $\{i_1,...,i_{N/2}\}$ are assigned to treatment, and units $\{j_1,...,j_{N/2}\}$ are assigned to control.
The right-hand side of Equation \ref{eq_A_vhat_crd} is a quadratic form in the observed outcomes. To show that $\hat{V}_{\text{sub}}$ is algebraically identical to $\hat{V}_{\text{Neyman}}$, we show that the coefficients of the two quadratic forms are equal. To this end, we assume without loss of generality, that the observed vector of treatment assignment $\bm{W}$ puts units $1,2,...,N/2$ in the treatment group and rest in the control group. 

First, we consider the coefficient of $(Y^{\text{obs}}_1)^2$ in Equation \ref{eq_A_vhat_crd}, which equals 
\begin{align}
 \frac{4}{N^2} \frac{1}{\binom{N/2}{N/4}^2} \times {\binom{N/2}{N/4}^2}   = \frac{4}{N^2},
\end{align}
which is same as the coefficient of $(Y^\text{obs}_1)^2$ in Equation \ref{eq_A_crd}. By symmetry, the coefficients for $(Y^\text{obs}_i)^2$ are the same in both quadratic forms, for all $i \in \{1,2,...,N\}$.

Next, we consider the coefficient of $Y^{\text{obs}}_1 Y^{\text{obs}}_2$ in Equation \ref{eq_A_vhat_crd}, which equals
\begin{align}
&    \frac{4}{N^2}\frac{1}{\binom{N/2}{N/4}^2} \times\left\{\binom{N/2 - 2}{N/4}\binom{N/2}{N/4} + \binom{N/2 - 2}{N/4 - 2}\binom{N/2}{N/4} - 2\binom{N/2 - 2}{N/4 - 1}\binom{N/2}{N/4} \right\} \nonumber \\
& = \frac{4}{N^2} \times \frac{(-2)}{N-2} = -\frac{8}{N^2(N-2)},
\end{align}
which is same as the coefficient of $Y^{\text{obs}}_1 Y^{\text{obs}}_2$ in Equation \ref{eq_A_crd}. By symmetry, the coefficients of $Y^{\text{obs}}_i Y^{\text{obs}}_i$ are the same in both quadratic forms, for all $i \neq j$ such that $W_i = W_j$. 

Finally, we consider the coefficient of $Y^{\text{obs}}_1 Y^{\text{obs}}_{N}$ in Equation \ref{eq_A_vhat_crd}, which equals
\begin{align}
&    \frac{4}{N^2}\frac{1}{\binom{N/2-1}{N/4}^2} \times\left\{2\binom{N/2 - 1}{N/4}\binom{N/2-1}{N/4-1} - 2\binom{N/2 - 1}{N/4}\binom{N/2-1}{N/4} \right\} \nonumber \\
& = \frac{4}{N^2} \times 0 = 0,
\end{align}
which is same as the coefficient of $Y^{\text{obs}}_1 Y^{\text{obs}}_N$ in Equation \ref{eq_A_crd}. By symmetry, the coefficients of $Y^{\text{obs}}_i Y^{\text{obs}}_i$ are the same in both quadratic forms, for all $i \neq j$ such that $W_i \neq W_j$. 

Thus, the two quadratic forms are identical. This completes the proof.
\qed
\subsection{Proof of Proposition \ref{prop_subs_example}}

Denote $a_i = Y_i(0)$, $b_i = Y_i(1)$. In this example, $\pi_i = 0.5$, and for $w \in \{0,1\}$, $\Pr_d(W_i = w, W_{i+1} = w) = \Pr_d(W_1 = w, W_4 = w) = 0.25$, $\Pr_d(W_{i} = w, W_{i+1} = 1-w) = \Pr_d(W_1 = w, W_4 = 1-w) = 0.25$, and $\Pr_d(W_1 = w, W_3 = 1-w) = \Pr_d(W_2 = w, W_4 = 1-w) = 0.5$. 
Using the decomposition in Proposition \ref{prop_decomp2} with the given choice of $\bm{Q}$, we get
\begin{align}
\Tilde{V}_d(\bm{Q}) &= \frac{2}{16}\sum_{i=1}^{4}(a^2_i + b^2_i) - \frac{2}{16}(b_1b_2+a_1a_2+b_2b_3+a_2a_3+b_3b_4+a_3a_4+b_1b_4+a_1a_4) \nonumber\\
& \quad + \frac{2}{16}(b_1a_2+a_1b_2+b_2a_3+a_2b_3+b_2a_3+a_2b_3+b_1a_4+a_1b_4) - \frac{4}{16}(b_1a_3 + a_1b_3 + b_2a_4 + a_2b_4)
\end{align}
Thus, denoting $\tilde{W}_i = 1 - W_i$, the Neymanian decomposition-based variance estimator is
\begin{align}
\hat{\Tilde{V}}_d(\bm{Q}) &= \frac{4}{16}\sum_{i=1}^{4}(Y^{\text{obs}}_i)^2 - \frac{8}{16}\left\{\sum_{i=1}^{3}(W_iW_{i+1} + \Tilde{W}_i\Tilde{W}_{i+1})Y^{\text{obs}}_iY^{\text{obs}}_{i+1} + (W_1W_4 + \Tilde{W}_1\Tilde{W}_4)Y^{\text{obs}}_1Y^{\text{obs}}_4 \right\} \nonumber\\
& \quad + \frac{8}{16}\left\{\sum_{i=1}^{3}(W_i\tilde{W}_{i+1} + \Tilde{W}_i {W}_{i+1})Y^{\text{obs}}_iY^{\text{obs}}_{i+1} + (W_1\tilde{W}_4 + \Tilde{W}_1{W}_4)Y^{\text{obs}}_1Y^{\text{obs}}_4 \right\} \nonumber\\
& \quad - \frac{8}{16} \left\{(W_1\tilde{W}_3 + \Tilde{W}_1{W}_3)Y^{\text{obs}}_1Y^{\text{obs}}_3 + (W_2\tilde{W}_4 + \Tilde{W}_2{W}_4)Y^{\text{obs}}_2Y^{\text{obs}}_4 \right\}.
\end{align}
In particular, when $\bm{W} \in \{(1,1,0,0)^\top, (0,0,1,1)^\top\}$, 
\begin{align}
    \hat{\Tilde{V}}_d(\bm{Q}) &= \frac{1}{4}\sum_{i=1}^{4}(Y^{\text{obs}}_i)^2 - \frac{1}{2}(Y^{\text{obs}}_1Y^{\text{obs}}_2 + Y^{\text{obs}}_1Y^{\text{obs}}_3 - Y^{\text{obs}}_1Y^{\text{obs}}_4 - Y^{\text{obs}}_2Y^{\text{obs}}_3 + Y^{\text{obs}}_2Y^{\text{obs}}_4 + Y^{\text{obs}}_3Y^{\text{obs}}_4) \nonumber \\
    & = \frac{1}{4} (Y^{\text{obs}}_1-Y^{\text{obs}}_2-Y^{\text{obs}}_3+Y^{\text{obs}}_4)^2, \label{eq_A_examp_1}
\end{align}
and when $\bm{W} \in \{(1,0,0,1)^\top, (0,1,1,0)^\top\}$, 
\begin{align}
    \hat{\Tilde{V}}_d(\bm{Q}) &= \frac{1}{4}\sum_{i=1}^{4}(Y^{\text{obs}}_i)^2 + \frac{1}{2}(Y^{\text{obs}}_1Y^{\text{obs}}_2 - Y^{\text{obs}}_1Y^{\text{obs}}_3 - Y^{\text{obs}}_1Y^{\text{obs}}_4 - Y^{\text{obs}}_2Y^{\text{obs}}_3 - Y^{\text{obs}}_2Y^{\text{obs}}_4 + Y^{\text{obs}}_3Y^{\text{obs}}_4) \nonumber \\
 & = \frac{1}{4} (Y^{\text{obs}}_1+Y^{\text{obs}}_2-Y^{\text{obs}}_3-Y^{\text{obs}}_4)^2, \label{eq_A_examp_2}
\end{align}
Now, in the contrast approach with the full set of substitutes, we estimate the variance of $\hat{\tau}$ by unbiasedly estimating
\begin{align}
    V = \frac{1}{16}\left\{(b_1 +a_2 - a_3 - b_4)^2 + (a_1 +b_2 - b_3 - a_4)^2 + (b_1-b_2-a_3+a_4)^2 + (a_1-a_2-b_3+b_4)^2\right\}
\end{align}
using a Horvitz-Thompson-type estimator. More specifically, we use the estimator
\begin{align}
 \hat{V} & = \frac{1}{16}\left[\frac{\mathbbm{1}\{\bm{W} = (1,0,0,1)^\top\} + \mathbbm{1}\{\bm{W} = (0,1,1,0)^\top\}}{1/4}(Y^{\text{obs}}_1+Y^{\text{obs}}_2-Y^{\text{obs}}_3-Y^{\text{obs}}_4)^2 \right. \nonumber \\
 & \quad \left. + \frac{\mathbbm{1}\{\bm{W} = (1,1,0,0)^\top\} + \mathbbm{1}\{\bm{W} = (0,0,1,1)^\top\}}{1/4}(Y^{\text{obs}}_1-Y^{\text{obs}}_2-Y^{\text{obs}}_3+Y^{\text{obs}}_4)^2 \right]. \label{eq_A_examp_3}
\end{align}
From Equations \ref{eq_A_examp_1}, \ref{eq_A_examp_2}, and \ref{eq_A_examp_3} it follows that, the estimators $\hat{\tilde{V}}_d(\bm{Q})$ and $\hat{V}$ are identical. This completes the proof.

Consider an alternative design where $\Pr(\bm{W} = (1,1,0,0)^\top) = \Pr(\bm{W} = (0,0,1,1)^\top) = 1/3$ and $\Pr(\bm{W} = (1,0,0,1)^\top) = \Pr(\bm{W} = (0,1,1,0)^\top) = 1/6$. In this case, the contrast approach-based estimator does not correspond to an estimator of the form $\hat{\tilde{V}}_d(\bm{Q})$ for any $\bm{Q}$. To see this, we can, without loss of generality, set the observed treatment assignment to $\bm{W} = (1, 1, 0, 0)^\top$. In this case, the coefficient of $(Y^{\text{obs}}_1)^2$ in the contrast approach-based estimator is 0.5, whereas in $\hat{\tilde{V}}_d(\bm{Q})$—regardless of the choice of $\bm{Q}$—it is 0.25.

\qed
\subsection{Proof of Theorem \ref{thm_substitute_matched}}

The proof of this theorem follows similar steps as in the proof of Theorem \ref{thm_substitute_neyman}. In a matched-pair setting, each pair $j$ has two units, labelled 1 and 2. Without loss of generality, we assume that in the observed data, the first unit in each pair is treated, i.e., $\bm{W} = (\underbrace{1,0}_{\text{pair 2}},\underbrace{1,0}_{\text{pair 2}},....,\underbrace{1,0}_{\text{pair $2k$}})^\top$. Let $d_j = Y^{\text{obs}}_{j1} - Y^{\text{obs}}_{j2}$ be the difference in observed outcome between units 1 and 2 in pair $j$, and let $\bar{d} = \frac{1}{2k}\sum_{j=1}^{2k}d_{j}$. For the given $\bm{W}$, we can express both $\hat{V}_{\text{pair}}$ and $\hat{V}_{\text{sub}}$ in terms of $d_j$. In particular, 
\begin{align}
\hat{V}_{\text{pair}} &= \frac{4k}{4k(4k-2)}\sum_{j=1}^{2k}(d_j - \bar{d})^2 \nonumber \\
& = \frac{1}{4k^2(4k-2)}\left[(4k-2)\sum_{j=1}^{2k}d^2_j - 4\mathop{\sum\sum}_{j<j'}d_jd_{j'} \right]. \label{eq_pair1}
\end{align}
Now, for any assignment vector $\bm{w}\in \mathcal{W}$, we see that its substitute is obtained by selecting $k$ pairs out of the $2k$ available pairs and switching the treatment labels within each pair (from what it was in $\bm{w}$). For instance, consider the special case of $k=1$ (i.e., $N-4$), and the assignment vector $\bm{w} = (\underbrace{1,0}_{\text{pair 1}},\underbrace{1,0}_{\text{pair 2}})^\top$. It has two substitutes: $(\underbrace{0,1}_{\text{pair 1}},\underbrace{1,0}_{\text{pair 2}})^\top$ (switching is done pair 1) and $(\underbrace{1,0}_{\text{pair 1}},\underbrace{0,1}_{\text{pair 2}})^\top$ (switching is done pair 2). Thus, in general, the total number of substitutes for any $\bm{w}$ is $|\mathcal{G}^*(\bm{w})| = \binom{2k}{k}$. 
\begin{align}
    \hat{V}_{\text{sub}} & =  \frac{4}{16k^2}\frac{1}{\binom{2k}{k}}\sum_{\bm{w} \in \mathcal{G}^*(\bm{W})} \{\bm{l}(\bm{w})^\top \bm{Y}^{\text{obs}}\}^2 \nonumber\\
    & = \frac{1}{4k^2}\frac{1}{\binom{2k}{k}} \sum_{\bm{w} \in \mathcal{G}^*(\bm{W})} \left(\sum_{j:z_j(\bm{w}) = 1}d_j - \sum_{j:z_j(\bm{w}) = 0}d_j \right)^2,
\end{align}
where $z_j(\bm{w}) =1$ if under $\bm{w}$, the 1st unit in pair $j$ is treated. 
Thus,
\begin{align}
    \hat{V}_{\text{sub}} = \frac{1}{4k^2}\frac{1}{\binom{2k}{k}} \sum_{\bm{w} \in \mathcal{G}^*(\bm{W})} \left(\sum_{j=1}^{2k}d^2_j + 2\mathop{\sum\sum}_{j <j':z_j(\bm{w}) = z_{j'}(\bm{w})}d_jd_{j'} - 2\mathop{\sum\sum}_{j <j':z_j(\bm{w}) \neq z_{j'}(\bm{w})}d_jd_{j'} \right). \label{eq_submatch}
\end{align}
To show that $\hat{V}_{\text{pair}}$ and $\hat{V}_{\text{sub}}$ are algebraically identical, we compare the coefficients of $d_j^2$ and $d_jd_{j'}$. By symmetry, it suffices to compare the coefficients of $d^2_1$ and $d_1d_2$. From Equation \ref{eq_pair1}, it follows that the coefficient of $d^2_1$ is $\frac{1}{4k^2}$ and the coefficient of $d_{1}d_{2}$ is $-\frac{1}{k^2(4k-2)}$. Now, from Equation \ref{eq_submatch}, the coefficient of $d^2_1$ is $\binom{2k}{k} \times \frac{1}{4k^2 \binom{2k}{k}} = \frac{1}{4k^2}$, and the coefficient of $d_{12}$ is
\begin{equation}
    \frac{1}{4k^2 \binom{2k}{k}}\left\{ 2\binom{2k-2}{k-2} + 2\binom{2k-2}{k} - 4\binom{2k-2}{k-1}  \right\} = -\frac{1}{k^2(4k-2)}.
\end{equation}
This completes the proof.
\subsection{Proof of Proposition \ref{prop_impute_gen}}
Consider a design $d$ and denote $\pi_i = \Pr_d(W_i = 1)$ and $p_{\bm{w}} = \Pr_d(\bm{W} = \bm{w})$. The corresponding Horvitz-Thompson (HT) estimator can be written as,
   \begin{align}
       \hat{\tau} &= \frac{1}{N}\sum_{i=1}^{N}\frac{W_i Y_i(1)}{\pi_i} - \frac{1}{N}\sum_{i=1}^{N}\frac{(1-W_i) Y_i(0)}{1-\pi_i} \nonumber \\
       & = \frac{1}{N}\sum_{i=1}^{N} \left(\frac{1}{\pi_i} + \frac{1}{1-\pi_i} \right)W_i c_i - \frac{1}{N}\sum_{i=1}^{N}\frac{Y_i(0)}{1-\pi_i}, 
   \end{align}
where $c_i = \frac{\frac{Y_i(1)}{\pi_i} + \frac{Y_i(0)}{1-\pi_i}}{\frac{1}{\pi_i} + \frac{1}{1-\pi_i}} = (1 - \pi_i)Y_i(1) + \pi_i Y_i(0)$. 
So,
\begin{align}
    \Var_d(\hat{\tau}) &= \Var_d\left\{\frac{1}{N}\sum_{i=1}^{N}\frac{W_ic_i}{\pi_i(1 - \pi_i)} \right\} \label{eq_A_impute_step0} \\
    & = \frac{1}{N^2}\E_d\left\{\sum_{i=1}^{N}\frac{W_ic_i}{\pi_i(1 - \pi_i)} - \sum_{i=1}^{N}\frac{c_i}{1-\pi_i} \right\}^2 \nonumber \\
    & = \frac{1}{N^2}\sum_{\bm{w}} p_{\bm{w}} \left(\sum_{i:w_{i} = 1}\frac{c_{i}}{\pi_{i}} - \sum_{i:w_{i} = 0}\frac{c_{i}}{1-\pi_{i}} \right)^2 =: \textcolor{black}{\psi_d}(\bm{c}) \label{eq_A_impute1}
\end{align}
Here, we note that the number of treated units is allowed to vary across $\bm{w}$. 

Also, from Equation \ref{eq_A_impute_step0}, we observe that $\Var_d(\hat{\tau}) = \Var_d(\bm{c}^\top\bm{D})$, where $\bm{D} = (D_1,...,D_N)^\top$ with $D_i = \frac{W_i}{N\pi_i(1-\pi_i)}$. Thus, $\Var_d(\hat{\tau}) = \bm{c}^\top \Var_d(\bm{D}) \bm{c}$. Since $\Var_d(\bm{D})$ is always non-negative definite, it follows that $\textcolor{black}{\psi_d}(\cdot)$ is convex.
This completes the proof.
\qed

\subsection{Proof of Proposition \ref{thm_imputation1}}

{\color{black}
Note that, 
\begin{align*}
    \E_d\{\hat{Y}_i(1)\} - Y_i(1) = \E_d\{W_iY_i(1) + (1-W_i)(Y_i(0) + \beta_i) - Y_i(1)\} = -\frac{Y_i(1)-Y_i(0)}{2} + \frac{\beta_i}{2},
\end{align*}
and
\begin{align*}
    \E_d\{\hat{Y}_i(0)\} - Y_i(0) = \E_d\{W_i (Y_i(1) - \beta_i) + (1-W_i) Y_i(0)\} - Y_i(0) = \frac{Y_i(1)-Y_i(0)}{2} - \frac{\beta_i}{2} = - [\E_d\{\hat{Y}_i(1)\} - Y_i(1)]
\end{align*}
Thus,
\begin{align*}
  \E_d(\hat{c}_i)  &= \frac{1}{2}\E_d\{\hat{Y}_i(1)+\hat{Y}_i(0)\} = \frac{1}{2}\{Y_i(1)+Y_i(0)\} = c_i.
\end{align*}
Now, since $\textcolor{black}{\psi_d}(\cdot)$ is convex, by Jensen's inequality, we have $\E_d\{\textcolor{black}{\psi_d}(\hat{\bm{c}})\} \geq \textcolor{black}{\psi_d}\{\E_d(\hat{\bm{c}})\} = \textcolor{black}{\psi_d}(\bm{c})$. This proves part (a). 
}
\textcolor{black}{Part (b)} follows directly from Corollary \ref{corollary_imputation}.

\subsection{Proof of Proposition \ref{prop_limit}}

\textcolor{black}{We assume that treatment effects are homogeneous.}
From Proposition \ref{thm_imputation1}\textcolor{black}{(b)}, we have $\E_d\{\textcolor{black}{\psi_d}(\hat{\bm{c}})\} = \Var_d(\hat{\tau}) + (\tau - \beta)^2\E_d\{\textcolor{black}{\psi_d}(\bm{W})\}$. Now, 
\begin{align}
    \E_d\{\textcolor{black}{\psi_d}(\bm{W})\} & = \sum_{\bm{w}\in \mathcal{W}}p_{\bm{w}}\E_d\left(\frac{W_{i_1}+ ...+W_{i_{N/2}}}{N/2} - \frac{W_{j_1}+ ...+W_{j_{N/2}}}{N/2}  \right)^2,
\end{align}
where $\{i_1,...,i_{N/2}\}$ and $\{j_1,...,j_{N/2}\}$ are the set of treated and control units under $\bm{w}$, respectively. Now, since $(W_1+W_2 + ... +W_N)^2 = \frac{N^2}{4},$ we have
\begin{align}
2\mathop{\sum\sum}_{i<j}W_iW_j & = \frac{N}{2}\left(\frac{N}{2} - 1 \right)    \nonumber \\
\iff \mathop{\sum\sum}_{i<j:w_i = w_j}\pi_{ij} + \mathop{\sum\sum}_{i<j:w_i \neq w_j}\pi_{ij} &= \frac{N}{4}\left(\frac{N}{2} - 1 \right). \label{eq_A_5.4_1}
\end{align}
Now, for a fixed $\bm{w} \in \mathcal{W}$, we have
\begin{align}
\E_d\left(\frac{W_{i_1}+ ...+W_{i_{N/2}}}{N/2} - \frac{W_{j_1}+ ...+W_{j_{N/2}}}{N/2}  \right)^2 & = \frac{4}{N^2}\left(\frac{N}{2} + 2\mathop{\sum\sum}_{i<j:w_i = w_j}\pi_{ij} - 2\mathop{\sum\sum}_{i<j:w_i \neq w_j}\pi_{ij}   \right) \nonumber \\
& = \frac{4}{N} + \frac{16}{N^2}\mathop{\sum\sum}_{i<j:w_i = w_j}\pi_{ij} - 1 \nonumber\\
& = \frac{16}{N^2}\mathop{\sum\sum}_{i<j:w_i = w_j}[\pi_{ij} - (N-2)/\{4(N-1)\}] + 1/(N-1),
\end{align}
where the penultimate equality holds due to Equation \ref{eq_A_5.4_1}. By the given condition,
\begin{align}
\E_d\left(\frac{W_{i_1}+ ...+W_{i_{N/2}}}{N/2} - \frac{W_{j_1}+ ...+W_{j_{N/2}}}{N/2}  \right)^2 & = o\left(1\right).
\end{align}
Thus, $\E_d\{\textcolor{black}{\psi_d}(\bm{W})\} = o(1)$.
In particular, for a CRD, $\pi_{ij} = (N-2)/\{4(N-1)\}$ and hence, $\E_d\{\textcolor{black}{\psi_d}(\bm{W})\} = 1/(N-1)$.
\qed

\subsection{Proof of Proposition \ref{prop_imputation_neyman}}

For a completely randomized design (CRD) with equal group size, Neyman's variance estimator is given by,
\begin{align}
    \hat{V}_{\text{Neyman}} & = \frac{1}{\frac{N}{2}(\frac{N}{2}-1)}\left\{\sum_{i:W_i = 1}(Y^{\text{obs}}_i - \bar{Y}_t)^2+ \sum_{i:W_i = 0}(Y^{\text{obs}}_i - \bar{Y}_c)^2 \right\} \nonumber\\
    & = \frac{4}{N^2}\sum_{i=1}^{N}(Y^{\text{obs}}_i)^2 - \frac{8}{N^2(N-2)}\left(\mathop{\sum\sum}_{i \neq j: W_i,W_j = 1}Y^{\text{obs}}_i Y^{\text{obs}}_j + \mathop{\sum\sum}_{i \neq j: W_i,W_j = 0}Y^{\text{obs}}_i Y^{\text{obs}}_j \right). \label{eq_A_impute_neyman}
\end{align}
Now, without loss of generality, suppose that the observed assignment vector assigns units $1,2,...,N/2$ to treatment and the rest to control. In that case, the realized value of $\hat{V}_{\text{Neyman}}$ is
\begin{align}
 \hat{V}_{\text{Neyman}} & =  \frac{4}{N^2}\sum_{i=1}^{N}(Y^{\text{obs}}_i)^2 - \frac{8}{N^2(N-2)}\left(\mathop{\sum\sum}_{i \neq j \in \{1,2,...,N/2\}}Y^{\text{obs}}_i Y^{\text{obs}}_j + \mathop{\sum\sum}_{i \neq j \in \{N/2+1,...,N\}}Y^{\text{obs}}_i Y^{\text{obs}}_j \right)   
\end{align}
Now, we consider the imputation approach where the missing potential outcomes are imputed as if $\hat{\tau}$ is the true unit-level causal effect. More formally, 
\begin{equation}
  \hat{Y}_i(0) = 
    \begin{cases}
     Y^{\text{obs}}_i - \hat{\tau} & \text{if $W_i = 1$}\\
    Y^{\text{obs}}_i & \text{if $W_i = 0$},
    \end{cases}       
\end{equation}
\begin{equation}
  \hat{Y}_i(1) =
    \begin{cases}
     Y^{\text{obs}}_i & \text{if $W_i = 1$}\\
    Y^{\text{obs}}_i + \hat{\tau} & \text{if $W_i = 0$}.
    \end{cases}       
\end{equation}
Thus, under homogeneity, 
\begin{align}
    \hat{c}_i = \frac{\hat{Y}_i(0) + \hat{Y}_i(1)}{2}.
\end{align}
In particular, when $\bm{W}$ assigns the first $N/2$ units to treatment and the rest to control then, $\hat{c}_i = Y^{\text{obs}}_i - \frac{\hat{\tau}}{2}$ for $i \in \{1,...,N/2\}$ and $\hat{c}_i = Y^{\text{obs}}_i + \frac{\hat{\tau}}{2}$ for $i \in \{N/2+1,...,N\}$. 
Now, the corresponding imputation estimator is given by,
\begin{align}
    \textcolor{black}{\psi_d}(\hat{\bm{c}}) & = \frac{4}{N^2}\frac{1}{\binom{N}{N/2}}\sum_{\bm{w}} \left(\sum_{i:w_{i} = 1}\hat{c}_{i} - \sum_{i:w_{i} = 0}\hat{c}_{i} \right)^2 
\end{align}
We introduce a few additional notations. First, for any $\bm{w}$, denote $\bm{l}(\bm{w}) = (l_1(\bm{w}),...,l_N(\bm{w}))^\top,$ where $l_i(\bm{w}) = 1$ if $w_i = 1$, and $l_i(\bm{w}) = -1$ otherwise. Moreover, denote for every $\bm{w}$, let $r(\bm{w}) = \#\{i:w_i = 1, W_i = 1\}$ be the number of units that are treated both under assignment vector $\bm{w}$ and the observed assignment vector $\bm{W}$. In this case, $r(\bm{w})$ is simply the number of treated units under $\bm{w}$ among the first $N/2$ units. It follows that, $\bm{l}(\bm{w})^\top \bm{1} = 0$, $\bm{l}(\bm{w})^\top \bm{l}(\bm{W}) = 4r(\bm{w}) - N$, and $\bm{l}(\bm{W})^\top \bm{Y}^{\text{obs}} = (N/2)\hat{\tau}$. Now, the variance estimator can be written as,
\begin{align}
     \textcolor{black}{\psi_d}(\hat{\bm{c}}) & = \frac{4}{N^2}\frac{1}{\binom{N}{N/2}}\sum_{\bm{w}}\{\bm{l}(\bm{w})^\top \hat{\bm{c}}\}^2 \nonumber\\
     & = \frac{4}{N^2}\frac{1}{\binom{N}{N/2}}\sum_{\bm{w}} \left\{\bm{l}(\bm{w})^\top\bm{Y}^{\text{obs}} - \frac{\hat{\tau}}{2}\bm{l}(\bm{w})^\top\bm{l}(\bm{W}) \right\}^2 \nonumber \\
     & = \frac{4}{N^2}\frac{1}{\binom{N}{N/2}}\sum_{\bm{w}} \left[ \left\{ \bm{l}(\bm{w}) - \frac{4r(\bm{w}) - N}{N}\bm{l}(\bm{W}) \right\}^\top \bm{Y}^{\text{obs}} \right]^2 \nonumber \\ 
     & = \frac{4}{N^2}\frac{1}{\binom{N}{N/2}}\sum_{\bm{w}}\{\bm{g}(\bm{w})^\top\bm{Y}^{\text{obs}}\}^2 \nonumber \\
     & = \frac{4}{N^2}\frac{1}{\binom{N}{N/2}}\left[ \sum_{i=1}^{N} \left\{\sum_{\bm{w}}g^2_i(\bm{w})\right\}(Y^{\text{obs}}_i)^2 + \mathop{\sum\sum}_{i \neq j}\left\{\sum_{\bm{w}}g_i(\bm{w})g_j(\bm{w})\right\}Y^{\text{obs}}_i Y^{\text{obs}}_j  \right], \label{eq_A_impute_hat1}
\end{align}
where $\bm{g}(\bm{w}) = (g_1(\bm{w}),...,g_N(\bm{w}))^\top$, where $g_i(\bm{w}) = l_i(\bm{w})  - \frac{4r(\bm{w}) - N}{N}l_i(\bm{W})$. Thus, $\textcolor{black}{\psi_d}(\hat{\bm{c}})$ is a quadratic form in $\bm{Y}^{\text{obs}}$. We now compare the coefficient of this quadratic form to that corresponding to Neyman's estimator in Equation \ref{eq_A_impute_neyman}.

To this end, we first compute the following sums
\begin{align}
    \sum_{\bm{w}}r(\bm{w}) & = \sum_{k = 0}^{N/2} k \#\left\{\bm{w}: \sum_{i=1}^{N/2}w_i = k\right\} \nonumber \\
    & = \sum_{k=0}^{N/2} k \binom{N/2}{k} \binom{N/2}{N/2 - k} \nonumber \\
    & = \binom{N}{N/2} \frac{N}{4},
\end{align}
\begin{align}
    \sum_{\bm{w}}r^2(\bm{w}) & = \sum_{k=0}^{N/2} k^2 \binom{N/2}{k} \binom{N/2}{N/2 - k} \nonumber \\
    & = \binom{N}{N/2}\frac{N^2}{16}\frac{N}{N-1}.
\end{align}
Next, for unit $i \in \{1,2,...,N/2\}$, 
\begin{align}
    \sum_{\bm{w}}l_i(\bm{w})r(\bm{w}) &= \sum_{\bm{w}:w_i = 1} r(\bm{w}) - \sum_{\bm{w}:w_i = 0} r(\bm{w}) \nonumber \\
    & = \sum_{k=1}^{N/2}k \binom{N/2-1}{k-1} \binom{N/2}{N/2-k} - \sum_{k=1}^{N/2}k \binom{N/2-1}{k} \binom{N/2}{N/2-k} \nonumber \\
    & = \binom{N}{N/2}\frac{N}{4(N-1)}.
\end{align}
Likewise, for unit $i \in \{N/2+1,...,N\}$, 
\begin{align}
    \sum_{\bm{w}}l_i(\bm{w})r(\bm{w}) &= \sum_{\bm{w}:w_i = 1} r(\bm{w}) - \sum_{\bm{w}:w_i = 0} r(\bm{w}) \nonumber \\
    & = \sum_{k=1}^{N/2}k \binom{N/2}{k} \binom{N/2-1}{N/2-1-k} - \sum_{k=1}^{N/2}k \binom{N/2}{k} \binom{N/2-1}{N/2-k} \nonumber \\
    & = -\binom{N}{N/2}\frac{N}{4(N-1)}.
\end{align}
Finally, for $i \in \{1,2,...,N/2\}$ and $j \in \{N/2+1,...,N\}$,
\begin{align}
    \sum_{\bm{w}}l_i(\bm{w})l_j(\bm{w}) &= \sum_{\bm{w}:w_i = w_j}1  - \sum_{\bm{w}:w_i \neq w_j}1 \nonumber\\
    & = \left\{\binom{N-2}{N/2-2} + \binom{N-2}{N/2-2} \right\} - \left\{ \binom{N-2}{N/2-1} + \binom{N-2}{N/2-1} \right\} \nonumber \\
    & = -\frac{1}{N-1}\binom{N}{N/2}.
\end{align}
Now, the coefficient of $(Y^{\text{obs}}_i)^2$ in Equation \ref{eq_A_impute_hat1}, for $i \in \{1,...,N/2\}$ is given by,
\begin{align}
    &\frac{4}{N^2}\frac{1}{\binom{N}{N/2}} \sum_{\bm{w}}g^2_i(\bm{w}) \nonumber \\
    & = \frac{4}{N^2}\frac{1}{\binom{N}{N/2}} \sum_{\bm{w}}\left\{l_i(\bm{w}) - \frac{4r(\bm{w})-N}{N}   \right\}^2 \nonumber\\
    & = \frac{4}{N^2}\frac{1}{\binom{N}{N/2}}\left\{2\binom{N}{N/2} + \frac{16}{N^2}\sum_{\bm{w}}r^2(\bm{w}) - \frac{8}{N}\sum_{\bm{w}}r(\bm{w}) - \frac{8}{N}\sum_{\bm{w}}l_i(\bm{w})r(\bm{w}) + 2\sum_{\bm{w}}l_i(\bm{w})\right\} \nonumber \\
    & = \frac{4}{N^2}\frac{1}{\binom{N}{N/2}}\binom{N}{N/2}\frac{N-2}{N-1}.
\end{align}
Similarly, the coefficient of $(Y^{\text{obs}}_i)^2$ in Equation \ref{eq_A_impute_hat1}, for $i \in \{N/2+1,...,N\}$ is given by,
\begin{align}
   & \frac{4}{N^2}\frac{1}{\binom{N}{N/2}} \sum_{\bm{w}}g^2_i(\bm{w}) \nonumber\\
   & = \frac{4}{N^2}\frac{1}{\binom{N}{N/2}}\sum_{\bm{w}}\left\{l_i(\bm{w}) + \frac{4r(\bm{w})-N}{N}   \right\}^2 \nonumber\\
    & = \frac{4}{N^2}\frac{1}{\binom{N}{N/2}}\left\{2\binom{N}{N/2} + \frac{16}{N^2}\sum_{\bm{w}}r^2(\bm{w}) - \frac{8}{N}\sum_{\bm{w}}r(\bm{w}) + \frac{8}{N}\sum_{\bm{w}}l_i(\bm{w})r(\bm{w}) - 2\sum_{\bm{w}}l_i(\bm{w})\right\} \nonumber \\
    & =  \frac{4}{N^2}\frac{1}{\binom{N}{N/2}} \binom{N}{N/2}\frac{N-2}{N-1} = \frac{4}{N^2}\frac{N-2}{N-1}.
\end{align}
Therefore, the coefficient of $(Y^{\text{obs}}_i)^2$ in $\textcolor{black}{\psi_d}(\hat{\bm{c}})$ is $\frac{N-2}{N-1}$ that of $\hat{V}_{\text{Neyman}}$.

Next, for $i,j \in \{1,2,...,N/2\}$, we consider the coefficient of $Y^{\text{obs}}_iY^{\text{obs}}_j$ in $\textcolor{black}{\psi_d}(\hat{\bm{c}})$.
\begin{align}
 &\frac{4}{N^2}\frac{1}{\binom{N}{N/2}}\sum_{\bm{w}} g_i(\bm{w})g_j(\bm{w}) \nonumber \\
 & = \frac{4}{N^2}\frac{1}{\binom{N}{N/2}}\sum_{\bm{w}}\left\{l_i(\bm{w}) - \frac{4r(\bm{w})-N}{N}   \right\}\left\{l_j(\bm{w}) - \frac{4r(\bm{w})-N}{N}   \right\} \nonumber \\
 & = \left[\sum_{\bm{w}}l_i(\bm{w})l_j(\bm{w}) + \frac{1}{N^2}\sum_{\bm{w}}\{16r^2(\bm{w}) + N^2 - 8Nr(\bm{w})\} - \frac{4}{N}\left\{\sum_{\bm{w}}r(\bm{w})l_i(\bm{w}) + \sum_{\bm{w}}r(\bm{w})l_j(\bm{w})\right\} \right. \nonumber \\
 &\quad + \left.\sum_{\bm{w}}\{l_i(\bm{w}) + l_j(\bm{w})\}\right]\frac{4}{N^2}\frac{1}{\binom{N}{N/2}} \nonumber \\
 & = \frac{4}{N^2}\frac{1}{\binom{N}{N/2}} \frac{(-2)}{N-1}\binom{N}{N/2} = -\frac{8}{N^2(N-2)}\frac{N-2}{N-1}.
\end{align}
Following similar steps, we can show that, for $i,j \in \{N/2+1,...,N\}$, the coefficient of $Y^{\text{obs}}_iY^{\text{obs}}_j$ is,
\begin{align}
    &\frac{4}{N^2}\frac{1}{\binom{N}{N/2}}\sum_{\bm{w}} g_i(\bm{w})g_j(\bm{w}) \nonumber \\
 & = \frac{4}{N^2}\frac{1}{\binom{N}{N/2}}\sum_{\bm{w}}\left\{l_i(\bm{w}) + \frac{4r(\bm{w})-N}{N}   \right\}\left\{l_j(\bm{w}) + \frac{4r(\bm{w})-N}{N}   \right\} \nonumber \\
 &= -\frac{8}{N^2(N-2)}\frac{N-2}{N-1}.
\end{align}
Thus, for $i,j$ with $W_i = W_j$, the coefficient of $Y^{\text{obs}}_iY^{\text{obs}}_j$ in $\textcolor{black}{\psi_d}(\hat{\bm{c}})$ is $\frac{N-2}{N-1}$ times that of $\hat{V}_{\text{Neyman}}$.

Finally, for $i \in \{1,2,...,N/2\}$ and $j \in \{N/2+1,...,N\}$, we consider the coefficient of $Y^{\text{obs}}_iY^{\text{obs}}_j$ in $\textcolor{black}{\psi_d}(\hat{\bm{c}})$.
\begin{align}
 &\frac{4}{N^2}\frac{1}{\binom{N}{N/2}}\sum_{\bm{w}} g_i(\bm{w})g_j(\bm{w}) \nonumber \\
 & = \frac{4}{N^2}\frac{1}{\binom{N}{N/2}}\sum_{\bm{w}}\left\{l_i(\bm{w}) - \frac{4r(\bm{w})-N}{N}   \right\}\left\{l_j(\bm{w}) + \frac{4r(\bm{w})-N}{N}   \right\} \nonumber \\
 & = \left[\sum_{\bm{w}}l_i(\bm{w})l_j(\bm{w}) + \frac{1}{N^2}\sum_{\bm{w}}\{16r^2(\bm{w}) + N^2 - 8Nr(\bm{w})\} + \frac{4}{N}\left\{\sum_{\bm{w}}r(\bm{w})l_i(\bm{w}) - \sum_{\bm{w}}r(\bm{w})l_j(\bm{w})\right\} \right. \nonumber \\
 &\quad - \left.\sum_{\bm{w}}l_i(\bm{w}) + \sum_{\bm{w}}l_j(\bm{w})\right]\frac{4}{N^2}\frac{1}{\binom{N}{N/2}} \nonumber \\
 & = \frac{4}{N^2}\frac{1}{\binom{N}{N/2}} \times 0 = 0.
\end{align}
Therefore, it follows that, the coefficients in the quadratic form corresponding to $\textcolor{black}{\psi_d}(\hat{\bm{c}})$ is $\frac{N-2}{N-1}$ times those corresponding to $\hat{V}_{\text{Neyman}}$. Thus, we have,
\begin{align}
    \textcolor{black}{\psi_d}(\hat{\bm{c}}) = \frac{N-2}{N-1}\hat{V}_{\text{Neyman}}.
\end{align}
This completes the proof.
\qed

\subsection{Proof of Proposition \ref{prop_imp_tauhat_asymp}}

Using Theorem \ref{thm_impute_general}, we can write
$$\textcolor{black}{\psi_d}(\hat{\bm{c}}) = \Var_d(\hat{\tau}) + A_1 + A_2,$$
where
$$A_1  = {(\hat{\tau} - \tau)^2}\sum_{\bm{w}}p_{\bm{w}}\left(\sum_{i:w_i = 1}\frac{W_i}{N/2} - \sum_{i:w_i = 0}\frac{W_i}{N/2}\right)^2,$$
and 
\begin{align*}
  A_2 & =   -2{(\hat{\tau} - \tau)}\sum_{\bm{w}}p_{\bm{w}} \left(\sum_{i:w_i = 1}\frac{Y_i(0)}{N/2} - \sum_{i:w_i = 0}\frac{Y_i(0)}{N/2} \right)\left(\sum_{i:w_i = 1}\frac{W_i}{N/2} - \sum_{i:w_i = 0}\frac{W_i}{N/2}\right).
\end{align*}
Now, since $\sum_{\bm{w}}p_{\bm{w}}\left(\sum_{i:w_i = 1}\frac{W_i}{N/2} - \sum_{i:w_i = 0}\frac{W_i}{N/2}\right)^2 \leq 1$, $0\leq A_1 \leq (\hat{\tau} - \tau)^2 = o_P(1)$. Moreover, since the potential outcomes are bounded, $A_1 \leq C$, where $C>0$ is a constant. Therefore, using the dominated convergence theorem, we have $\E_d(A_1) = o(1).$

Similarly, by Cauchy-Schwarz inequality, 
\begin{align}
|A_2| &\leq 2|\hat{\tau} - \tau|\sqrt{\sum_{\bm{w}}p_{\bm{w}}\left(\sum_{i:w_i = 1}\frac{W_i}{N/2} - \sum_{i:w_i = 0}\frac{W_i}{N/2}\right)^2 \sum_{\bm{w}}p_{\bm{w}}\left(\sum_{i:w_i = 1}\frac{Y_i(0)}{N/2} - \sum_{i:w_i = 0}\frac{Y_i(0)}{N/2}\right)^2} \nonumber \\
& \leq 2C'|\hat{\tau} - \tau|,
\end{align}
where $C'>0$ is a constant. The last inequality holds since $\sum_{\bm{w}}p_{\bm{w}}\left(\sum_{i:w_i = 1}\frac{W_i}{N/2} - \sum_{i:w_i = 0}\frac{W_i}{N/2}\right)^2 \leq 1$ and the potential outcomes are bounded. Using consistency of $\hat{\tau}$ we have, $A_2 = o_P(1)$. Moreover, since the potential outcomes are bounded, $|A_2|\leq C''$ for some constant $C''>0$. Using the dominated convergence theorem once again, we get $\E_d(A_2) = o(1)$. This completes the proof.
\qed

\subsection{Proof of Proposition \ref{prop_bias}}
\begin{equation}
  \hat{c}_i =
    \begin{cases}
      \frac{1-\pi_i}{\pi_i} Y_i(1) - (1-\pi_i)\gamma_i & \text{if $W_i = 1$}\\
    \frac{\pi_i}{1-\pi_i} Y_i(0) + \pi_i\gamma_i & \text{if $W_i = 0$},
    \end{cases}  
\end{equation}
So, 
\begin{align}
\E(\hat{c}_i) &= \pi_i \left\{\frac{1-\pi_i}{\pi_i}Y_i(1) - (1-\pi_i)\E(\gamma_i|W_i = 1) \right\} + (1-\pi_i) \left\{\frac{\pi_i}{1-\pi_i}Y_i(0) + \pi_i\E(\gamma_i|W_i = 0) \right\}. \nonumber\\
\implies \E(\hat{c}_i) - c_i &= \pi_i(1-\pi_i)\{\E_d(\gamma_i|W_i = 0) - \E_d(\gamma_i|W_i = 1)\}.
\end{align}

\qed

\subsection{Proof of Proposition \ref{prop_jack_general2}}
{\color{black}
\begin{align}
\E(\hat{\theta}_{(-i)}|W_i = 1) &= \E\left\{ \frac{1}{|\mathcal{S}_1(i)|}\sum_{j\in \mathcal{S}_1(i)}\frac{W_jY^{\text{obs}}_j(1-\pi_j)}{\tilde{\pi}_j \pi_j} - \frac{1}{|\mathcal{S}_0(i)|}\sum_{j\in\mathcal{S}_0(i)}\frac{(1-W_j)Y^{\text{obs}}_j {\pi}_j}{(1-\tilde{\pi}_j)(1-\pi_j)} \Bigg|W_i = 1 \right\} \nonumber\\
& = \frac{1}{|\mathcal{S}_1(i)|}\sum_{j\in \mathcal{S}_1(i)}\frac{\E_d(W_j|W_i = 1)Y_j(1)(1-\pi_j)}{\Pr_d(W_j = 1|W_i = 1) \pi_j} - \frac{1}{|\mathcal{S}_0(i)|}\sum_{j\in \mathcal{S}_0(i)}\frac{\E_d(1-W_j|W_i=1)Y_j(0) {\pi}_j}{\Pr_d(W_j = 0|W_i = 1)(1-\pi_j)} \nonumber\\
& = \frac{1}{|\mathcal{S}_1(i)|}\sum_{j \in \mathcal{S}_1(i)}\frac{1-\pi_j}{\pi_j}Y_j(1) - \frac{1}{|\mathcal{S}_0(i)|}\sum_{j \in \mathcal{S}_0(i)}\frac{\pi_j}{1-\pi_j}Y_j(0).
\end{align}
Following similar steps, we can show that $\E(\hat{\theta}_{(-i)}|W_i = 0) = \frac{1}{|\mathcal{S}_1(i)|}\sum_{j \in \mathcal{S}_1(i)}\frac{1-\pi_j}{\pi_j}Y_j(1) - \frac{1}{|\mathcal{S}_0(i)|}\sum_{j \in \mathcal{S}_0(i)}\frac{\pi_j}{1-\pi_j}Y_j(0)$. 
By Proposition \ref{prop_bias}, it follows that $\hat{\bm{c}}$ is unbiased for $\bm{c}$. Finally, by Jensen's inequality, we have $\E_d\{{\psi_d}(\hat{\bm{c}})\} \geq {\psi_d}\{\E_d(\hat{\bm{c}})\} = \Var_d(\hat{\tau}).$
This completes the proof.
}
\qed

\subsection{Proof of Theorem \ref{thm_jack2_bias}}

In this case, the estimated $\hat{c}_i$ can be written as 
\begin{equation}
    \hat{c}_i = Y_i(0) + W_i \tau + \hat{\theta}_{(-i)}(0.5 -W_i).
\end{equation}
Moreover, in this case, 
\begin{equation}
     \hat{\theta}_{(-i)} =
    \begin{cases}
      \frac{1}{N/2-1}\sum_{j \neq i} W_jY_j(1) - \frac{1}{N/2}\sum_{j \neq i} (1-W_j)Y_j(0)  & \text{if $W_i = 1$}\\
    \frac{1}{N/2}\sum_{j \neq i} W_jY_j(1) - \frac{1}{N/2 - 1}\sum_{j \neq i} (1-W_j)Y_j(0) & \text{if $W_i = 0$}.
    \end{cases}  
\end{equation}
It follows that,
\begin{equation}
     \hat{\theta}_{(-i)} =
    \begin{cases}
      \hat{\tau} + \frac{1}{N/2(N/2-1)}\sum_{j \neq i} W_jY_j(1) - \frac{1}{N/2-1}W_i Y_i(1) & \text{if $W_i = 1$}\\
    \hat{\tau} - \frac{1}{N/2(N/2-1)}\sum_{j \neq i} (1-W_j)Y_j(0) + \frac{1}{N/2}(1-W_i) Y_i(0) & \text{if $W_i = 0$}.
    \end{cases}  
\end{equation}
Let $\bar{Y}(0)_t$, $\bar{Y}(0)_c$, and $\bar{Y}(0)$ be the means of $Y_i(0)$ in the treatment, control, and the overall sample, respectively. Under homogeneity, we can write,
\begin{equation}
     \hat{\theta}_{(-i)} = \hat{\tau} + \frac{2}{N-2}W_i 2(\bar{Y}(0) - Y_i(0)) +\frac{2}{N-2}(Y_i(0) - \bar{Y}(0)_c).
\end{equation}
Now,
\begin{align}
    \hat{V}_{\text{Jack}}  &= \frac{4}{N^2}\sum_{\tilde{\bm{w}}\in \mathcal{W}}\frac{1}{\binom{N}{N/2}}\left(\sum_{i:\tilde{w}_i = 1}\hat{c}_i  - \sum_{i:\tilde{w}_i = 0}\hat{c}_i\right)^2.
\end{align}
Substituting the expressions of $\hat{c}_i$ (and of $\hat{\theta}_{(-i)}$), we get
\begin{align}
    &\hat{V}_{\text{Jack}} \nonumber\\
    &= \frac{4}{N^2}\frac{(N-1)^2}{(N-2)^2}\sum_{\tilde{\bm{w}}\in \mathcal{W}}\frac{1}{\binom{N}{N/2}}\left\{\left(\sum_{i:\tilde{w}_i = 1}Y_i(0)  - \sum_{i:\tilde{w}_i = 0}Y_i(0)\right) - (\hat{\tau} - \tau)\left(\sum_{i:\tilde{w}_i = 1}W_i  - \sum_{i:\tilde{w}_i = 0}W_i\right) \right\}^2 \nonumber\\
& =  \frac{4}{N^2}\frac{(N-1)^2}{(N-2)^2}\frac{1}{\binom{N}{N/2}}\left[ \sum_{\tilde{\bm{w}}\in \mathcal{W}}\left(\sum_{i:\tilde{w}_i = 1}Y_i(0)  - \sum_{i:\tilde{w}_i = 0}Y_i(0)\right)^2 \right.\nonumber\\
& \hspace{1.5in} - 2\left(\sum_{i:\tilde{w}_i = 1}Y_i(0)  - \sum_{i:\tilde{w}_i = 0}Y_i(0)\right)(\hat{\tau} - \tau)\left(\sum_{i:\tilde{w}_i = 1}W_i  - \sum_{i:\tilde{w}_i = 0}W_i\right) \nonumber \\
& \hspace{1.5in} \left. +  (\hat{\tau} - \tau)^2\left(\sum_{i:\tilde{w}_i = 1}W_i  - \sum_{i:\tilde{w}_i = 0}W_i\right)^2 \right]. \label{eq_A_5.9_0} 
\end{align}
We now derive the expectations of the terms on the right-hand side.
\begin{align}
&    \E\left\{(\hat{\tau} - \tau)\left(\sum_{i:\tilde{w}_i = 1}W_i  - \sum_{i:\tilde{w}_i = 0}W_i\right) \right\} \nonumber\\
& = \E\left\{\frac{1}{N/2}\left(\sum_{j}W_jY_j(0) - \sum_{j}(1-W_j)Y_j(0)\right)\left(\sum_{i:\tilde{w}_i = 1}W_i  - \sum_{i:\tilde{w}_i = 0}W_i\right) \right\} \nonumber\\
& = \frac{4}{N}\E\left\{\left(\sum_{j:\tilde{w}_j = 1}W_jY_j(0) + \sum_{j:\tilde{w}_j = 0}W_jY_j(0) \right)\left(\sum_{i:\tilde{w}_i = 1}W_i  - \sum_{i:\tilde{w}_i = 0}W_i\right) \right\}. \label{eq_A_5.9_1}
\end{align}
Under a CRD, and for $i \neq j$, $\E(W_iW_j) = \frac{N/2(N/2-1)}{N(N-1)}$. Thus, expanding the product in Equation \ref{eq_A_5.9_1}, we get
\begin{align}
     \E\left\{(\hat{\tau} - \tau)\left(\sum_{i:\tilde{w}_i = 1}W_i  - \sum_{i:\tilde{w}_i = 0}W_i\right) \right\} = \frac{1}{N-1}\left( \sum_{i:\tilde{w}_i = 1}Y_i(0)  - \sum_{i:\tilde{w}_i = 0}Y_i(0)\right).
\end{align}
Next, we consider
\begin{align}
 & \E\left\{(\hat{\tau} - \tau)^2\frac{1}{\binom{N}{N/2}}\sum_{\tilde{\bm{w}}}\left(\sum_{i:\tilde{w}_i = 1}W_i  - \sum_{i:\tilde{w}_i = 0}W_i\right)^2 \right\} \nonumber\\
  & = \E\left\{(\hat{\tau} - \tau)^2\frac{1}{\binom{N}{N/2}}\sum_{\tilde{\bm{w}}}\left(\sum_{i=1}^{N}2\tilde{w}_iW_i  - N/2\right)^2 \right\} \nonumber\\
  & = \E\left\{(\hat{\tau} - \tau)^2 \left(\sum_{i=1}^{N}2\tilde{W}_iW_i  - N/2\right)^2 \right\},
\end{align}
where $\tilde{\bm{W}} = (\tilde{W}_1,...,\tilde{W}_N)^\top$ is an independent and identical copy of $\bm{W}$. Using the law of iterated expectations,
\begin{align}
&    \E\left\{(\hat{\tau} - \tau)^2 \left(\sum_{i=1}^{N}2\tilde{W}_iW_i  - N/2\right)^2 \right\} \nonumber\\
& = \E\left[(\hat{\tau} - \tau)^2 \E\left\{\left(\sum_{i=1}^{N}2\tilde{W}_iW_i  - N/2\right)^2\Bigg|\bm{W} \right\}\right] \nonumber\\
& = \E\left[ (\hat{\tau} - \tau)^2 \left\{\frac{N}{2} + 2\frac{N-2}{4(N-1)}\mathop{\sum\sum}_{i<i'}(2W_i-1)(2W_{i'}-1) \right\}  \right] \nonumber\\
& = \E\left[ (\hat{\tau} - \tau)^2 \left\{ \frac{N}{2} - \frac{N(N-2)}{4(N-1)} + \frac{(N-2)}{4(N-1)} \left(\sum_{i=1}^{N}(2W_i-1)\right)^2 \right\}  \right] \nonumber\\
& = \frac{N^2}{4(N-1)} \E(\hat{\tau} - \tau)^2 =  \frac{N^2}{4(N-1)} \Var(\hat{\tau}).
\end{align}
Therefore, Equation \ref{eq_A_5.9_0} implies,
\begin{align}
& \E(\hat{V}_{\text{Jack}}) \nonumber \\
& = \frac{(N-1)^2}{(N-2)^2}\left\{\Var(\hat{\tau}) -2 \frac{1}{N-1}\Var(\hat{\tau}) + \frac{4}{N^2}\frac{N^2}{4(N-1)} \Var(\hat{\tau})  \right\} \nonumber\\
& = \Var(\hat{\tau}) \frac{N-1}{N-2},
\end{align}
which completes the proof. 


\section{\textcolor{black}{Additional simulation results}}
\label{sec_simulation}

\subsection{\textcolor{black}{Completely randomized designs}}

In this section, we evaluate the performance of the direct imputation approach for different versions of the variance estimator using a simulation study. To this end, we consider six scenarios, each corresponding to a combination of the design parameters and the potential outcomes. In particular,
\begin{itemize}
    \item Scenarios 1 and 2: CRD with $N = 6$, $N_t = N_c = 3$
    \item Scenario 3: CRD with $N = 6$, $N_t = 4$, $N_c = 2$
    \item Scenarios 4 and 5: CRD with $N = 8$, $N_t = N_c = 4$
    \item Scenario 6, CRD with $N = 8$, $N_t = 5$, $N_c = 3$.
\end{itemize}
In each scenario, the potential outcomes under control are generated independently from a $\text{Uniform}(0,10)$. For scenarios 1,3,4, and 6, the unit-level treatment effects are homogeneous, with the common value drawn from a $\text{Uniform}(-5,5)$ distribution. Conversely, for scenarios 2 and 5, treatment effects are allowed to vary across units, with values independently drawn from a $\text{Uniform}(-5,5)$ distribution.

We consider four different choices of $\gamma_i$, namely $0$, $\hat{\tau}$, $\hat{\tau}_{(-i)}$ (as in Equation \ref{eq_jack2}), and $\hat{\theta}_{(-i)}$ (as in Equation \ref{eq_jack3}). Under each scenario and for each choice of $\gamma_i$, we compute the relative bias of the resulting variance estimator $\hat{V}$, defined as $\frac{\E_d(\hat{V}) - \Var_d(\hat{\tau})}{\Var_d(\hat{\tau})}$. This process is repeated 100 times, each time independently generating the potential outcomes according to the specified data-generating process. 

Figure \ref{fig:jack1} shows the distribution of the relative bias of each estimator across the six scenarios.

\begin{figure}[!ht]
    \centering
    \includegraphics[scale =0.65]{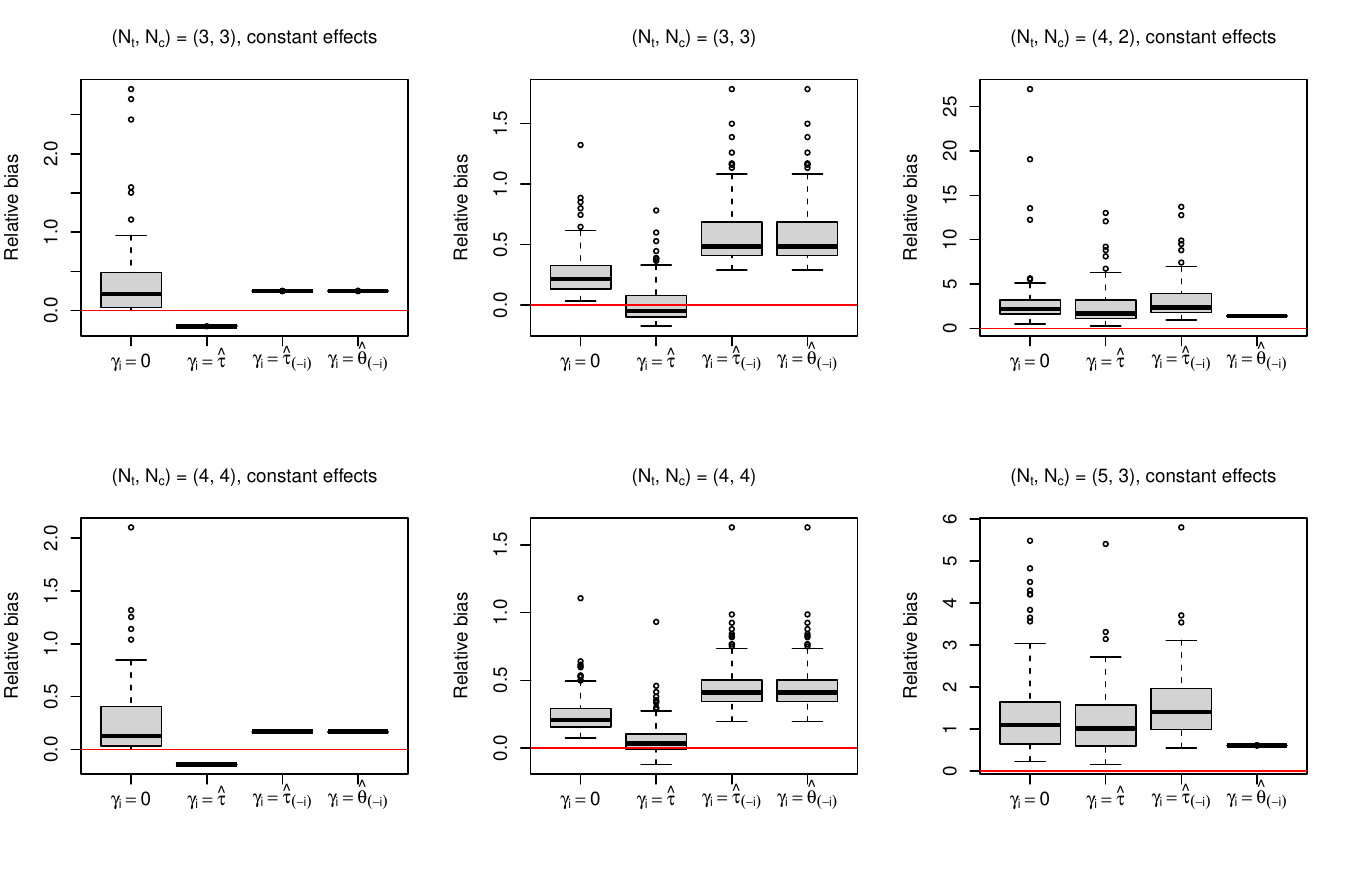}
    \caption{Relative bias of direct imputation based variance estimators for different choices of $\gamma_i$ under complete randomization with $N_t$ treated and $N_c$ control units}
    \label{fig:jack1}
\end{figure}
Figure \ref{fig:jack1} shows that, the direct imputation estimators with $\gamma_i = 0$, $\hat{\tau}_{(-i)}$, or $\hat{\theta}_{(-i)}$ exhibit non-negative relative biases across all scenarios since they are conservative by construction. The exception is the estimator with $\gamma_i = \hat{\theta}$, which, for completely randomized designs with equal group sizes and under homogeneity, is known to be anti-conservative. This phenomenon is evident in the negative relative biases observed in scenarios 1 and 4.
Moreover, the jackknifed variance estimators with $\gamma_i = \hat{\tau}_{(-i)}$ and $\gamma_i = \hat{\theta}_{(-i)}$ exhibit identical distributions of relative bias when $N_t = N_c$. This observation is not surprising, as the two estimators are equivalent when $\pi_i = 0.5$. For fixed $N_t$ and $N_c$ and under treatment effect homogeneity, the relative bias corresponding to $\gamma_i = \hat{\theta}_{(-i)}$ is constant, i.e., it does not depend on the potential outcomes. This observation aligns with Theorem \ref{thm_jack2_bias}. Overall, the jackknifed variance estimator with $\gamma_i = \hat{\theta}_{(-i)}$ performs reasonably well across scenarios when treatment effects are homogeneous.

In Figure \ref{fig:jack2}, we perform similar comparisons under CRDs with larger group sizes, e.g., $(N_t, N_c) \in \{(30,30), (15,45)\}$. Here too, the jackknifed variance estimator with $\gamma_i = \hat{\theta}_{(-i)}$ performs well across different scenarios. 

\begin{figure}[!ht]
  \centering
  \begin{subfigure}[b]{0.32\textwidth}
    \includegraphics[width=\textwidth]{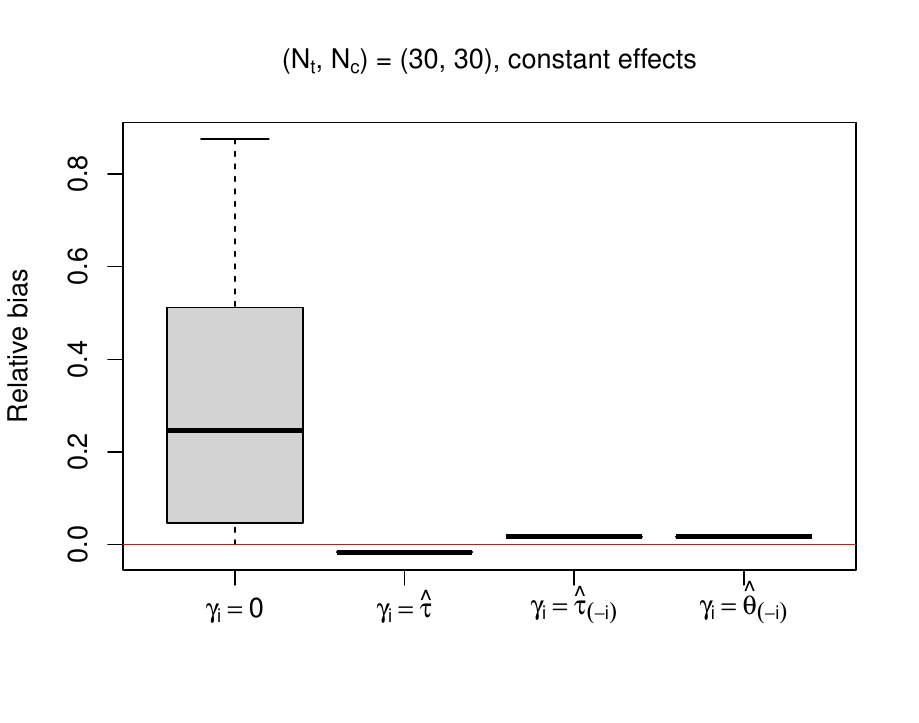}
  \end{subfigure}
  \hfill
  \begin{subfigure}[b]{0.32\textwidth}
    \includegraphics[width=\textwidth]{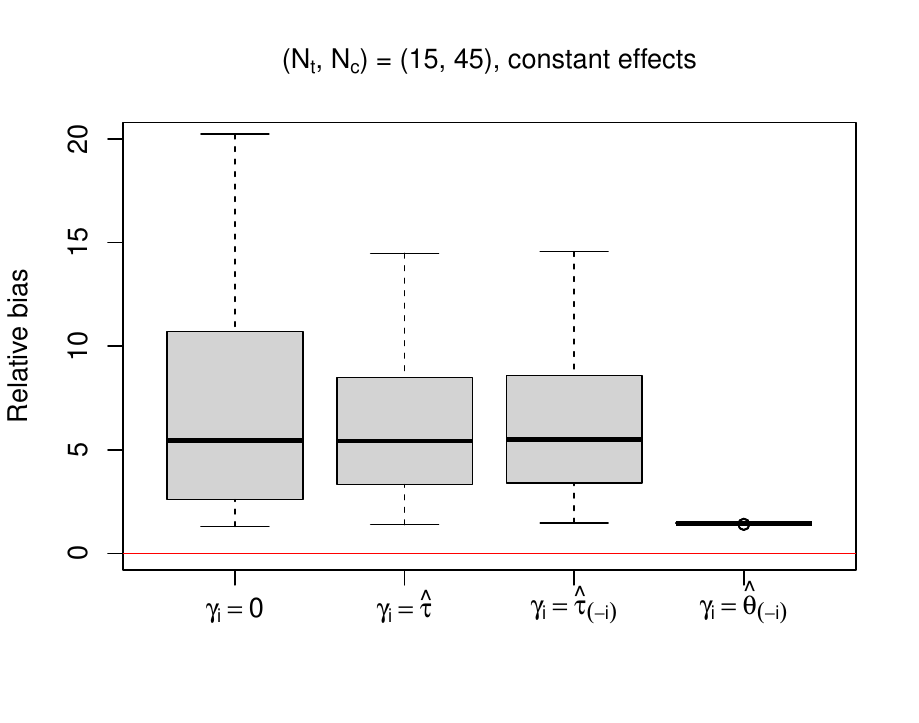}
  \end{subfigure}
  \hfill
  \begin{subfigure}[b]{0.32\textwidth}
    \includegraphics[width=\textwidth]{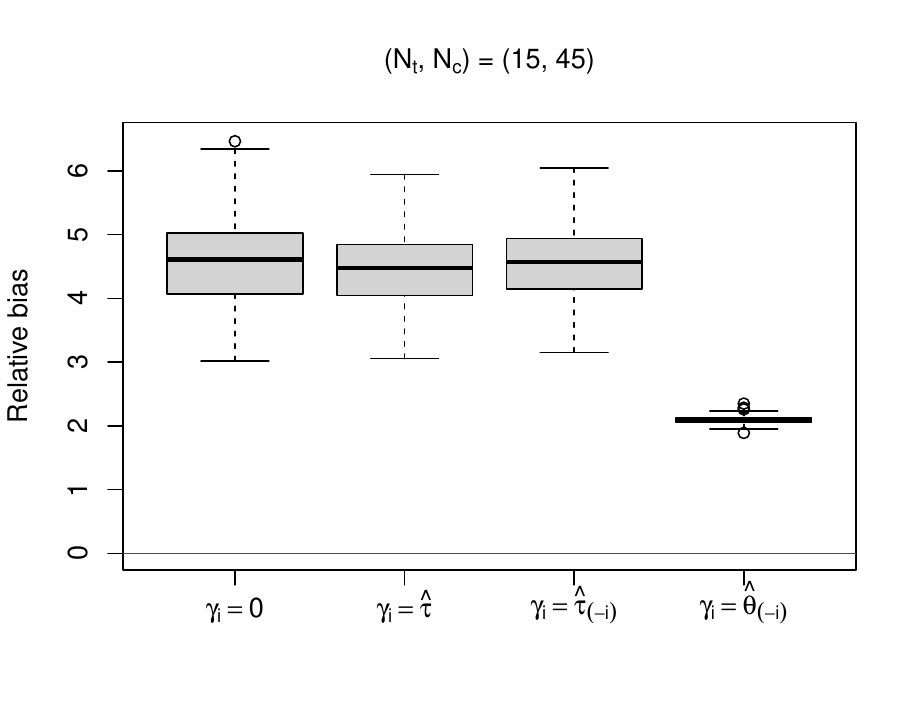}
  \end{subfigure}
  \caption{Relative bias of direct imputation based variance estimators for different choices of $\gamma_i$ under complete randomization with $N_t$ treated and $N_c$ control units.}
  \label{fig:jack2}
\end{figure}

{\color{black}
\subsection{Stratified randomized designs}

In this section, we perform a similar comparison for stratified randomized designs. We consider two designs:
\begin{itemize}
    \item Design 1: $N = 80$, $H = 4$ strata of sizes 20 each, and equal allocation within each stratum.
    \item Design 2: $N = 80$, $H = 4$ strata of sizes 10, 10, 30, 30, and equal allocation within each stratum.
\end{itemize}
Under each design, we consider three different generative models for the potential outcomes: 
 \begin{enumerate}
        \item No effect whatsoever: $Y_i(0) \overset{iid}{\sim} \text{Unif}(0,10)$, $Y_i(1) = Y_i(0)$ for all $i$. 

        \item Constant effect: $Y_i(0) \overset{iid}{\sim} \text{Unif}(0,10)$, $Y_i(1) = Y_i(0) + 5$ for all $i$. 

        \item Heterogeneous effects: $Y_i(0) \overset{iid}{\sim} \text{Unif}(0,10)$, $Y_i(1) = Y_i(0) + \tau_i$ for all $i$, where $\tau_i \overset{iid}{\sim} \text{Unif}(-5,5)$ and are independent of $\bm{Y}(0)$. 
    \end{enumerate}
Figure \ref{fig:app_simu1} and \ref{fig:app_simu2} display the distributions of relative bias and standard deviation of the direct imputation estimator with three different choices of $\gamma_i$, namely $0$, $\hat{\tau}$, and $\hat{\theta}_{(-i)}$, under Designs 1 and 2, respectively. 
Overall, the results show similar patterns as observed in the cases of complete randomization and rerandomization (in Section \ref{sec_imputation_extension}). The imputation estimator with $\gamma_i = 0$ and $\gamma_i = \hat{\theta}_{(-i)}$ are always conservative, while the estimator with $\gamma_i = \hat{\tau}$ may be anti-conservative. Comparatively, $\gamma_i = \hat{\theta}_{(-i)}$ (the jackknife-based imputation estimator) performs well across all scenarios, particularly when the effects are homogeneous, both in terms of relative bias and standard error. 
\begin{figure}[!ht]
    \centering
    \includegraphics[width=0.9\linewidth]{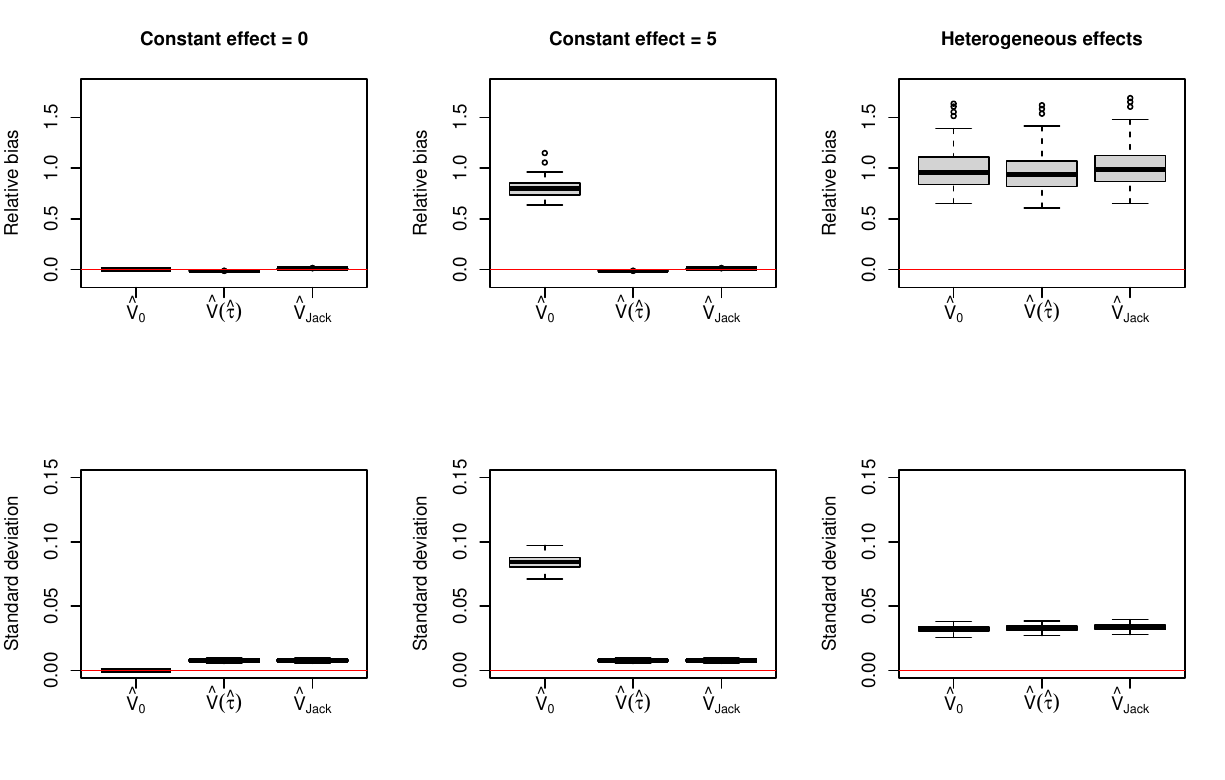}
    \caption{\textcolor{black}{Relative biases and standard errors of variance estimators under Design 1.}}
    \label{fig:app_simu1}
\end{figure}

\begin{figure}[!ht]
    \centering
    \includegraphics[width=0.9\linewidth]{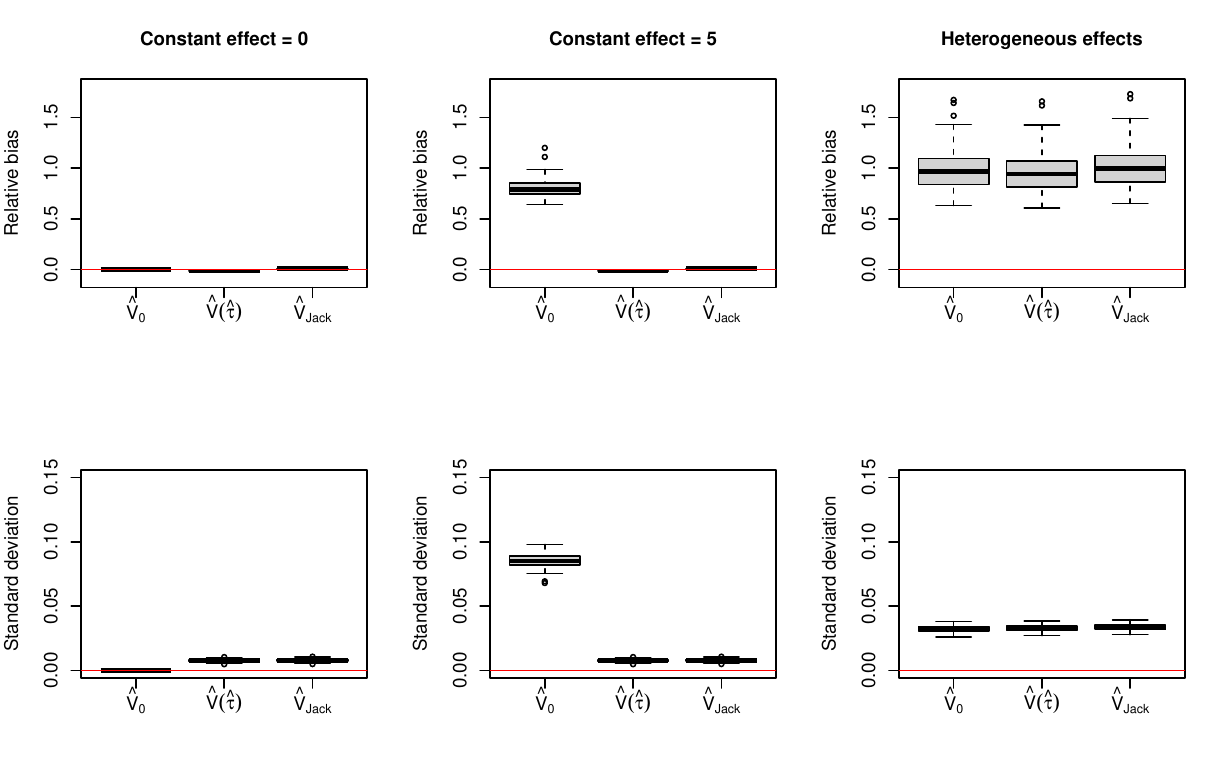}
    \caption{\textcolor{black}{Relative biases and standard errors of variance estimators under Design 2.}}
    \label{fig:app_simu2}
\end{figure}

}

\end{document}